\PassOptionsToPackage{table}{xcolor}
\documentclass[preprint]{jmlr}

\jmlrvolume{}
\jmlryear{}
\jmlrworkshop{}
\jmlrproceedings{}{}

\definecolor{darkgreen}{rgb}{0.0, 0.5, 0.0}

\usepackage{booktabs,longtable,caption}
\usepackage{placeins}

\usepackage{algorithm}
\usepackage{algorithmic}
\usepackage{booktabs}
\usepackage{multirow,amssymb,amsmath,amsfonts}
\usepackage{mathtools}

\usepackage{enumitem} 

\renewcommand{\P}[1]{\mathbb{P}\left[#1\right]}

\newcommand{\I}[1]{\mathbb{I}\left[#1\right]}

\newcommand\indep{\protect\mathpalette{\protect\independenT}{\perp}}\def\independenT#1#2{\mathrel{\rlap{$#1#2$}\mkern2mu{#1#2}}}

\firstpageno{1}

\begin{document}

\title{Conformal Survival Bands for Risk Screening under Right-Censoring}
\author{\Name{Matteo Sesia}\Email{sesia@marshall.usc.edu}\\
  \addr{University of Southern California, Los Angeles, California, USA}\\
  \Name{Vladimir Svetnik}\Email{vladimir\_svetnik@merck.com}\\
  \addr{Merck \& Co., Inc., Rahway, New Jersey, USA}
}


\maketitle

\begin{abstract}
We propose a method to quantify uncertainty around individual survival distribution estimates using right-censored data, compatible with any survival model. Unlike classical confidence intervals, the survival bands produced by this method offer predictive rather than population-level inference, making them useful for personalized risk screening. For example, in a low-risk screening scenario, they can be applied to flag patients whose survival band at 12 months lies entirely above 50\%, while ensuring that at least half of flagged individuals will survive past that time on average. Our approach builds on recent advances in conformal inference and integrates ideas from inverse probability of censoring weighting and multiple testing with false discovery rate control. We provide asymptotic guarantees and show promising performance in finite samples with both simulated and real data.

\end{abstract}

\thispagestyle{empty}   


\section{Introduction} \label{sec:intro}

\subsection{Background and Motivation}

Survival analysis focuses on data involving time-to-event outcomes, such as the time until death or disease relapse in medicine, or mechanical failure in engineering. Its defining challenge is \emph{censoring}, which occurs when the event of interest is not observed for all individuals. In the common case of right-censoring, we only know that the event has not occurred up to a certain time, beyond which the individual is no longer followed. For example, a cancer patient may still be alive at their last clinical follow-up, five years after diagnosis, but their true time of death is unknown because no data are available after that point. In this case, the patient’s censoring time is five years and their survival time is unobserved.

In many applications, the goal is to use fitted survival models to generate personalized inferences in the form of \emph{individual survival curves}. These curves estimate, for an individual with specific features, the probability of remaining event-free beyond any future time point. For example, based on a patient’s medical history, a model might predict a 90\% chance of surviving past one year, 75\% past three years, and 40\% past five years. These predictions can be intuitively visualized as a decreasing curve over time. Such personalized survival curves are widely used to guide decisions, such as identifying high-risk patients for early intervention or recognizing low-risk individuals who may safely avoid aggressive treatment.

Traditional approaches to survival analysis rely on statistical models that support uncertainty quantification through confidence intervals and hypothesis testing. These include parametric models, semi-parametric models, and nonparametric estimators such as the Kaplan--Meier (KM) curve \citep{kaplan1958nonparametric}. Parametric models, such as the exponential or Weibull distributions, assume a specific form for how the event probability evolves over time. Semi-parametric models, most notably the Cox proportional hazards model \citep{cox1972regression}, relax this assumption by leaving the baseline distribution unspecified, while still imposing a fixed relationship between covariates and risk. The Kaplan--Meier estimator, by contrast, avoids strong modeling assumptions altogether, but estimates only population-level survival curves and does not incorporate covariate information.

Although classical survival models have seen many successful applications, they can be limiting in modern settings that involve large datasets with rich covariate information. When the relationship between patient features and survival outcomes is complex or poorly understood, strong modeling assumptions—such as proportional hazards or specific parametric distributions—may be difficult to justify. In these situations, it is often more practical to treat classical models as black-box predictors: tools that generate individual survival curves, but whose internal assumptions are not relied on for inference. This perspective also motivates the growing use of more flexible, data-driven approaches—particularly machine learning methods such as random survival forests \citep{Ishwaran2008}, deep neural networks \citep{katzman2018deepsurv}, and gradient boosting \citep{barnwal2022survival}—which can model complex relationships with high-dimensional covariates \citep{spooner2020comparison}.

While black-box models can produce accurate and personalized survival predictions, they typically lack principled methods for uncertainty quantification. Classical statistical models offer confidence intervals and hypothesis tests, but—as discussed above—their guarantees rely on strong distributional assumptions that may not hold in practice. Machine learning models, by contrast, often omit uncertainty estimates altogether; and when such measures are available, they are typically heuristic and lack formal statistical justification. This limitation is particularly concerning in high-stakes applications such as clinical triage or treatment planning, where decisions must rely on calibrated, trustworthy risk estimates. There is therefore a need for widely applicable statistical tools that can provide rigorous, distribution-free uncertainty quantification for black-box survival models.

\subsection{Preview of Our Contributions}

We develop a method to construct principled and interpretable ``uncertainty bands'' around individual survival curves estimated by black-box models, using right-censored data. These bands, which we call \textit{conformal survival bands}, are designed to support clinical screening tasks within a {\em multiple testing} framework and are rigorously calibrated in a useful predictive sense, as explained below.

Here, ``multiple testing'' refers to the fact that our method constructs uncertainty bands for each of several individuals in a test set by processing the entire test set jointly. This is analogous to how the Benjamini–Hochberg procedure processes all $p$-values for multiple hypotheses at once to control the false discovery rate (FDR) \citep{benjamini1995controlling}. The multiple testing nature of our approach will become clearer once we present its details and the theoretical guarantees it provides.

A preview of our method is provided in Figure~\ref{fig:calibration_band_example}, which displays conformal survival bands for four test patients in a simulated dataset. At first glance, these bands may resemble personalized versions of classical confidence intervals—such as those drawn around Kaplan--Meier survival curves—but their interpretation is different. Classical confidence intervals aim to estimate population-level quantities like the true survival probability at a given time, which is typically only possible under strong parametric assumptions or after binning the covariates into a small number of population subgroups. Our bands, by contrast, are fully nonparametric and support individualized inferences without requiring covariate binning.

\begin{figure}[!htb]
    \centering
    \includegraphics[width=\textwidth]{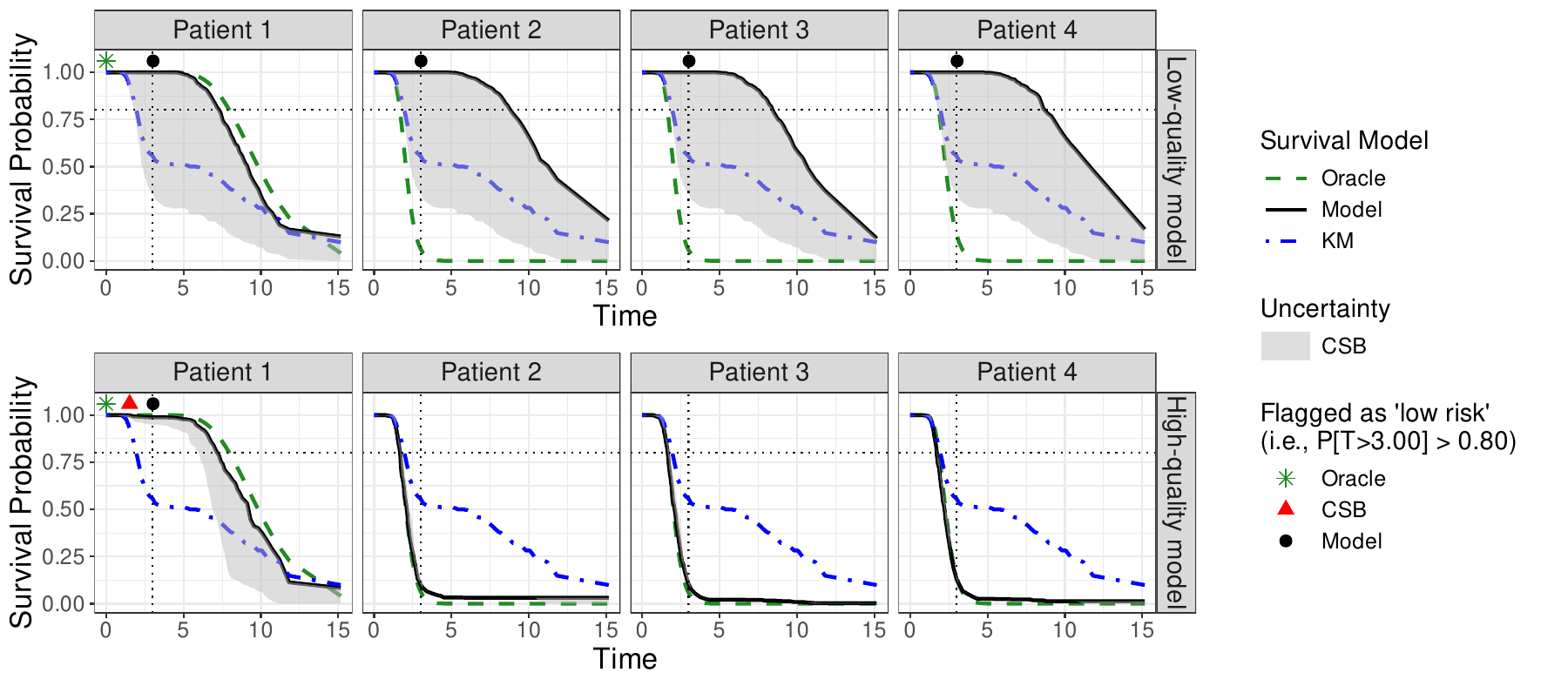}
    \caption{Illustration of the use of conformal survival bands (shaded regions) for screening test patients in simulated censored data.
Solid black curves show survival estimates from either an inaccurate (top) or accurate (bottom) survival forest model. Dashed green curves represent the true survival probabilities. The goal is to identify low-risk patients—those with more than 80\% probability of surviving beyond 3 years (vertical line). A patient is flagged by our method if their entire conformal band lies above the 80\% threshold (horizontal line) at this time, ensuring at least 80\% of flagged individuals survive longer. Flagging decisions are indicated by colored markers: red triangles (our method), green asterisks (oracle), and black circles (estimated model). Kaplan-Meier makes no selections here. Patients 2--4 are mistakenly flagged as low risk by the inaccurate model. With the accurate model, our method can identify truly low-risk patients (e.g., Patient 1), unlike the KM curve.
}
    \label{fig:calibration_band_example}
\end{figure}

Rather than attempting to directly estimate population-level parameters, our bands are calibrated for distribution-free predictive screening. They allow practitioners to identify individuals whose estimated survival probability lies below or above a clinically meaningful threshold, while guaranteeing that among all such flagged individuals, the false discovery rate is approximately controlled in a precise predictive sense. This is why our method is intrinsically a multiple testing approach that processes the entire test set jointly. 

For example, in Figure~\ref{fig:calibration_band_example}, we illustrate a scenario where one wishes to identify ``low-risk'' patients—defined as those with more than an $p=80\%$ chance of surviving beyond time $t=3$ years. A patient is flagged if their conformal band lies entirely above the 80\% line at $t = 3$. Our method \textit{asymptotically} guarantees that, among all flagged individuals, on average at least 80\% will survive beyond this time. This guarantee holds \textit{pointwise} over all fixed choices of the probability threshold $p$ and time horizon $t$, and for ``low-risk'' as well as ``high-risk'' screening, making the method flexible and broadly applicable.

This predictive notion of uncertainty within a multiple testing framework is well-aligned with how survival curves are commonly used in real-world settings: to guide actionable decisions such as prioritizing patients for further testing or early intervention. However, black-box survival models typically provide only point estimates, without any reliable measure of uncertainty or formal calibration guarantees. As in the example shown, screening based on point predictions alone, or even KM estimates, may lead to too many patients incorrectly labeled as low risk—highlighting the practical benefits of our calibrated bands, which offer rigorous control over such errors.

Our approach integrates three key concepts: conformal inference \citep{vovk2005algorithmic,fontana2023conformal}, inverse probability of censoring weighting (IPCW) \citep{robins1992recovery,robins1993information}, and FDR control \citep{benjamini1995controlling}. Concretely, we build on tests for {\em random} null hypotheses stating that $T \geq t$ or $T \leq t$, where $T$ is the (unobserved) survival time of the test individual and $t > 0$ is a user-specified time threshold.
We will compute {\em conformal p-values} to test these hypotheses and show how to use them to construct the conformal survival bands previewed in Figure~\ref{fig:calibration_band_example}.
Although each of the statistical concepts upon which we build already exists on its own, the way we integrate them is original and supported by novel theoretical results.

\subsection{Related Work}

Conformal inference for survival analysis is a very active area of research, with several recent works addressing different inferential goals. Some methods focus on prediction of survival times \citep{bostrom2019predicting,pmlr-v204-bostrom23a,sun2024conformal} or classification into survival/event categories \citep{bostrom2017conformal}, while others aim to refine the point estimates of survival probabilities produced black-box models \citep{qi2024conformalized}. The works most closely related to our own concern the construction of predictive lower bounds for survival times at a fixed confidence level~$\alpha \in (0,1)$, using censored data. These methods seek to deliver distribution-free prediction intervals while correcting for the non-exchangeability \citep{barber2023conformal} introduced by censoring.

This direction was initiated by \citet{candes2023conformalized}, whose method provides valid lower bounds but is often overly conservative in practice. \citet{gui2024conformalized} introduced refinements that improve coverage tightness, yet even their method tends to produce wide intervals at moderate coverage levels (e.g., $\alpha = 0.5$), limiting its utility in settings that require a finer characterization of predictive uncertainty. In addition, both of these methods assume type-I censoring—where censoring times are fully observed for all individuals—an assumption that is quite limiting in practice. More recent contributions by \citet{davidovconformalized}, \citet{sesia2024doubly}, and \citet{farina2025doubly} relax this constraint to accommodate general right-censoring. However, these approaches remain restricted to one-sided predictive inference and do not support threshold-based screening or two-sided uncertainty quantification.

Among the few existing methods that attempt to construct two-sided predictive intervals under censoring are those of \citet{qi2024conformalized}, \citet{qin2024conformal}, and \citet{yi2025survival}.
 The approach of \citet{qi2024conformalized} approaches the problem by `de-censoring'' the data: it estimates latent survival times in the calibration data set using the fitted survival model and then applies standard conformal techniques as if those predictions were observed.
While convenient, this method assumes that the survival model used for de-censoring is reliable. As argued by \citet{sesia2024doubly}, this assumption undermines the primary appeal of conformal inference—namely, its robustness to model misspecification. Indeed, empirical evidence suggests that the method of \citet{qi2024conformalized} may fail to provide valid coverage when the survival model is inaccurate. By contrast, \citet{qin2024conformal} take a more model-agnostic approach, based on resampling under censoring, but their work focuses on predictive intervals for survival times rather than on calibrating survival screening decisions.
Our work is also related to recent, independent research by \citet{yi2025survival}, who construct two-sided predictive intervals for survival times using a method that, like ours, utilizes conformal p-values with IPC weighting. However, their focus is on survival time prediction rather than decision-based screening or the calibration of individual survival curves.

In this paper, we address the problem of estimating uncertainty around individual survival curves to support threshold-based screening tasks. Our approach is based on computing weighted conformal $p$-values for testing hypotheses of the form $H_{\mathrm{lt}}(t; j): T_{n+j} \ge t$ and $H_{\mathrm{rt}}(t; j): T_{n+j} \le t$, where $T_{n+j}$ denotes the (unobserved) survival time of a new individual and $t > 0$ is a fixed time threshold.
While it may be possible in principle to obtain such p-values by inverting predictive intervals for survival times obtained using existing methods, this is practically infeasible for several reasons: (i) most methods only provide lower bounds, (ii) even the best available bounds remain conservative, especially at moderate levels of confidence, and (iii) the inversion procedure would be computationally expensive, as it requires solving nested optimization problems for each time threshold $t$.

Our method integrates and extends two previously distinct lines of work. First, we build on the screening-based conformal inference framework of \citet{jin2023selection}, which develops p-values for testing whether unobserved outcomes exceed user-specified thresholds. \citet{jin2023model} later generalized this approach to allow for importance weighting under covariate shift \citep{tibshirani2019conformal}, but their theory and assumptions are not suited to censoring. Second, we draw on \citet{farina2025doubly}, who introduced IPCW \citep{robins1992recovery,robins1993information} into conformal inference to produce predictive intervals under right-censoring. While their method is tailored to one-sided prediction bounds, we adapt IPCW techniques to develop conformal p-values for our survival threshold hypotheses.

\section{Methods}

\subsection{Problem Setup and Assumptions} \label{sec:method-setup}

We consider a right-censored survival setting based on a sample of $n$ individuals indexed by $[n] := \{1, \ldots, n\}$, drawn i.i.d.\ from an unknown population. For each individual $i \in [n]$, we observe a vector of covariates $X_i \in \mathcal{X} \subseteq \mathbb{R}^d$, along with right-censored survival data: the event indicator $E_i := \mathbb{I}(T_i < C_i)$ and the observed time $\tilde{T}_i := \min(T_i, C_i)$, where $T_i > 0$ is the true survival time and $C_i > 0$ is the censoring time. These $n$ observations form the calibration dataset, denoted by $\mathcal{D}_{\text{cal}} := \{(X_i, \tilde{T}_i, E_i)\}_{i=1}^n$,  which will be used to quantify uncertainty---i.e., to calibrate the predictions of a black-box survival model.

We assume access to two black-box models trained on an independent dataset, which may itself be censored and need not follow the same distribution as the calibration or test sets. The only assumption is that the training data are independent of all other samples, allowing us to treat both models as fixed throughout the analysis. The first model is the \emph{survival model}, $\hat{\mathcal{M}}_T$. This model produces an estimated individual survival function $\hat{S}_T(t \mid x)$, which is intended to approximate the true conditional survival probability $S_T(t \mid x) := \mathbb{P}(T > t \mid X = x)$. The second is an auxiliary \emph{censoring model}, $\hat{\mathcal{M}}_C$, which estimates the conditional survival function of the censoring distribution, $\hat{S}_C(t \mid x)$, an approximation to $S_C(t \mid x) := \mathbb{P}(C > t \mid X = x)$. The role of the censoring model is to reweight the calibration data to correct for the missing information due to censoring, enabling valid uncertainty quantification for the survival model.

In addition to the calibration set, we consider a disjoint test set consisting of $m$ individuals, indexed by $\{n+1, \ldots, n+m\}$, also drawn independently from the same population. For each test individual $j$, we observe only the covariates $X_j \in \mathcal{X}$. Our goal is to use the survival model $\hat{\mathcal{M}}_T$ to estimate the survival curve $\hat{S}_T(t \mid X_j)$ for each $X_j$, and to construct a well-calibrated \emph{conformal survival band} around this curve that reflects uncertainty.

Rather than aiming for our conformal survival bands to provide valid confidence intervals for the true survival function $S_T(t \mid x)$, which would not be a feasible goal without additional assumptions \citep{barber2020distribution}, we focus on a practically useful objective that is more naturally attainable within a conformal inference framework: producing principled and interpretable uncertainty estimates around $\hat{S}_T(t \mid X_j)$ that can support confident predictive screening decisions.
For example, given a survival probability threshold $q \in (0,1)$ and a clinically meaningful time point $t > 0$, our goal may be to identify \emph{high-risk} test individuals whose predicted survival probability $\hat{S}_T(t \mid X_j)$ falls \emph{significantly} below $q$, while guaranteeing that, among all such flagged individuals, the expected proportion who survive beyond $t$ remains controlled below level $q$.
Although inherently predictive in nature, this type of calibration guarantee is intuitive and aligns closely with how physicians and practitioners often interpret the output of survival models in practice: as actionable, patient-specific risk estimates that can support confident decision-making.

Because our method is built on the idea of \emph{conformal p-values}, we first review the key concepts underlying this approach.

\subsection{Preliminaries: Review of Conformal Inference without Censoring} \label{sec:preliminaries}

To build intuition for our method, we begin by reviewing how existing conformal inference techniques for regression can be applied to survival analysis in a simplified setting without censoring. In this case, imagine observing complete data $\{(X_i, T_i)\}_{i=1}^n$, where each $(X_i, T_i)$ is drawn i.i.d.\ from an unknown distribution over $\mathcal{X} \times \mathbb{R}_+$.
Given a new test point with covariates $X_{n+j}$, consider the task of testing the null hypothesis
\begin{align} \label{eq:hyp-lt}
  H_{\mathrm{lt}}(t; j) : T_{n+j} \geq t,
\end{align}
for a fixed time \( t > 0 \). This hypothesis asserts the individual will experience the event after time \( t \); rejecting it provides evidence that the individual is unlikely to survive beyond \( t \).

If complete (uncensored) data are available, this hypothesis can be tested as follows.
Let $\hat{s}_{\mathrm{lt}}(t'; x, t)$ be any {\em monotone increasing} function of $t'$, which may depend on $x \in \mathcal{X}$ and the fixed time $t \in \mathbb{R}_+$ in the definition of the null hypothesis~\eqref{eq:hyp-lt}.
We refer to $\hat{s}_{\mathrm{lt}}$ as the \emph{left-tail} nonconformity scoring function.

An intuitive choice of scoring function satisfying the required monotonicity property is
\begin{align} \label{eq:score-lt}
\hat{s}_{\mathrm{lt}}(t'; x, t) := \left( 1-\hat{S}_T(t' \mid x) \right) \I{t' \geq t},
\end{align}
where $\hat{S}_T(t' \mid x) \in (0,1)$ denotes the conditional survival probability at time $t'$ for an individual with covariates $x$, estimated by the model $\hat{\mathcal{M}}_T$. 
We have found that this scoring function works well in practice, although it is not the only possible option.

Using any suitable monotone increasing scoring function $\hat{s}_{\mathrm{lt}}$, we compute a non-conformity score $\hat{s}_{\mathrm{lt}}(T_i, X_i, t)$ for each calibration point $i \in [n]$ as well as for the test point $X_{n+j}$, in the latter case using the placeholder value $t$ instead of the unknown event time $T_{n+j}$---i.e., defining the test score as $\hat{s}_{\mathrm{lt}}(t, X_{n+j}, t)$. Then, following the approach of \citet{jin2023selection}, we define the left-tail conformal p-value as:
\begin{align} \label{eq:pval-lt-full}
\tilde{\phi}_{\mathrm{lt}}(t; X_{n+j}) := \frac{1 + \sum_{i=1}^n \mathbb{I}\left\{ \hat{s}_{\mathrm{lt}}(T_i; X_i, t) \ge \hat{s}_{\mathrm{lt}}(t; X_{n+j}, t ) \right\}}{n+1}.
\end{align}
For example, in the special case where the scoring function takes the form defined in~\eqref{eq:score-lt}, under the relatively mild assumption that $\hat{S}_T(t \mid X_{n+j}) < 1$ almost surely, then it is easy to verify that the left-tail conformal p-value in~\eqref{eq:pval-lt-full} can be written as:
\begin{align*}
\tilde{\phi}_{\mathrm{lt}}(t; X_{n+j}) := \frac{1 + \sum_{i=1}^n \I{T_i \geq t} \mathbb{I}\left\{ \hat{S}_T(T_i \mid X_i) \leq \hat{S}_T(t \mid X_{n+j}) \right\}}{n+1}.
\end{align*}

The following result, due to \citet{jin2023selection}, highlights the statistical validity of the left-tail p-value defined in~\eqref{eq:pval-lt-full} for testing the null hypothesis~\eqref{eq:hyp-lt}. A formal proof is included in Appendix~\ref{app:proofs} for completeness.
Appendix~\ref{app:proofs} also contains all other proofs.

\begin{proposition}[\citet{jin2023selection}] \label{prop:superuniform}
Assume the calibration data and test point $(X_{n+j}, T_{n+j})$ are exchangeable and, for a fixed $t>0$, the scoring function \( \hat{s}_{\mathrm{lt}}(t'; x, t) \) is monotone increasing in \( t' \) for all $x$. Then, for any \( \alpha \in (0,1) \), the conformal p-value $\tilde{\phi}_{\mathrm{lt}}(t; X_{n+j})$ defined in~\eqref{eq:pval-lt-full}, computed using the full (uncensored) data, satisfies:
\[
\P{\tilde{\phi}_{\mathrm{lt}}(t; X_{n+j}) \le \alpha,\, T_{n+j} \geq t } \le \alpha.
\]
\end{proposition}

This implies $\tilde{\phi}_{\mathrm{lt}}(t; X_{n+j})$ can be interpreted as a p-value for testing the \emph{random} hypothesis~\eqref{eq:hyp-lt}, in the sense that small values provide evidence against the null. Notably, the guarantee in Proposition~\ref{prop:superuniform} does not condition on the null event $\{T_{n+j} \geq t\}$; instead, it controls the \emph{joint} probability that both the null is true and the p-value is small. This marginal formulation differs from the classical frequentist setup, where hypotheses are non-random, but it is sufficient for our purposes. As we will see later, this form of validity is precisely what enables us to obtain valid calibration guarantees for personalized risk screening.

Before turning to censored data, it is helpful to note that the above ideas also apply for testing the complementary \emph{right-tail} null hypothesis
\begin{align} \label{eq:hyp-rt}
  H_{\mathrm{rt}}(t; j) : T_{n+j} \leq t,
\end{align}
which is useful for identifying low-risk individuals. In this case, we consider a right-tail scoring function assumed to be monotone \emph{decreasing} in \( t' \); e.g.,
\begin{align} \label{eq:score-rt}
\hat{s}_{\mathrm{rt}}(t'; x, t) := \hat{S}_T(t' \mid x) \I{t' \leq t},
\end{align}
interpreted as the predicted probability of experiencing the event after time \( t' \) for an individual with covariates $x$.
The corresponding right-tail conformal p-value is
\[
\tilde{\phi}_{\mathrm{rt}}(t; X_{n+j}) := \frac{1 + \sum_{i=1}^n \mathbb{I}\left\{ \hat{s}_{\mathrm{rt}}(T_i; X_i, t) \ge \hat{s}_{\mathrm{rt}}(t; X_{n+j}, t ) \right\}}{n+1},
\]
which is super-uniform under the same assumptions as in Proposition~\ref{prop:superuniform}:
\[
\P{\tilde{\phi}_{\mathrm{rt}}(t; X_{n+j}) \le \alpha,\, T_{n+j} \leq t } \le \alpha.
\]

This completes the review of standard conformal p-values in the uncensored setting, for both left- and right-tail survival hypotheses. In the next section, we extend this framework to account for censoring, enabling application to real-world survival data.

\subsection{Conformal Inference for Survival Analysis under Right-Censoring} \label{sec:method-censoring}

We now return to the right-censored survival setting described in Section~\ref{sec:method-setup}, where the calibration data consist of i.i.d.\ samples $\{(X_i, \tilde{T}_i, E_i)\}_{i=1}^n$.
As in Section~\ref{sec:preliminaries}, our goal is to test the left-tail and right-tail hypotheses defined in~\eqref{eq:hyp-lt} and~\eqref{eq:hyp-rt} for new test points $X_{n+j}$.

The main challenge compared to the idealized setting of Section~\ref{sec:preliminaries} is that we do not observe the true event times $T_i$ for all calibration individuals. Consequently, we cannot compute the full-data conformal p-values $\tilde{\phi}_{\mathrm{lt}}(t; X_{n+j})$ or $\tilde{\phi}_{\mathrm{rt}}(t; X_{n+j})$ as defined earlier.

To address this, we adapt the conformal p-value construction using IPCW \citep{robins1992recovery,robins1993information}. This approach uses censored calibration samples but carefully reweights them to correct for censoring-induced selection bias, under the {\em non-informative censoring} assumption that $T \indep C \mid X$. Specifically, we rely on a fitted censoring model $\hat{\mathcal{M}}_C$ that estimates the conditional censoring survival function $\hat{S}_C(t \mid x) \approx \mathbb{P}(C > t \mid X = x)$.
Based on this model, we define the weight function
\begin{align} \label{eq:weights}
  \hat{w}(t,x) := \frac{1}{\hat{S}_C(t \mid x)},
\end{align}
approximating the inverse probability of being uncensored at time $t$ given covariates $x$.

Bringing this into the conformal framework, we define the IPCW version of the left-tail conformal p-value as:
\begin{align} \label{eq:pval-lt-censored}
  \hat{\phi}_{\mathrm{lt}}(t; X_{n+j}) := \frac{1 + \sum_{i=1}^{n} E_i \cdot \hat{w}(T_i, X_i) \cdot \mathbb{I}\left\{ \hat{s}_{\mathrm{lt}}(T_i; X_i, t) \geq \hat{s}_{\mathrm{lt}}(t; X_{n+j}, t) \right\}}{1 + \sum_{i=1}^{n} E_i \cdot \hat{w}(T_i, X_i)},
\end{align}
where the numerator includes only data points whose true event times are observed.
For example, in the special case where the scoring function takes the form defined in~\eqref{eq:score-lt}, under the relatively mild assumption that $\hat{S}_T(t \mid X_{n+j}) < 1$ almost surely, the left-tail conformal p-value in~\eqref{eq:pval-lt-censored} can be written as:
\begin{align*}
  \hat{\phi}_{\mathrm{lt}}(t; X_{n+j}) := \frac{1 + \sum_{i=1}^{n} E_i \cdot \hat{w}(T_i, X_i) \cdot \I{T_i \geq t} \cdot \mathbb{I}\left\{ \hat{S}_T(T_i \mid X_i) \leq \hat{S}_T(t \mid X_{n+j}) \right\}}{1 + \sum_{i=1}^{n} E_i \cdot \hat{w}(T_i, X_i)}.
\end{align*}

Analogously, we define the IPCW right-tail p-value as:
\begin{align} \label{eq:pval-rt-censored}
  \hat{\phi}_{\mathrm{rt}}(t; X_{n+j}) := \frac{1 + \sum_{i=1}^{n} E_i \cdot \hat{w}(T_i, X_i) \cdot \mathbb{I}\left\{ \hat{s}_{\mathrm{rt}}(T_i; X_i, t) \geq \hat{s}_{\mathrm{rt}}(t; X_{n+j}, t) \right\}}{1 + \sum_{i=1}^{n} E_i \cdot \hat{w}(T_i, X_i)}.
\end{align}
In the special case where the scoring function takes the form defined in~\eqref{eq:score-rt}, under the relatively mild assumption that $\hat{S}_T(t \mid X_{n+j}) > 0$ almost surely, this can be written as:
\begin{align*}
  \hat{\phi}_{\mathrm{rt}}(t; X_{n+j}) := \frac{1 + \sum_{i=1}^{n} E_i \cdot \hat{w}(T_i, X_i) \cdot \I{T_i \leq t} \cdot \mathbb{I}\left\{ \hat{S}_T(T_i \mid X_i) \geq \hat{S}_T(t \mid X_{n+j}) \right\}}{1 + \sum_{i=1}^{n} E_i \cdot \hat{w}(T_i, X_i)}.
\end{align*}

These IPCW p-values are fully computable from the observed censored data. As we now show, they yield asymptotically valid inference under the assumption that the censoring model $\hat{\mathcal{M}}_C$ is consistent, along with mild regularity conditions on the data distribution.

We define the censoring weight estimation error as:
\[
\Delta_N := \left( \mathbb{E} \left[ \left( \frac{1}{\hat{w}(T; X)} - \frac{1}{w^*(T; X)} \right)^2 \right] \right)^{1/2},
\]
where \( w^*(t,x) := 1 / S_C(t \mid x) \) denotes the true (unknown) censoring weight function, and $N$ represents the number of training data points.

We now list the assumptions under which asymptotic validity holds:
\begin{enumerate}[label=\textit{(A\arabic*)}]
\item \label{asmp:independence} The data are split into three independent parts:
  \begin{itemize}
  \item a training set of cardinality $N$ used to estimate \( \hat{w}(t; x) \), \( \hat{s}_{\mathrm{lt}}(t'; x, t) \), and \( \hat{s}_{\mathrm{rt}}(t'; x, t) \);
  \item an i.i.d.\ calibration set \( (X_i, T_i, C_i) \) for \( i = 1,\dots,n \), which is censored;
  \item and an i.i.d.\ test point \( (X_{n+j}, T_{n+j}, C_{n+j}) \), of which we only see the covariates.
  \end{itemize}
\item \label{asmp:cic} Censoring is non-informative, or conditionally independent: $T \indep C \mid X$.
\item \label{asmp:score-mono} The scoring function $\hat{s}_{\mathrm{lt}}(t'; x, t)$ is almost surely monotone increasing in $t'$ for all $x,t$, while $\hat{s}_{\mathrm{rt}}(t'; x, t)$ is monotone decreasing. Moreover, $\hat{s}_{\mathrm{lt}}(t; x, t) \in (0,1)$ and $\hat{s}_{\mathrm{rt}}(t; x, t) \in (0,1)$ for all $x,t$, almost surely.
\item \label{asmp:uncensored-prob} The probability of observing an event is bounded away from zero: \( \pi := \mathbb{P}(T \le C) > 0 \).
\item \label{asmp:weights-bounded} The estimated weights are bounded below: \( \hat{w}(T; X) \ge \omega_{\min} > 0 \) almost surely.
\item \label{asmp:weights-consistent} The weight estimation error vanishes asymptotically: \( \Delta_N \to 0 \) as \( N \to \infty \).
\end{enumerate}

\begin{theorem} \label{thm:asymptotic-validity}
Under Assumptions~\ref{asmp:independence}--\ref{asmp:weights-consistent}, for any fixed \( t > 0 \) and $\alpha \in (0,1)$, the IPCW conformal p-values defined in~\eqref{eq:pval-lt-censored} satisfy:
\begin{align*}
\limsup_{N,n \to \infty} \P{ \hat{\phi}_{\mathrm{lt}}(t; X_{n+1}) \le \alpha,\ T_{n+1} \geq t } & \le \alpha, \\
\limsup_{N,n \to \infty} \P{ \hat{\phi}_{\mathrm{rt}}(t; X_{n+1}) \le \alpha,\ T_{n+1} \leq t } & \le \alpha.
\end{align*}

\end{theorem}

In Section~\ref{sec:methods-bands}, we will explain how to use IPCW conformal p-values to construct \emph{conformal survival bands}.

\subsection{Conformal Survival Bands for Screening High- or Low-Risk Individuals} \label{sec:methods-bands}

Consider $m$ test patients with feature vectors $X_{n+1}, X_{n+2}, \ldots, X_{n+m}$. For each $j \in [m]:= \{1, \ldots, m\}$ and a fixed time $t > 0$, define
\begin{align} \label{eq:band-U}
\hat{U}(t; X_{n+j}) := \hat{\phi}_{\mathrm{lt}}^{\mathrm{BH}}(t; X_{n+j}),
\end{align}
where $\hat{\phi}_{\mathrm{lt}}^{\mathrm{BH}}(t; X_{n+j})$ denotes the Benjamini-Hochberg (BH) adjusted left-tail p-value.

More precisely, let $\hat{\phi}_{(1)} \le \hat{\phi}_{(2)} \le \cdots \le \hat{\phi}_{(m)}$ denote the ordered values of the unadjusted left-tail p-values $\{\hat{\phi}_{\mathrm{lt}}(t; X_{n+j})\}_{j=1}^m$, and let $\pi(j)$ be the index such that $\hat{\phi}_{(j)} = \hat{\phi}_{\mathrm{lt}}(t; X_{n+\pi(j)})$. Then the BH-adjusted p-values are defined as
\begin{align}
\label{eq:adjusted-pvals}
\hat{\phi}_{\mathrm{lt}}^{\mathrm{BH}}(t; X_{n+\pi(j)}) := \min_{k \ge j} \left\{ \frac{m}{k} \cdot \hat{\phi}_{(k)} \right\} \wedge 1, \quad \text{for } j = 1, \ldots, m.
\end{align}
Each value $\hat{\phi}_{\mathrm{lt}}^{\mathrm{BH}}(t; X_{n+j})$ represents the smallest FDR level at which the null hypothesis $H_{\mathrm{lt}}(t; j)$ can be rejected by the BH procedure \citep{benjamini1995controlling}. In practice, these adjusted p-values can be very easily computed using the \texttt{p.adjust} function in \texttt{R} with the argument \texttt{method = "BH"}.

The quantity $\hat{U}(t; X_{n+j})$ can then be interpreted as a calibrated upper bound on the survival probability at time $t$ for individual $X_{n+j}$. If $\hat{U}(t; X_{n+j}) \le \alpha$ for some threshold $\alpha \in (0,1)$, then the individual may be confidently flagged as \emph{high-risk}, with the guarantee that the expected proportion of false discoveries—individuals who survive past $t$ despite being flagged as ``high-risk''—remains below $\alpha$ in the large-sample limit.

Symmetrically, we define
\begin{align} \label{eq:band-L}
\hat{L}(t; X_{n+j}) := 1 - \hat{\phi}_{\mathrm{rt}}^{\mathrm{BH}}(t; X_{n+j}),
\end{align}
where $\hat{\phi}_{\mathrm{rt}}^{\mathrm{BH}}(t; X_{n+j})$ is the BH-adjusted p-value computed from the collection of right-tail IPCW conformal p-values $\{\hat{\phi}_{\mathrm{rt}}(t; X_{n+\ell})\}_{\ell=1}^m$. If $\hat{L}(t; X_{n+j}) \ge 1 - \alpha$, the patient can be confidently flagged as \emph{low-risk}, with the guarantee that the expected proportion of individuals who fail before $t$—despite being classified as low-risk—approximately remains below $\alpha$.

This construction defines our calibrated individual survival band
\begin{align} \label{eq:band-LU}
[\hat{L}(t; X_{n+j}),\, \hat{U}(t; X_{n+j})]
\end{align}
which could be interpreted as a calibrated range for likely survival probabilities at time $t$—offering statistically principled decision support for confident high- and low-risk screening.
The full procedure described above is summarized in Algorithm~\ref{alg:calibration-bands}.

\begin{algorithm}[!htb]
\caption{Construction of Conformal Survival Bands}
\label{alg:calibration-bands}
\begin{algorithmic}[1]
\INPUT Censored training dataset $\mathcal{D}_{\text{train}}$; censored calibration dataset $\mathcal{D}_{\text{cal}} = \{(X_i, \tilde{T}_i, E_i)\}_{i=1}^{n}$; test covariates $X_{n+1}, \ldots, X_{n+m}$; time grid $\mathcal{T} = \{t_1, \ldots, t_K\}$; trainable survival model $\hat{\mathcal{M}}_T$ and censoring model $\hat{\mathcal{M}}_C$.


\STATE Define the active calibration subset: $\mathcal{D}_{\text{cal}}' := \{(X_i, T_i) : E_i = 1,\, i \in [n]\}$.

\STATE Train $\hat{\mathcal{M}}_T$ and $\hat{\mathcal{M}}_C$ on $\mathcal{D}_{\text{train}}$.

\STATE Estimate censoring weights $\hat{w}(t, x)$ for the data in $\mathcal{D}_{\text{cal}}'$ using $\hat{\mathcal{M}}_C$ as in Eq.~\eqref{eq:weights}.

\FOR{each time $t \in \mathcal{T}$}
\FOR{each calibration point $(X_i, T_i) \in \mathcal{D}_{\text{cal}}'$}
\STATE Compute $\hat{s}_{\mathrm{lt}}(T_i; X_i, t)$ and $\hat{s}_{\mathrm{rt}}(T_i; X_i, t)$ using Eqs.~\eqref{eq:score-lt} and~\eqref{eq:score-rt}.
\ENDFOR
    \FOR{each test point $X_{n+j}$, for $j = 1, \ldots, m$}
        \STATE Compute scores $\hat{s}_{\mathrm{lt}}(t; X_{n+j}, t)$ and $\hat{s}_{\mathrm{rt}}(t; X_{n+j}, t)$ using Eqs.~\eqref{eq:score-lt} and~\eqref{eq:score-rt}.
        \STATE Compute IPCW p-values $\hat{\phi}_{\mathrm{lt}}(t; X_{n+j})$ and $\hat{\phi}_{\mathrm{rt}}(t; X_{n+j})$ using Eqs.~\eqref{eq:pval-lt-censored} and~\eqref{eq:pval-rt-censored}.
    \ENDFOR

    \STATE Apply BH procedure to p-values across test set to obtain adjusted values $\hat{\phi}_{\mathrm{lt}}^{\mathrm{BH}}(t; X_{n+j})$ and $\hat{\phi}_{\mathrm{rt}}^{\mathrm{BH}}(t; X_{n+j})$ for all $j \in [m]$ as in Eq.~\eqref{eq:adjusted-pvals}.

    \FOR{each test point $j \in [m]$}
        \STATE Compute endpoints $\hat{U}(t; X_{n+j})$ and $\hat{L}(t; X_{n+j})$ using Eqs.~\eqref{eq:band-U} and~\eqref{eq:band-L}.
    \ENDFOR
\ENDFOR

\OUTPUT Personalized survival bands $[\hat{L}(t; X_{n+j}), \hat{U}(t; X_{n+j})]$ for each $j \in [m]$ and all $t \in \mathcal{T}$.
\end{algorithmic}
\end{algorithm}

This algorithm is straightforward to implement and computationally efficient. Model training is performed once on the training dataset, and nonconformity scores for the uncensored calibration points are computed a single time, independently of the test set and the evaluation time grid. The only components that depend on the chosen time points are the nonconformity scores for the test samples, the IPCW conformal p-values, and their BH-adjusted counterparts. While the BH adjustment requires sorting across test samples, it is very fast in practice. All other steps can be parallelized over both time points and test samples. As a result, the overall computational cost of conformal calibration is modest, and typically negligible compared to the cost of training the survival or censoring models.

To implement Algorithm~\ref{alg:calibration-bands} in an even more data-efficient way, one may choose to re-allocate all censored observations from the calibration set $\mathcal{D}_{\text{cal}}$ into the training set. This yields a larger training dataset
$\mathcal{D}_{\text{train}}^\prime := \mathcal{D}_{\text{train}} \cup \{(X_i, \tilde{T}_i, E_i) : E_i = 0,\, i \in [n]\}$.
This does not affect the statistical validity of our method, since the conformal $p$-values $\hat{\phi}_{\mathrm{lt}}(t; X_{n+j})$ and $\hat{\phi}_{\mathrm{rt}}(t; X_{n+j})$ are computed exclusively using uncensored calibration points. While this approach deviates from the standard sample-splitting paradigm commonly used in conformal inference, it remains valid within our framework due to the use of inverse probability of censoring weighting (IPCW), which automatically corrects for the potential selection bias introduced by reallocation during the calibration phase.
However, it is not guaranteed that re-allocating the censored observations to the training set will improve model estimation, as doing so introduces its own form of sampling bias in the training data. Exploring the trade-offs involved in this reallocation strategy—particularly in combination with the possible use of importance weighting during training—remains an open direction for future work.

\subsection{Theoretical Guarantees for Calibrated Survival Bands} \label{sec:guarantees}

We will now formally state the asymptotic statistical guarantee provided by our method, which requires the following additional assumption:

\begin{enumerate}[label=\textit{(A\arabic*)}]
\setcounter{enumi}{5}
\item \label{asmp:score-regularity} For all \( t' > 0 \), the score functions \( \hat{s}_{\mathrm{lt}}(t'; X, t) \) admit a Lebesgue density \( f_{t'} \) satisfying
\[
0 < f_{\min} \le f_{t'}(u) \le f_{\max} < \infty \quad \text{for all } u > 0,
\]
for some constants \( f_{\min}, f_{\max} \in (0, \infty) \).
\end{enumerate}

\begin{theorem}
\label{thm:asymptotic-fdr}
Under Assumptions~\ref{asmp:independence}--\ref{asmp:score-regularity}, fix a target level \( \alpha \in (0,1) \) and survival cutoff time \( t > 0 \). Let \( p_j := \hat{\phi}_{\mathrm{lt}}(t; X_{n+j}) \) denote the IPCW p-value used to test the null hypothesis \( H_{\mathrm{lt}}(t; j) : T_{n+j} \ge t \), for \( j = 1,\dots,m \).
Define the false discovery rate (FDR) as
\[
\mathrm{FDR}_{m,n} := \mathbb{E} \left[ \frac{|\mathcal{R} \cap \mathcal{H}_0|}{|\mathcal{R}| \vee 1} \right],
\]
where \( \mathcal{R} := \{j : p_j \le \widehat{\tau} \} \) is the rejection set obtained by applying the BH procedure at level \( \alpha\) to \( (p_1, \dots, p_m) \) and \( \mathcal{H}_0 := \{j : T_{n+j} \ge t\} \) is the (random) set of true null hypotheses.
Then,
\[
\limsup_{N,n \to \infty} \mathrm{FDR}_{m,n} \le \alpha \quad \text{for any fixed } m.
\]
If in addition \( m = m_n \to \infty \) and \( m_n^2 \epsilon_n \to 0 \), then
\[
\limsup_{N,m,n \to \infty} \mathrm{FDR}_{m,n} \le \alpha.
\]
An analogous result holds for the right-tail p-values \( \hat{\phi}_{\mathrm{rt}}(t; X_{n+j}) \), used to test \( H_{\mathrm{rt}}(t; j)\).
\end{theorem}

A corollary of Theorem~\ref{thm:asymptotic-fdr} is that high- and low-risk screening rules based on the survival bands produced by Algorithm~\ref{alg:calibration-bands} described above are asymptotically well-calibrated. For any $t>0$ and $\alpha \in (0,1)$, define the set of test individuals flagged as high-risk at time \( t \) as
\[
\mathcal{F}_{\mathrm{hi}}(t; \alpha) := \left\{ j \in [m] : \hat{U}(t; X_{n+j}) \le \alpha \right\}
= \left\{j \in [m] : \hat{\phi}_{\mathrm{lt}}^{\mathrm{BH}}(t; X_{n+j}) \leq \alpha \right\},
\]
and similarly, define the set of individuals flagged as low-risk as
\[
\mathcal{F}_{\mathrm{lo}}(t; \alpha) := \left\{ j \in [m] : \hat{L}(t; X_{n+j}) \ge 1 - \alpha \right\}
= \left\{j \in [m] : \hat{\phi}_{\mathrm{rt}}^{\mathrm{BH}}(t; X_{n+j}) \leq \alpha \right\},
\]
Then, Theorem~\ref{thm:asymptotic-fdr} implies that, in the large-sample limit, the expected proportion of individuals in \( \mathcal{F}_{\mathrm{hi}}(t; \alpha) \) who actually survive past \( t \), and the expected proportion of individuals in \( \mathcal{F}_{\mathrm{lo}}(t; \alpha) \) who fail before \( t \), are both asymptotically controlled below level \( \alpha \).

\subsection{Extension: Making the Survival Bands Doubly Robust} \label{sec:extension-dr}

The validity guarantee provided by Theorem~\ref{thm:asymptotic-fdr} relies on the assumption that the censoring model used to estimate the inverse probability of censoring (IPC) weights is asymptotically consistent. Under this condition, our method enables approximately well-calibrated screening rules without requiring consistency of the black-box survival model.

We now describe a simple extension that makes our method {\em more conservative}, aiming to improve robustness to misspecification of the censoring model. Specifically, this variant targets \emph{double robustness}: it is designed to deliver approximately calibrated screening rules as long as \emph{either} the censoring model \emph{or} the survival model is consistent---though not necessarily both. 
We do not prove a formal double robustness guarantee, as the paper is already quite long. Nonetheless, we believe that a formal double robustness result for this extension could be established using techniques similar to those in \cite{sesia2024doubly}. 

For each test point $j \in [m]$, we define the adjusted conformal survival band:
\begin{align} \label{eq:band-LU-DR}
[\hat{L}^{\mathrm{DR}}(t; X_{n+j}),\, \hat{U}^{\mathrm{DR}}(t; X_{n+j})],
\end{align}
where
\begin{align} \label{eq:band-U-DR}
\hat{U}^{\mathrm{DR}}(t; X_{n+j}) := \min \left\{ \hat{\phi}_{\mathrm{lt}}^{\mathrm{BH}}(t; X_{n+j}), \hat{S}_T(t \mid X_{n+j}) \right\},
\end{align}
and
\begin{align} \label{eq:band-L-DR}
\hat{L}^{\mathrm{DR}}(t; X_{n+j}) := \max \left\{ 1 - \hat{\phi}_{\mathrm{rt}}^{\mathrm{BH}}(t; X_{n+j}), \hat{S}_T(t \mid X_{n+j}) \right\}.
\end{align}
Above, $\hat{S}_T(t \mid X_{n+j})$ is the point estimate of the survival probability produced by $\hat{\mathcal{M}}_T$, and $\hat{\phi}_{\mathrm{lt}}^{\mathrm{BH}}(t; X_{n+j})$ and $\hat{\phi}_{\mathrm{rt}}^{\mathrm{BH}}(t; X_{n+j})$ are the upper and lower bounds produced by Algorithm~\ref{alg:calibration-bands}.
This adjusted construction was used in the illustrative example shown in Figure~\ref{fig:calibration_band_example}.

A useful side benefit of this adjusted band is enhanced interpretability: Equation~\eqref{eq:band-LU-DR} guarantees that the predicted survival probability $\hat{S}_T(t \mid X_{n+j})$ always lies within the band itself. This aligns with common practitioner expectations, simplifying communication.

\section{Numerical Experiments with Synthetic Data} \label{sec:experiments-synthetic}

\subsection{Setup}

\paragraph{Synthetic data.}
We consider four data-generating distributions, summarized in Table~\ref{tab:distributions-synthetic} (Appendix~\ref{app:experiments-synthetic}), which span a range of interesting settings, inspired by previous works. In each setting, $p=100$ covariates $X = (X_1, \ldots, X_p)$ are generated independently, while $T$ and $C$ are sampled independently conditional on $X$, from either a log-normal distribution—$\log T \mid X \sim \mathcal{N}(\mu(X), \sigma(X))$—or an exponential distribution—$C \mid X \sim \text{Exp}(\lambda(X))$.

The first three settings are borrowed from \cite{sesia2024doubly} and are ordered by decreasing modeling difficulty. The first two simulate challenging scenarios where the true survival distribution is highly nonlinear or complex, making accurate model fitting difficult and highlighting the need for robust uncertainty estimation via conformal inference. The third setting, originally appearing in \cite{candes2023conformalized}, corresponds to a simpler survival distribution that is easier to learn with standard survival models, reducing the gap between model-based and conformal prediction. Finally, the fourth setting, which is new, introduces a covariate shift between the training and calibration/test distributions, leading to potential model misspecification even if the survival model fits the training distribution well, and further emphasizing the value of conformal inference.

\paragraph{Design, Methods, and Performance Metrics.}
We generate independent training and calibration datasets with varying sample sizes, along with a test dataset containing 1000 samples. Right-censoring is introduced by replacing the true event time $T$ and censoring time $C$ with the observed time $\tilde{T} = \min(T, C)$ and event indicator $E = \mathbb{I}(T \leq C)$. The censored training data are used to fit survival and censoring models, as described below.

Our goal is to evaluate the proposed conformal survival band (CSB) method and compare it to natural baseline approaches for screening test patients according to rules of the form $P[T > t] > p$ (for \textit{low-risk} selection) or $P[T > t] < p$ (for \textit{high-risk} selection), evaluated at various time thresholds $t$ and probability levels $p$. We consider three methods for screening patients based on estimated conditional survival probabilities: (i) \textit{model}, using the point estimates of conditional survival probabilities computed by a black-box survival model $\hat{\mathcal{M}}^{\text{surv}}$ fitted on all available data (training and calibration); (ii) \textit{KM}, using a Kaplan–Meier estimator fitted on the calibration data; and (iii) \textit{CSB}, using the conformal survival band produced by Algorithm~\ref{alg:calibration-bands}, implemented with the doubly robust extension described in Section~\ref{sec:extension-dr}. 

We focus on the doubly robust implementation of CSB in our experiments because it is more interpretable, ensuring that the predicted survival probability always lies within the band, as a practitioner would likely expect. Although formally studying the double robustness property of this extension is beyond the scope of this paper, this approach is already supported by the statistical guarantees established in Section~\ref{sec:guarantees}, since it is always more conservative than the baseline method described in Section~\ref{sec:methods-bands}.

In addition, we compare against an ideal \textit{oracle} that uses the true conditional survival probabilities from the data-generating distribution; although not practical, the oracle provides valuable insight into the achievable performance limits.

Performance is evaluated on the independent test set using several metrics. The \textit{screened proportion} measures the fraction of test patients selected by the screening rule; this quantity should ideally be large, reflecting high power, provided that selections are accurate. The \textit{survival rate} measures the proportion of selected patients who survive beyond time $t$, and should be larger than $p$ for low-risk screening and smaller than $p$ for high-risk screening. Finally, \textit{precision} and \textit{recall} are computed relative to the oracle screening decisions, treating screening as a binary classification task, and both should ideally be high.

All experiments are independently repeated 100 times, and results are reported as means across repetitions, with error bars representing two standard errors (SE).

\paragraph{Models.}
We consider four model families for fitting the censoring and survival models, denoted respectively by $\hat{\mathcal{M}}^{\text{cens}}$ and $\hat{\mathcal{M}}^{\text{surv}}$, to ensure consistent comparisons across different calibration methods. The model families are:
(1) \textit{grf}, a generalized random forest (R package \texttt{grf});
(2) \textit{survreg}, an accelerated failure time (AFT) model with a log-normal distribution (R package \texttt{survival});
(3) \textit{rf}, a random survival forest (R package \texttt{randomForestSRC});
and (4) \textit{cox}, the Cox proportional hazards model (R package \texttt{survival}).
The survival model $\hat{\mathcal{M}}^{\text{surv}}$ is used both by the \textit{model} screening method and as input to the CSB method. In contrast, the censoring model $\hat{\mathcal{M}}^{\text{cens}}$ is only used within the CSB method, to compute the weights for constructing conformal survival bands.

\paragraph{Leveraging Prior Knowledge on $P_{C \mid X}$.}
 To examine the effect of incorporating prior knowledge about the censoring distribution, we fit $\hat{\mathcal{M}}^{\text{cens}}$ using only the first $p_1 \leq p$ covariates, knowing that $C$ is independent of $(X_{p_1+1}, \ldots, X_p)$ given $(X_1, \ldots, X_{p_1})$ within the data-generating distributions used in these experiments (Table~\ref{tab:distributions-synthetic}). Therefore, for $10 \leq p_1 \leq p = 100$, this prior knowledge helps improve the censoring model by excluding irrelevant predictors and mitigating overfitting. We start with $p_1 = 10$ and later evaluate the impact of larger $p_1$, representing weaker prior knowledge. The case $p_1 = p$ corresponds to no prior knowledge, where all covariates are used to fit the censoring model.

\subsection{Results}

\paragraph{Effect of the Training Sample Size.}
Figure~\ref{fig:exp-setup1-setting1} compares the performance of the four methods in Setting 1 (the harder case), using the \textit{grf} models. The training sample size varies between 100 and 10,000, while the calibration sample size is fixed at 500. We evaluate two screening rules: selecting low-risk patients with $P(T > 6.00) > 0.80$ and selecting high-risk patients with $P(T > 12.00) < 0.80$.

\textit{Low-risk screening results (top panel).}
The \textit{model} selects too many patients, resulting in an excess of false positives; the average survival rate among its selected patients falls below 60\%, despite the target threshold of 80\%. The \textit{KM} method performs even worse, selecting nearly all patients, with a survival rate close to 50\%. In contrast, \textit{CSB} achieves the desired survival rate among selected patients, provided that the training sample size is not too small, and its performance improves steadily as the sample size increases, approaching that of the oracle. Even when the training size is small (e.g., 100), and the fitted survival and censoring models are relatively inaccurate, CSB is more robust than the \textit{model}.
As expected, the oracle achieves a survival rate above 80\%—in fact closer to 100\%—while selecting approximately 50\% of the test patients. It is important to note that the survival rate of selected patients does not need to match the threshold $p$ exactly, even for the oracle, because many patients have true survival probabilities much higher (or lower) than $p$.

\textit{High-risk screening results (bottom panel).}
In this setting, all methods lead to subsets of selected patients whose survival rates are below 80\%, as desired. However, the \textit{KM} method again selects too many patients, including many whose true survival probabilities are actually higher than 80\%, resulting in lower precision. This behavior stems from the non-personalized nature of KM estimates: KM cannot flexibly select subsets of patients based on individual characteristics and must either select nearly all or none. In contrast, the \textit{model} and \textit{CSB} methods achieve higher precision and recall, and both approach the oracle performance as the training sample size increases.

\begin{figure}[!htb]
    \centering
    \includegraphics[width=\textwidth]{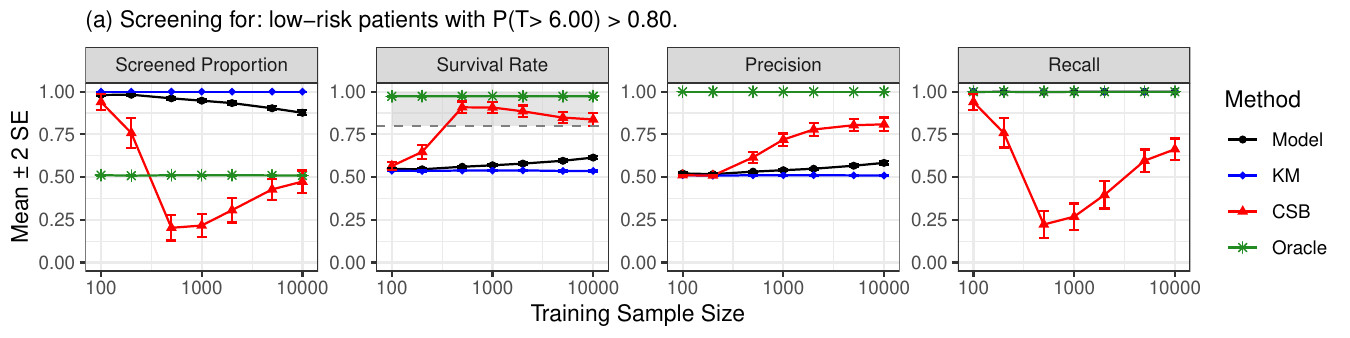}
    \includegraphics[width=\textwidth]{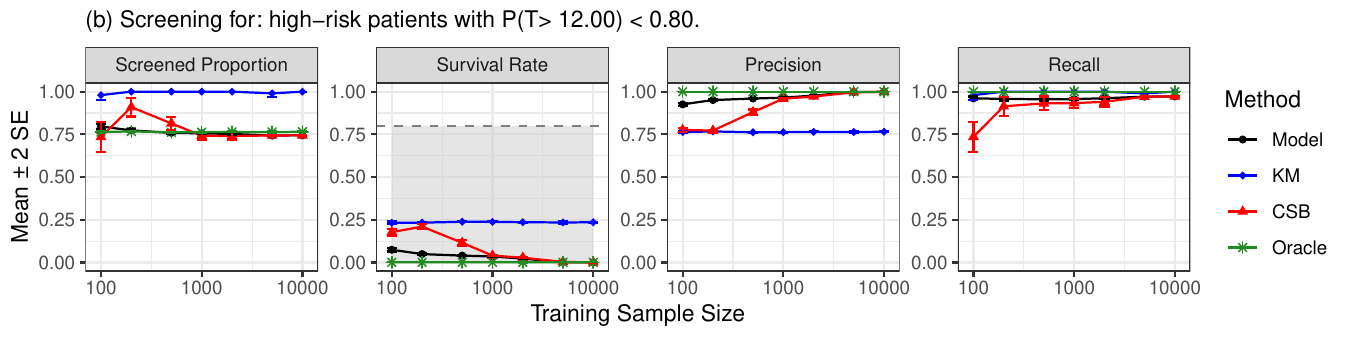}
\caption{
Effect of training sample size on patient screening performance of different methods in a challenging synthetic data scenario with complex survival and censoring distributions, using \textit{grf} models.
Top: results for low-risk screening at time $t=6$; bottom: high-risk screening at time $t=12$.
The calibration sample size is fixed at 500.
The conformal survival band (CSB) method successfully achieves survival rates above the target $p = 0.8$ for low-risk screening and below $p = 0.8$ for high-risk screening, while making personalized selections that approach the oracle performance as the training sample size increases.
}
    \label{fig:exp-setup1-setting1}
\end{figure}

Figures~\ref{fig:exp-setup1-setting2} and~\ref{fig:exp-setup1-setting3} in Appendix~\ref{app:experiments-synthetic} present similar results for Settings 2 and 3, respectively.
Across these experiments, the CSB method consistently achieves survival rates on the correct side of the target threshold and maintains relatively high precision and recall compared to the other methods, more closely approximating the performance of the ideal oracle as the training sample size increases.

\paragraph{Effect of the Calibration Sample Size.}
Figure~\ref{fig:exp-setup3-setting1} examines the impact of the calibration sample size on the performance of CSB in the same challenging synthetic data setting as before, with the training sample size fixed at 5000.
In this regime, where the survival and censoring models are already accurate thanks to a relatively large training set, increasing the calibration sample size improves screening performance: the survival rates among selected patients consistently falls on the correct side of the target threshold, and precision and recall tend to approach the behavior of the ideal oracle.
In general, however, if the training sample size is small and the fitted models are inaccurate, it may be more beneficial to prioritize training over calibration. Overall, a moderate calibration sample size (on the order of a few hundred) is typically sufficient to obtain reliable conformal inferences.
Figures~\ref{fig:exp-setup3-setting2} and~\ref{fig:exp-setup3-setting3} in Appendix~\ref{app:experiments-synthetic} present qualitatively similar results from analogous experiments conducted under synthetic data under Settings 2 and 3.

\begin{figure}[!htb]
    \centering
    \includegraphics[width=\textwidth]{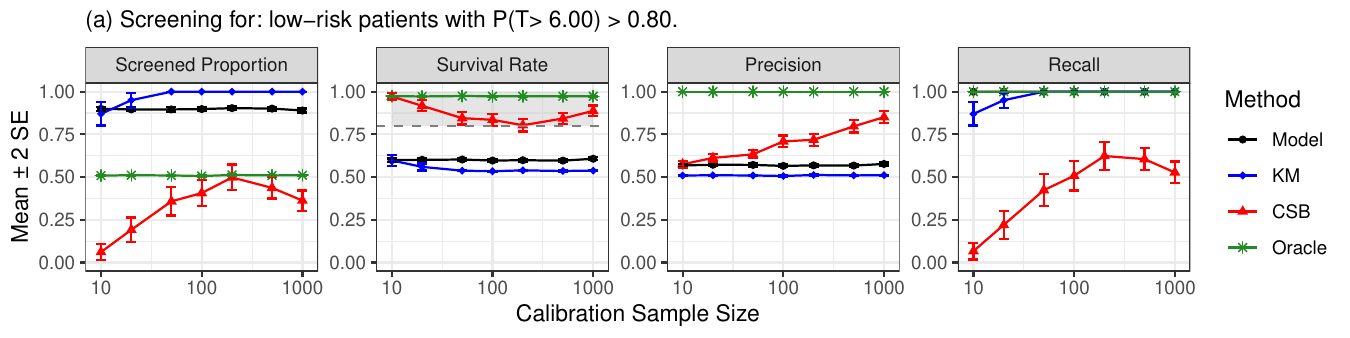}
    \includegraphics[width=\textwidth]{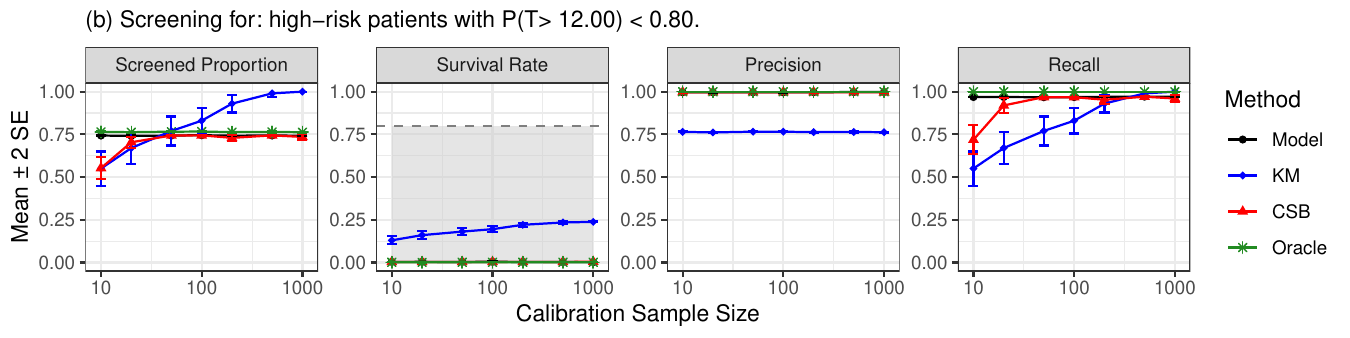}
\caption{
Effect of calibration sample size on patient screening performance with conformal survival bands (CSB) in a challenging synthetic data scenario with complex survival and censoring distributions, as in Figure~\ref{fig:exp-setup1-setting1}.
The training sample size is fixed at 5000.
The CSB method successfully achieves survival rates above the target $p = 0.8$ for low-risk screening and below $p = 0.8$ for high-risk screening, while making selections that more closely resemble those of an ideal oracle as the calibration sample size increases.
}
    \label{fig:exp-setup3-setting1}
\end{figure}

\paragraph{Effect of the Training Sample Size for the Censoring Model.}
Figure~\ref{fig:exp-setup5-setting1} presents results from experiments similar to those in Figure~\ref{fig:exp-setup1-setting1}, but here the censoring model is fit using only a subset from a total of 5000 training samples.
The goal is to study how the quality of the censoring model affects the performance of the CSB method.
In the top panel (low-risk screening), we observe that when the censoring model is trained on too few samples, CSB may fail to provide valid screening selections.
As the training size increases and the censoring model improves, the survival rate among selected patients eventually falls on the correct side of the target threshold, consistent with our asymptotic theoretical results.
In the bottom panel (high-risk screening), CSB maintains validity even with a small censoring training set, but its power (i.e., ability to identify appropriate patients) improves as the censoring model quality increases.
Figures~\ref{fig:exp-setup5-setting2} and~\ref{fig:exp-setup5-setting3} in Appendix~\ref{app:experiments-synthetic} present qualitatively similar results for the relatively easier Settings 2 and 3.

\begin{figure}[!htb]
    \centering
    \includegraphics[width=\textwidth]{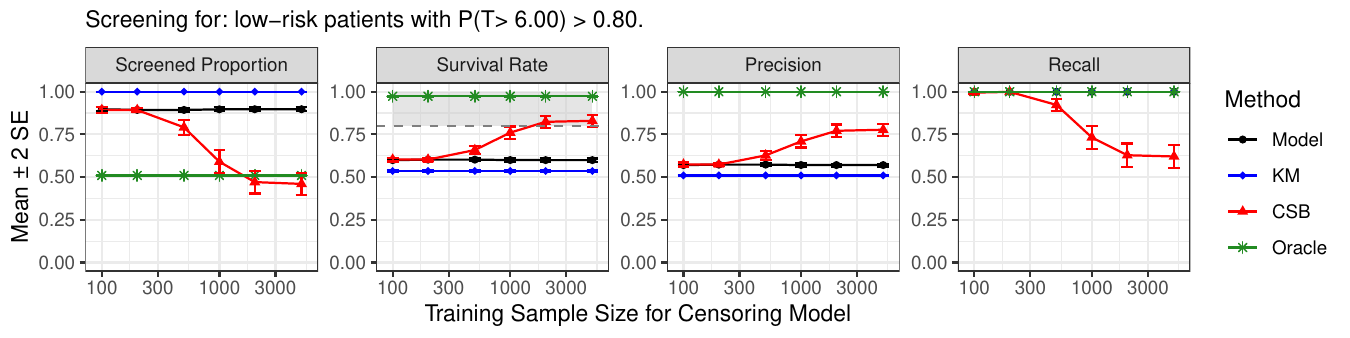}
    \includegraphics[width=\textwidth]{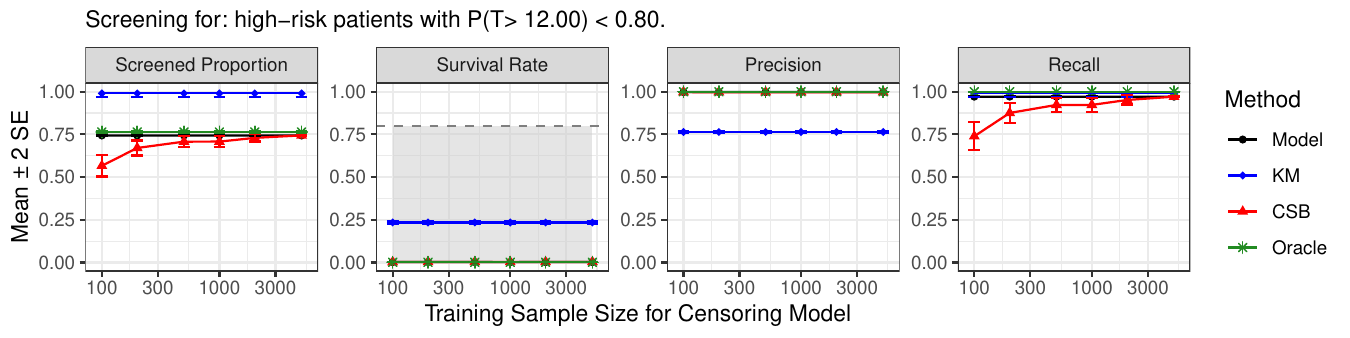}
\caption{
Effect of the training sample size used for fitting the censoring model on patient screening performance with conformal survival bands (CSB) in a challenging synthetic data scenario with complex survival and censoring distributions, as in Figure~\ref{fig:exp-setup1-setting1}.
The overall training sample size is fixed at 5000, but only a subset is used to fit the censoring model.
CSB tends to produce higher-quality selections as the censoring model improves with more training data, achieving survival rates on the correct side of the target threshold (top), and improving power toward oracle performance (bottom).
}
    \label{fig:exp-setup5-setting1}
\end{figure}

\paragraph{Additional Experiments.}
Additional experiments studying the effect of the number of features used to fit the censoring model are presented in Figures~\ref{fig:exp-setup6-setting1} and~\ref{fig:exp-setup6-setting3} in Appendix~\ref{app:experiments-synthetic}.
These results show that when too many features are included—especially with limited sample sizes—performance may degrade due to challenges in accurately estimating the censoring model, consistent with the theory.
Tables~\ref{tab:setup4-setting1-lr}--\ref{tab:setup4-setting3-hr} in Appendix~\ref{app:experiments-synthetic} report on additional experiments using different survival and censoring models, leading to qualitatively similar results.

\paragraph{Effect of Distribution Shift.}
Finally, we evaluate robustness under a distribution shift between the training and calibration/test distributions, corresponding to Setting~4 described in Table~\ref{tab:distributions-synthetic}.
Specifically, we now allow the survival distribution to change between training and calibration/test phases, which may lead to model misspecification even if the survival model fits the training data well.
In these experiments, the training set consists of 5000 samples and the calibration set consists of 500 samples.

Table~\ref{tab:setup2-lr} summarizes the performance of different methods for low-risk screening under this setting.
The ``low-quality model'' refers to a survival model trained on the shifted training distribution, resulting in inaccurate predictions on the calibration and test distributions.
In this case, the \textit{model} baseline selects all patients as ``low risk'' at time $t=3$, but the observed survival rate is only about 54\%, far below the target threshold of $p = 0.80$.
By contrast, the \textit{KM} and \textit{CSB} methods conservatively select no patients, thereby maintaining validity at the cost of reduced power.
With a ``high-quality model'' trained directly on the calibration/test distribution, the \textit{CSB} method successfully screens approximately half of the patients while maintaining valid survival rates, closely approaching oracle performance.
Meanwhile, the \textit{KM} method remains overly conservative due to its lack of adaptivity, continuing to select no patients even when safe screening is feasible.
Figure~\ref{fig:calibration_band_example}, introduced in Section~\ref{sec:intro}, illustrates survival curves, conformal survival bands, and screening decisions corresponding to these experiments for four example patients in the low-risk setting.

\begin{table}[!htb]
\centering
\caption{
Performance of different methods for screening low-risk patients with $P(T > 3) > 0.80$ under distribution shift (Setting~4 from Table~\ref{tab:distributions-synthetic}).
Results are shown separately for a low-quality survival model trained on the shifted training distribution and a high-quality model trained directly on the calibration/test distribution.
The survival rate highlighted in red illustrates how a misspecified \textit{grf} model can lead to invalid screening rules, resulting in substantially lower survival rates than expected among patients labeled as ``low-risk.''
}
\label{tab:setup2-lr}

\begin{tabular}{lcccc}
\toprule
Method & Screened & Survival & Precision & Recall\\
\midrule
\addlinespace[0.3em]
\multicolumn{5}{l}{\textbf{Low-Quality Model}}\\
\hspace{1em}Model & 1.000 ± 0.000 & \textcolor{red}{0.538 ± 0.003} & 0.499 ± 0.003 & 1.000 ± 0.000\\

\hspace{1em}KM & 0.000 ± 0.000 & NA & NA & 0.000 ± \vphantom{1} 0.000\\

\hspace{1em}CSB & 0.000 ± 0.000 & NA & NA & 0.000 ± 0.000\\

\hspace{1em}Oracle & 0.499 ± 0.003 & \textcolor{black}{1.000 ± 0.000} & 1.000 ± 0.000 & 1.000 ± 0.000\\

\addlinespace[0.3em]
\multicolumn{5}{l}{\textbf{High-Quality Model}}\\
\hspace{1em}Model & 0.494 ± 0.004 & \textcolor{black}{1.000 ± 0.000} & 1.000 ± 0.000 & 0.994 ± 0.001\\

\hspace{1em}KM & 0.000 ± 0.000 & NA & NA & 0.000 ± 0.000\\
\hspace{1em}CSB & 0.492 ± 0.004 & \textcolor{black}{1.000 ± 0.000} & 1.000 ± 0.000 & 0.989 ± 0.001\\

\hspace{1em}Oracle & 0.497 ± 0.004 & \textcolor{black}{1.000 ± 0.000} & 1.000 ± 0.000 & 1.000 ± 0.000\\
\bottomrule
\end{tabular}

\end{table}

Additional results for high-risk screening under distribution shift are presented in Table~\ref{tab:setup2-hr} in Appendix~\ref{app:experiments-synthetic}.
In this case, model misspecification primarily reduces the power of all methods, leading to fewer selections.
A similar figure illustrating survival curves, conformal bands, and screening decisions for the high-risk case is provided in Appendix~\ref{app:experiments-synthetic} as Figure~\ref{fig:calibration_band_example-hr}.

\section{Application to Real Data} \label{sec:application}

\subsection{Setup}

We apply our method to the same seven publicly available datasets analyzed by \citet{sesia2024doubly}: COLON, GBSG, HEART, METABRIC, PBC, RETINOPATHY, and VALCT. Details regarding the number of observations, covariates, and data sources are provided in Table~\ref{tab:datasets_summary} in Appendix~\ref{app:data}.
Additional information about the standard preprocessing procedures applied to these datasets---designed to ensure compatibility with various survival and censoring models---is also available in Appendix~\ref{app:data}.

We follow the procedure from Section~\ref{sec:experiments-synthetic}, comparing \textit{CSB} to the \textit{model} and \textit{KM} benchmarks. Since the true data-generating distribution is unknown, the oracle method is not applicable.
All methods are applied separately using each of the same four model types from Section~\ref{sec:experiments-synthetic} to estimate the survival distribution—generalized random forests (\textit{grf}), parametric survival regression (\textit{survreg}), random survival forests (\textit{rf}), and Cox proportional hazards models (\textit{cox})—with the censoring distribution always estimated using \textit{grf}. Each dataset is split into 60\% training, 20\% calibration, and 20\% testing sets, and all experiments are repeated 100 times with independent random splits.

As in the synthetic experiments, we screen test patients using rules of the form $P[T > t] > p$ (low-risk) or $P[T > t] < p$ (high-risk), evaluated at fixed time thresholds $t$ and probability levels $p$. We consider four tasks: low- and high-risk selection with $p = 0.8$ at time $t_1$, and with $p = 0.25$ at $t_2$. To facilitate comparisons while accounting for heterogeneity in the distribution of survival times across data sets, the time thresholds $t_1$ and $t_2$ are set to the 0.1 and 0.9 quantiles of all observed times, respectively.

Performance is evaluated on the test set using two metrics. The \textit{screened proportion} measures the fraction of test patients selected for each screening task, and should ideally be large under accurate selection. The \textit{survival rate} is the proportion of selected patients who survive beyond time $t$. Since this cannot be computed exactly due to censoring, we report deterministic bounds: the lower bound treats censored patients as failures at their censoring times, and the upper bound assumes they survive indefinitely. Together, these provide a conservative interval for the true survival rate among selected patients.

\subsection{Results}

Table~\ref{tab:data-summary} summarizes the results of this data analysis, aggregating the performance of each screening method across tasks, datasets, and repetitions, separately for each survival model. For each combination of survival model and screening method, we report the average proportion of test patients selected (\textit{screened proportion}) and the empirical distribution of verification outcomes: \textit{valid}, \textit{dubious}, and \textit{invalid}. Verification outcomes are determined for each task by assessing whether the survival rate bounds, averaged over 100 random repetitions, fall entirely on the correct side of the threshold $p$ (valid), straddle it (dubious), or lie entirely on the wrong side (invalid), accounting for uncertainty via two standard errors.

\begin{table}[!htb]
\centering
\caption{Summary of screening performance for high- and low-risk patient selection methods based on different survival models, aggregated across datasets, screening tasks, and repetitions. The \textit{screened proportion} denotes the average fraction of test patients selected. Although false positives cannot be directly verified due to censoring, approximate verification is possible and classified as \textit{valid}, \textit{dubious}, or \textit{invalid}. Note that the aggregated screened proportion is always 0.5 for the \textit{Model} and \textit{KM} benchmarks, which do not account for uncertainty and always assign each patient to either the high- or low-risk group for each pair $(p, t)$ of screening parameters. In contrast, CSB may leave some patients unclassified—neither high- nor low-risk—when the $(p,t)$ point falls within the ``gray area'' uncertainty band (e.g., see Figure~\ref{fig:calibration_band_example}), resulting in lower and more variable aggregated screening rates.}
\label{tab:data-summary}

\begin{tabular}{lccccc}
\toprule
\multicolumn{3}{c}{ } & \multicolumn{3}{c}{Verification Outcome} \\
\cmidrule(l{3pt}r{3pt}){4-6}
Method & Survival Model & Screened Proportion & Valid & Dubious & Invalid\\
\midrule
 & grf & 0.500 & 0.643 & 0.357 & 0.000\\

 & survreg & 0.500 & 0.714 & 0.250 & 0.036\\

 & rf & 0.500 & 0.536 & 0.464 & 0.000\\

\multirow[t]{-4}{*}{\raggedright\arraybackslash Model} & Cox & 0.500 & 0.571 & 0.393 & 0.036\\

 & grf & 0.500 & 0.786 & 0.214 & 0.000\\

 & survreg & 0.500 & 0.786 & 0.214 & 0.000\\

 & rf & 0.500 & 0.786 & 0.214 & 0.000\\

\multirow[t]{-4}{*}{\raggedright\arraybackslash KM} & Cox & 0.500 & 0.786 & 0.214 & 0.000\\

 & grf & 0.293 & 0.929 & 0.071 & 0.000\\

 & survreg & 0.245 & 0.964 & 0.036 & 0.000\\

 & rf & 0.301 & 0.857 & 0.143 & 0.000\\

\multirow[t]{-4}{*}{\raggedright\arraybackslash CSB} & Cox & 0.299 & 0.857 & 0.143 & 0.000\\
\bottomrule
\end{tabular}

\end{table}

CSB yields the highest proportion of valid selections across all survival models, with few or no dubious or invalid outcomes, despite screening fewer patients. In contrast, the \textit{Model} and \textit{KM} benchmarks screen more patients but exhibit lower verification quality.
These results illustrate a fundamental trade-off between power and reliability, unsurprisingly.

Compared to the synthetic data experiments from Section~\ref{sec:experiments-synthetic}, the advantage of CSB is somewhat less obvious here. This may be due in part to imperfect verification, stemming from censoring in the test data and the absence of oracle knowledge. Moreover, survival models may estimate the conditional survival distributions more accurately here than in the synthetic settings, reducing the relative gains from conformal inference. As noted by \citet{sesia2024doubly}, these datasets do not appear to be especially difficult to model. Nevertheless, in general, it remains difficult to assess the reliability of screening results produced by black-box methods. Our approach can provides more confidence in the calibration of such screening results, which may be particularly valuable in high-stakes applications where the cost of false positives is substantial, even if it entails some reduction in power.

Additional breakdowns of these results for specific screening tasks are provided in Appendix~\ref{app:data}, with low-risk and high-risk summaries shown in Tables~\ref{tab:data-detailed-lr} and~\ref{tab:data-detailed-hr}, respectively.

\FloatBarrier
\section{Discussion}

This paper introduced a conformal inference method for constructing uncertainty bands around individual survival curves under right-censoring, enabling statistically principled personalized risk screening under minimal assumptions on the data-generating process.

Experiments with synthetic data showed that screening based on uncalibrated black-box models can be unreliable, particularly in hard-to-model settings, whereas our method provides greater robustness, albeit with some conservativeness.
In our real data applications, standard survival models seemed to produce reasonable screening performance, likely because the fitted models were able to approximate the true survival distribution quite well. Nonetheless, our method remains appealing in more complex or high-stakes scenarios, where model misspecification is a concern and formal uncertainty quantification is essential.

Two limitations of the current approach are its reliance on asymptotic FDR control and its focus on pointwise inference, which assumes that the screening thresholds $t$ and $p$ are fixed in advance. While it may be possible to obtain finite-sample and simultaneous inference guarantees, doing so would likely require a significantly more conservative method. Understanding this trade-off between statistical rigor and screening power in greater depth presents an intriguing direction for future research.

An additional limitation is that our method can be conservative in some settings, resulting in fewer selections than the ideal oracle or even fewer than the black-box survival model when that model appears well calibrated. A possible way to mitigate this over-conservativeness in future work would be to incorporate an adaptive estimator of the proportion of null hypotheses \citep{10.1093/biomet/asae051}, similar to how Storey's method extends the BH procedure \citep{storey2004strong}.

Another promising direction for improvement concerns the treatment of censored calibration points. In Section~\ref{sec:methods-bands}, we proposed reallocating these samples to the training set but did not do this in practice due to concerns about sampling bias potentially degrading model quality. Recent work by \citet{farina2025doubly} introduces a more complex strategy for incorporating censored observations directly into the calibration phase, which may be possible to adapt to our setting. A potentially simpler alternative is to reallocate censored calibration samples to the training set and apply importance weighting during training to mitigate bias, while retaining the calibration procedure described in this paper, which uses only uncensored observations. Both strategies deserve further investigation.

Additional extensions include de-randomizing our inferences via model aggregation across different data splits \citep{carlsson2014aggregated,linusson2017calibration}, potentially using e-values \citep{vovk2023confidence, wang2022false, NEURIPS2023_cec8ad77}, and adapting the method to handle noisy or corrupted data \citep{sesia2023adaptive,pmlr-v230-clarkson24a}.

\section*{Software Availability}
Open-source software implementing our methods and experiments is available at \url{https://github.com/msesia/conformal_survival_screening}.

\section*{Acknowledgements} 
The authors thank two anonymous referees for detailed feedback and valuable suggestions on an earlier version of this manuscript.
M.~S.~is partly supported by NSF grant DMS 2210637, a Capital One CREDIF Research Award, and a Google Research Scholar Award.



\bibliography{bibliography}

\clearpage

\appendix
\section{Mathematical Proofs} \label{app:proofs}

\subsection{Proof of Proposition~\ref{prop:superuniform}}

\begin{proof}[of Proposition~\ref{prop:superuniform}]

Under the null hypothesis that $T_{n+j} \geq t$, because $\hat{s}_{\mathrm{lt}}(t'; x, t)$ is monotone increasing in $t'$ for all $x$, we have that, almost surely
\begin{align*}
  \hat{s}_{\mathrm{lt}}(t; X_{n+j}, t) \leq \hat{s}_{\mathrm{lt}}(T_{n+j}; X_{n+j}, t),
\end{align*}
and thus
\begin{align*}
  \tilde{\phi}_{\mathrm{lt}}(t; X_{n+j}) 
  & := \frac{1 + \sum_{i=1}^n \mathbb{I}\left\{ \hat{s}_{\mathrm{lt}}(T_i; X_i, t) \ge \hat{s}_{\mathrm{lt}}(t; X_{n+j}, t) \right\}}{n+1} \\
  & \geq \frac{1 + \sum_{i=1}^n \mathbb{I}\left\{ \hat{s}_{\mathrm{lt}}(T_i; X_i, t) \ge \hat{s}_{\mathrm{lt}}(T_{n+j}; X_{n+j}, t) \right\}}{n+1}.
\end{align*}
Hence
\begin{align*}
  \P{\tilde{\phi}_{\mathrm{lt}}(t; X_{n+j}) \le \alpha,\, T_{n+j} \le t}
  & \leq \P{\frac{1 + \sum_{i=1}^n \mathbb{I}\left\{ \hat{s}_{\mathrm{lt}}(T_i; X_i, t) \ge \hat{s}_{\mathrm{lt}}(T_{n+j}; X_{n+j}, t) \right\}}{n+1} \leq \alpha} \\
  & \leq \alpha,
\end{align*}
where the last inequality follows immediately through standard conformal inference arguments because $( \hat{s}_{\mathrm{lt}}(T_1; X_1, t), \ldots, \hat{s}_{\mathrm{lt}}(T_n; X_n, t), \hat{s}_{\mathrm{lt}}(T_{n+j}; X_{n+j}, t))$ are exchangeable.
\end{proof}

\subsection{Proofs of Theorems~\ref{thm:asymptotic-validity} and~\ref{thm:asymptotic-fdr}}

\begin{proof}[of Theorem~\ref{thm:asymptotic-validity}: Asymptotic validity of IPCW p-values]
We apply the finite-sample bound from Theorem~\ref{thm:pval-valid}, which states that for any \( t > 0 \) and any \( n \geq 1 \),
\begin{align*}
\mathbb{P} \left[ \hat{\phi}_{\mathrm{lt}}(t; X_{n+1}) \leq \alpha ,  T_{n+1} \geq t \right]
&\leq \alpha 
+ \frac{1}{\omega_{\min} \pi} \left[ \frac{2(2\omega_{\min} + 1)}{n} + 8 \sqrt{ \frac{ 2 \log(2n) }{n} } \right]
+ 2 \Delta_N.
\end{align*}

By assumption:
\begin{itemize}
\item \( \hat{w}(T; X) \geq \omega_{\min} > 0 \) almost surely,
\item \( \Delta_N := \left( \mathbb{E}\left[ \left( \frac{1}{\hat{w}(T; X)} - \frac{1}{w^*(T; X)} \right)^2 \right] \right)^{1/2} \to 0 \) as \( N \to \infty \).
\end{itemize}
Since the remaining terms in the bound all vanish as \( N,n \to \infty \), it follows that
\[
\limsup_{N,n \to \infty} \mathbb{P} \left[ \hat{\phi}_{\mathrm{lt}}(t; X_{n+1}) \leq \alpha ,  T_{n+1} \geq t \right] \leq \alpha.
\]

An identical argument applies to the limit involving $\hat{\phi}_{\mathrm{rt}}(t; X_{n+1})$, completing the proof.
\end{proof}

\begin{proof}[of Theorem~\ref{thm:asymptotic-fdr}: Asymptotic FDR control]
We prove the result for the left-tail p-values \( \hat{\phi}_{\mathrm{lt}}(t; X_{n+j}) \); the argument for right-tail p-values is analogous.

The proof proceeds in three steps:
\begin{enumerate}
    \item[(1)] We compare each empirical p-value to an oracle counterpart and establish uniform closeness.
    \item[(2)] We show that the BH rejection sets based on these two sets of p-values are close with high probability.
    \item[(3)] We conclude FDR control for the empirical procedure via stability and the oracle FDR guarantee.
\end{enumerate}

\paragraph{Step 1: Empirical p-values are close to oracle p-values.}
For each test point \( j \in \{1, \dots, m\} \), define the empirical and oracle p-values as:
\[
p_j := \hat{\phi}_{\mathrm{lt}}(t; X_{n+j}), \qquad p_j^* := \phi^*_{\mathrm{lt}}(t; X_{n+j}) := \mathbb{P} \left( \hat{s}_{\mathrm{lt}}(T; X, t) \ge \hat{s}_{\mathrm{lt}}(t; X_{n+j}, t) \right),
\]
where the probability in \( p_j^* \) is taken over an independent copy \( (T, X) \sim P \).

By Corollary~\ref{cor:survival-score-bound-both}, for any \( \delta > 0 \), with probability at least \( 1 - \delta \) over the calibration data,
\[
\max_{1 \le j \le m} |p_j - p_j^*| \le \epsilon_n,
\]
where
\[
\epsilon_n = \frac{1}{\omega_{\min} \pi} \left[ \frac{2(\omega_{\min} + 1)}{n} + 8 \sqrt{ \frac{ \log(2n/\delta) }{n} } \right] + 2 \Delta_N.
\]

\paragraph{Step 2: Stability of BH rejection sets.}
Let \( \mathcal{R} := \{j : p_j \le \widehat{\tau} \} \) and \( \mathcal{R}^* := \{j : p_j^* \le \widehat{\tau}^* \} \) be the rejection sets obtained by applying the BH procedure at level \( \alpha \) to the empirical and oracle p-values, respectively.

Before we can apply Lemma~\ref{lemma:BH-stability} (a general stability bound for the BH procedure), we need to verify that the oracle p-values \( p_j^* \) have a density bounded above. Since
\[
p_j^* = \phi^*_{\mathrm{lt}}(t; X_{n+j}) = 1 - F(\hat{s}_{\mathrm{lt}}(t; X_{n+j}, t)),
\]
where \( F \) is the CDF of the score variable \( \hat{s}_{\mathrm{lt}}(T; X, t) \), with corresponding density $f(u)=F'(u)$ uniformly bounded above and below across all \( u > 0 \). Then, the change-of-variables formula gives that the density $f^*$ of $p^*_j$ is:
\[
f^*(u) = \frac{f_{t}(F^{-1}(1 - u))}{f(F^{-1}(1 - u))},
\]
which, thanks to Assumptions~\ref{asmp:score-mono} and~\ref{asmp:score-regularity}, we know to be bounded from both above and below by:
\[
\frac{f_{\min}}{f_{\max}} \leq f^*(u) \leq \frac{f_{\max}}{f_{\min}} \quad \text{for all } u \in [0,1].
\]
Now we can apply Lemma~\ref{lemma:BH-stability} with \( p = (p_1, \dots, p_m) \) and \( q = (p_1^*, \dots, p_m^*) \), obtaining:
\[
\mathbb{P}(\mathcal{R} \ne \mathcal{R}^*) \le \delta + 2 \epsilon_n m^2 \cdot \frac{f_{\max}}{f_{\min}}.
\]

\paragraph{Step 3: Transferring FDR control from the oracle to empirical procedure.}
By Lemma~\ref{lemma:super-uniformity-oracle}, the oracle p-values \( p_j^* \) satisfy the super-uniformity property:
\[
\mathbb{P}(p_j^* \le \alpha, T_{n+j} \ge t) \le \alpha, \qquad \text{ for all } \alpha \in (0,1),
\]
and the p-values \( p_j^* \) are mutually independent (conditional on the calibration data). Therefore, by Theorem~\ref{thm:bh-fdr-joint-superuniform}, BH applied to \( (p_1^*, \ldots, p_m^*) \) controls the FDR below \( \alpha \):
\[
\mathrm{FDR}^* := \mathbb{E} \left[ \frac{|\mathcal{R}^* \cap \mathcal{H}_0|}{|\mathcal{R}^*| \vee 1} \right] \le \alpha,
\]
where \( \mathcal{H}_0 := \{j : T_{n+j} \ge t\} \) is the (random) set of true nulls.

\paragraph{Conclusion.}
By definition of the FDR and triangle inequality,
\[
|\mathrm{FDR} - \mathrm{FDR}^*| \le \mathbb{P}(\mathcal{R} \ne \mathcal{R}^*).
\]
Combining this with the oracle guarantee \( \mathrm{FDR}^* \le \alpha \) and the stability bound from Lemma~\ref{lemma:BH-stability}, we obtain the finite-sample bound:
\[
\mathrm{FDR} \le \alpha + \delta + 2 \epsilon_n m^2 \cdot \frac{f_{\max}}{f_{\min}}.
\]
Letting \( \delta \to 0 \) and noting that \( \epsilon_n \to 0 \) as \( n \to \infty \), we conclude:
\[
\limsup_{N,n \to \infty} \mathrm{FDR}_n \le \alpha \quad \text{for fixed } m.
\]
Moreover, if \( m = m_n \to \infty \) and \( m_n^2 \epsilon_n \to 0 \), then the finite-sample error term vanishes:
\[
\delta + 2 \epsilon_n m^2 \cdot \frac{f_{\max}}{f_{\min}} \to 0,
\]
and we conclude the joint limit:
\[
\limsup_{N,m,n \to \infty} \mathrm{FDR}_{m,n} \le \alpha.
\]
\end{proof}

\subsection{Auxiliary Results}

\subsubsection{Validity of Oracle p-Values}

\begin{lemma}
\label{lemma:super-uniformity-oracle}
Let $(X, T)$ and $(X', T')$ be independent copies of a pair of random variables taking values in $\mathcal{X} \times \mathbb{R}$, and let $s: \mathbb{R} \times \mathcal{X} \to \mathbb{R}$ be a fixed function that is monotone increasing in its first argument.
Define the oracle p-value function
\[
\phi^*(t, x) := \mathbb{P}\left( s(T, X) \ge s(t, x) \right),
\]
where the probability is taken over $(X, T)$.
Then, for any $\alpha \in [0,1]$,
\[
\mathbb{P}\left( \phi^*(t, X') \le \alpha,\ T' \ge t \right) \le \alpha.
\]

\end{lemma}

\begin{proof}[of Lemma~\ref{lemma:super-uniformity-oracle}]
Let $S := -s(T, X)$ and $S' := -s(T', X')$, and define the cumulative distribution function $F_S(u) := \mathbb{P}(S \le u)$. By definition of $\phi^*$, we can write:
\[
\phi^*(t, X') = \mathbb{P}(S \le -s(t, X')) = F_S(-s(t, X')).
\]
Now observe that, under the condition $T' \ge t$ and monotonicity of $s$ in $t$, we have:
\[
s(t, X') \le s(T', X') \quad \Rightarrow \quad -s(t, X') \ge -s(T', X'),
\]
and since $F_S$ is non-decreasing we have the almost-surely inequality:
\[
\phi^*(t, X') = F_S(-s(t, X')) \geq F_S(-s(T', X')) = F_S(S').
\]
Hence:
\[
\mathbb{P}\left( \phi^*(t, X') \le \alpha,\ T' \ge t \right)
\le \mathbb{P}\left( F_S(S') \le \alpha,\ T' \ge t \right)
\le \mathbb{P}(F_S(S') \le \alpha) = \alpha,
\]
where the final equality uses the probability integral transform: since $S'$ has the same distribution as $S$ and $F_S$ is its CDF, the random variable $F_S(S')$ is uniform on $[0,1]$.
\end{proof}

\subsubsection{Finite-sample version of Theorem~\ref{thm:asymptotic-validity}}

\begin{theorem} \label{thm:pval-valid}
Under the assumptions of Theorem~\ref{thm:asymptotic-validity},
\begin{align*}
  \mathbb{P} \left[ \hat{\phi}_{\mathrm{lt}}(t; X_{n+1}) \leq \alpha, T_{n+1} \geq t \right] 
  & \leq \alpha + \frac{1}{\omega_{\min} \pi} \left[ \frac{2(2\omega_{\min} + 1)}{n} + 8 \sqrt{ \frac{ 2\log(2n \sqrt{n}) }{n} } \right] + 2 \Delta_N,
\end{align*}
for any $t>0$, and
\begin{align*}
  \mathbb{P} \left[ \hat{\phi}_{\mathrm{rt}}(t; X_{n+1}) \leq \alpha, T_{n+1} \leq t \right] 
  & \leq \alpha + \frac{1}{\omega_{\min} \pi} \left[ \frac{2(2\omega_{\min} + 1)}{n} + 8 \sqrt{ \frac{ 2\log(2n \sqrt{n}) }{n} } \right] + 2 \Delta_N.
\end{align*}

\end{theorem}

\begin{proof}[of Theorem~\ref{thm:pval-valid}]
We prove the result for \( \hat{\phi}_{\mathrm{lt}} \); the argument for \( \hat{\phi}_{\mathrm{rt}} \) is identical and omitted.

For $\delta \in (0,1)$, define the event 
\begin{align*}
  \mathcal{E}_{t,\delta} & := \left\{ \sup_{x \in \mathbb{R}^d } \left| \hat{\phi}_{\mathrm{lt}}(t; x) - \phi^*_{\mathrm{lt}}(t; x) \right| \leq \epsilon_n(\delta) \right\}, \\
  \epsilon_n(\delta) & := \frac{1}{\omega_{\min} \pi} \left[ \frac{2(\omega_{\min} + 1)}{n} + 8 \sqrt{ \frac{ \log(2n/\delta) }{n} } \right] + 2 \Delta_N.
\end{align*}
Corollary~\ref{cor:survival-score-bound-both} implies:
\begin{align*}
  & \P{ \hat{\phi}_{\mathrm{lt}}(t; x) \leq \alpha, T_{n+1} \geq t \mid X_{n+1}=x } \\
  & \qquad = \P{ \hat{\phi}_{\mathrm{lt}}(t; x) \leq \alpha, \mathcal{E}_{t,\delta}, T_{n+1} \geq t \mid X_{n+1}=x }
   + \P{ \hat{\phi}_{\mathrm{lt}}(t; x) \leq \alpha, \mathcal{E}^c_{t,\delta}, T_{n+1} \geq t \mid X_{n+1}=x } \\
  & \qquad \leq \P{ \phi^*_{\mathrm{lt}}(t; x) \leq \alpha + \epsilon_n(\delta), T_{n+1} \geq t \mid X_{n+1}=x } + \delta.
\end{align*}
Therefore,
\begin{align*}
  \P{ \hat{\phi}_{\mathrm{lt}}(t; x) \leq \alpha, T_{n+1} \geq t }
  & \leq \P{ \phi^*_{\mathrm{lt}}(t; x) \leq \alpha + \epsilon_n(\delta), T_{n+1} \geq t } + \delta \\
  & \leq \alpha + \epsilon_n(\delta) + \delta,
\end{align*}
where the second inequality above follows directly from Lemma~\ref{lemma:super-uniformity-oracle}.
Finally, setting $\delta = 1/n$:
\begin{align*}
  \mathbb{P} \left[ \hat{\phi}_{\mathrm{lt}}(t; X_{n+1}) \leq \alpha, T_{n+1} \geq t \right] 
  & \leq \alpha + \frac{1}{\omega_{\min} \pi} \left[ \frac{2(2\omega_{\min} + 1)}{n} + 8 \sqrt{ \frac{ 2 \log(2n) }{n} } \right] + 2 \Delta_N.
\end{align*}

\end{proof}

\begin{corollary}[Specialized Instance of Theorem~\ref{thm:uniform-bound-reweighted}]
\label{cor:survival-score-bound-both}
Under the assumptions of Theorem~\ref{thm:asymptotic-validity}, for any $x \in \mathbb{R}^d$ and $t > 0$, define:
\begin{align*}
  \phi^*_{\mathrm{lt}}(t; x)
  & := \mathbb{P}\left( \hat{s}_{\mathrm{lt}}(T; X, t) \geq \hat{s}_{\mathrm{lt}}(t; x, t) \right), \\
  \phi^*_{\mathrm{rt}}(t; x)
  & := \mathbb{P}\left( \hat{s}_{\mathrm{rt}}(T; X, t) \geq \hat{s}_{\mathrm{rt}}(t; x, t) \right),
\end{align*}
where the probabilities are taken with respect to the randomness in $(X,T)$.

Then, for any \( t > 0 \) and \( \delta \in (0,1) \), with probability at least \( 1 - \delta \) over the randomness in the calibration data $\mathcal{D} = \{(X_i,T_i, C_i)\}_{i=1}^{n}$:
\begin{align*}
\sup_{x \in \mathbb{R}^d } \left| \hat{\phi}_{\mathrm{lt}}(t; x) - \phi^*_{\mathrm{lt}}(t; x) \right|
&\leq \frac{1}{\omega_{\min} \pi} \left[ \frac{2(\omega_{\min} + 1)}{n} + 8 \sqrt{ \frac{ \log(2n/\delta) }{n} } \right] + 2 \Delta_N.
\end{align*}
Similarly, for any \( t > 0 \) and \( \delta \in (0,1) \), with probability at least \( 1 - \delta \):
\begin{align*}
\sup_{x \in \mathbb{R}^d } \left| \hat{\phi}_{\mathrm{rt}}(t; x) - \phi^*_{\mathrm{rt}}(t; x) \right|
&\leq \frac{1}{\omega_{\min} \pi} \left[ \frac{2(\omega_{\min} + 1)}{n} + 8 \sqrt{ \frac{ \log(2n/\delta) }{n} } \right] + 2 \Delta_N.
\end{align*}
\end{corollary}

\begin{proof}[of Corollary~\ref{cor:survival-score-bound-both}]
We prove the bound for \( \hat{\phi}_{\mathrm{lt}} \); the argument for \( \hat{\phi}_{\mathrm{rt}} \) is identical and omitted.
Define the following variables:
\begin{align*}
  Z_i &= (X_i, T_i), \\
  D_i &= \mathbf{1}(T_i \leq C_i), \\
  E_i &= \hat{w}(T_i; X_i), \\
  e(Z_i) &= w^*(T_i; X_i), \\
  \psi(Z_i) &= \hat{s}_{\mathrm{lt}}(T_i; X_i, t).
\end{align*}
Recall the definitions of \( R_n(u) \) and \( R(u) \) from Theorem~\ref{thm:uniform-bound-reweighted}:
\begin{align*}
R_n(u) &:= \frac{1 + \sum_{i=1}^n (D_i / E_i) \cdot \mathbf{1}(\psi(Z_i) \geq u)}{1 + \sum_{i=1}^n (D_i / E_i)}, \\
R(u) &:= \mathbb{E} \left[ \frac{D}{e(Z)} \cdot \mathbf{1}(\psi(Z) \geq u) \right].
\end{align*}
For any \( x \in \mathbb{R}^d \) and \( t > 0 \), let $\hat{\phi}_{\mathrm{lt}}(t; x) := R_n\left( \hat{s}_{\mathrm{lt}}(t; x, t) \right)$ and note that, by Assumption~\ref{asmp:cic},
\[
\phi^*_{\mathrm{lt}}(t; x) = R\left( \hat{s}_{\mathrm{lt}}(t; x, t) \right).
\]
Therefore,
\begin{align*}
\sup_{x \in \mathbb{R}^d} \left| \hat{\phi}_{\mathrm{lt}}(t; x) - \phi^*_{\mathrm{lt}}(t; x) \right|
&= \sup_{x \in \mathbb{R}^d} \left| R_n(\hat{s}_{\mathrm{lt}}(t; x, t)) - R(\hat{s}_{\mathrm{lt}}(t; x, t)) \right| \\
&\leq \sup_{u \in \mathbb{R}} \left| R_n(u) - R(u) \right|,
\end{align*}
and the result follows directly from Theorem~\ref{thm:uniform-bound-reweighted}, noting that \( E_i \geq \omega_{\min} \) and \( \mathbb{E}[D] = \pi \).
\end{proof}

\subsubsection{Uniform Concentration for General Reweighted Survival Estimator}
\begin{theorem}\label{thm:uniform-bound-reweighted}
Let $(Z_i, D_i, E_i)_{i=1}^n$ be i.i.d.\ samples from some unknown distribution $\mathcal{P}$, where $Z_i \in \mathbb{R}^d$, $D_i \in \{0,1\}$, and $E_i \in [\underline{e}, 1]$ almost surely for some constant $\underline{e} > 0$.
For any $z \in \mathbb{R}^d$, define $e(z) := \mathbb{P}(D_i = 1 \mid Z_i = z)$, which may be interpreted as a propensity score.
 For any fixed function $\psi : \mathbb{R}^{d+1} \to \mathbb{R}$ and any $u \in \mathbb{R}$, define
\begin{align*}
& R_n(u) := \frac{1 + \sum_{i=1}^n (D_i / E_i) \cdot \mathbf{1}(\psi(Z_i) \geq u)}{1 + \sum_{i=1}^n (D_i / E_i)},
& R(u) := \mathbb{E} \left[ \frac{D}{e(X)} \cdot \mathbf{1}(\psi(Z) \geq u) \right],
\end{align*}
where $(Z, D, E)$ is a generic random sample from $\mathcal{P}$.
Define also
\begin{align*}
& \Delta_N := \left( \mathbb{E} \left[ \left( \frac{1}{E} - \frac{1}{e(X)} \right)^2 \right] \right)^{1/2},
& \pi := \mathbb{E}[D].
\end{align*}
Then for any $\delta \in (0,1)$, with probability at least $1 - \delta$,
\begin{align}
\sup_{u \in \mathbb{R}} |R_n(u) - R(u)|
&\leq \frac{1}{\underline{e}\pi} \left[ \frac{2(\underline{e} + 1)}{n} + 8 \sqrt{ \frac{ \log(2n/\delta) }{n} } \right] + 2 \Delta_N.
\end{align}
\end{theorem}

\begin{proof}[of Theorem~\ref{thm:uniform-bound-reweighted}]
The proof splits the error into two parts: a stochastic deviation term (comparing the estimator $R_n(u)$ to a population analogue $\tilde{R}(u)$ based on the approximate weights $E_i$), and a bias term due to discrepancy between $E_i$ and the propensity score $e(X_i)$.

Define
\begin{align*}
\tilde{R}(u) := \frac{ \mathbb{E} \left[ (D/E) \cdot \mathbf{1}(\psi(Z) \geq u) \right] }{ \mathbb{E} \left[ D/E \right] },
\end{align*}
where the expectations are over the joint distribution of $(X, D, Z, E)$. Then decompose the error:
\begin{align*}
|R_n(u) - R(u)| \leq |R_n(u) - \tilde{R}(u)| + |\tilde{R}(u) - R(u)|.
\end{align*}

\paragraph{Step 1 (Bounding the stochastic deviation).} We begin by computing a uniform bound for $|R_n(u) - \tilde{R}(u)|$. Let $W_i := D_i / E_i \in [0, 1/\underline{e}]$ and $W := D / E \in [0, 1/\underline{e}]$. Define
\begin{align*}
N_n(u) &:= \frac{1 + \sum_{i=1}^n W_i \cdot \mathbf{1}(\psi(Z_i) \geq u)}{1 + n} , \\
A_n &:= \frac{1 + \sum_{i=1}^n W_i}{1 + n}, \\
\tilde{N}(u) &:= \mathbb{E}[W \cdot \mathbf{1}(\psi(Z) \geq u)], \\
\tilde{A} &:= \mathbb{E}[W],
\end{align*}
so that
\begin{align*}
& R_n(u) = \frac{N_n(u)}{A_n},
& \tilde{R}(u) = \frac{\tilde{N}(u)}{\tilde{A}}.
\end{align*}

Using the inequality
\begin{align*}
\left| \frac{a}{b} - \frac{a'}{b'} \right| \leq \frac{|a - a'|}{b'} + |a| \left| \frac{1}{b} - \frac{1}{b'} \right|,
\end{align*}
with \( a = N_n(u), a' = \tilde{N}(u), b = A_n, b' = \tilde{A} \), and the fact that \( N_n(u) \leq A_n \), we obtain
\begin{align} \label{eq:decomp-1}
\begin{split}
|R_n(u) - \tilde{R}(u)| 
  &\leq \frac{|N_n(u) - \tilde{N}(u)|}{\tilde{A}} + A_n \left| \frac{1}{A_n} - \frac{1}{\tilde{A}} \right| \\
  &= \frac{|N_n(u) - \tilde{N}(u)| + |A_n - \tilde{A}|}{\tilde{A}} \\
  &\leq \frac{|N_n(u) - \tilde{N}(u)| + |A_n - \tilde{A}|}{\pi},
\end{split}
\end{align}
because $E \leq 1$ and thus
$\tilde{A} = \mathbb{E}[W] = \mathbb{E}[D/E] \geq \mathbb{E}[ D ] = \pi$.\\

\textit{Uniform concentration bound for $|N_n(u) - \tilde{N}(u)|$.}
To bound the deviation
\begin{align*}
\sup_{u \in \mathbb{R}} \left| \frac{1}{n} \sum_{i=1}^n W_i \mathbf{1}(\psi(Z_i) \geq u) - \mathbb{E}[W \cdot \mathbf{1}(\psi(Z) \geq u)] \right|,
\end{align*}
note that the class \( \{ u \mapsto W_i \mathbf{1}(\psi(Z_i) \geq u) \} \) has envelope bounded by \( 1/\underline{e} \) and VC dimension 1. Applying Theorem 4.10 and Lemma 4.14 in Wainwright (2019), we obtain: with probability at least \( 1 - \delta \),
\begin{align*}
\sup_{u \in \mathbb{R}} \left| \frac{1}{n} \sum_{i=1}^n W_i \mathbf{1}(\psi(Z_i) \geq u) - \tilde{N}(u) \right|
\leq \frac{4}{\underline{e}} \sqrt{ \frac{ \log(n+1) }{n} } + \sqrt{ \frac{2 \log(1/\delta)}{n \underline{e}^2} }.
\end{align*}

Now recall that
\begin{align*}
N_n(u) = \frac{1}{1+n} + \frac{1}{1+n} \sum_{i=1}^n W_i \mathbf{1}(\psi(Z_i) \geq u),
\end{align*}
so we may write
\begin{align*}
|N_n(u) - \tilde{N}(u)| 
&\leq \frac{1 + 1/\underline{e}}{1+n} + \left| \frac{1}{n} \sum_{i=1}^n W_i \mathbf{1}(\psi(Z_i) \geq u) - \tilde{N}(u) \right|.
\end{align*}

Therefore, with probability at least \( 1 - \delta \),
\begin{align} \label{eq:uniform-bound-1}
\sup_{u \in \mathbb{R}} |N_n(u) - \tilde{N}(u)| 
\leq \frac{1 + 1/\underline{e}}{1+n} 
+ \frac{1}{\underline{e}} \left( 4 \sqrt{ \frac{ \log(n+1) }{n} } + \sqrt{ \frac{2 \log(1/\delta)}{n} } \right).
\end{align}

\textit{Hoeffding bound for \(|A_n - \tilde{A}|\).}
Recall that
\begin{align*}
A_n := \frac{1 + \sum_{i=1}^n W_i}{1 + n}, \qquad \tilde{A} := \mathbb{E}[W].
\end{align*}
We decompose the deviation:
\begin{align*}
|A_n - \tilde{A}| 
&\leq \frac{1+1/\underline{e}}{1+n} + \left| \frac{1}{n} \sum_{i=1}^n W_i - \mathbb{E}[W] \right|.
\end{align*}
Applying Hoeffding’s inequality for \( W_i \in [0, 1/\underline{e}] \), we obtain that, with probability at least \( 1 - \delta \),
\begin{align}  \label{eq:hoeffding-bound-1}
|A_n - \tilde{A}| 
\leq \frac{1+1/\underline{e}}{1+n} + \frac{1}{\underline{e}} \sqrt{ \frac{ \log(2/\delta) }{2n} }.
\end{align}

Combining~\eqref{eq:decomp-1} with~\eqref{eq:uniform-bound-1} and~\eqref{eq:hoeffding-bound-1}, we obtain that, with probability at least \( 1 - \delta \),
\begin{align} \label{eq:stochastic-bound}
\begin{split}
  \sup_{u \in \mathbb{R}} |R_n(u) - \tilde{R}(u)| 
  & \leq \frac{1}{\underline{e}\pi} \left[ \frac{2(\underline{e} + 1)}{1+n} 
    + 4 \sqrt{ \frac{ \log(n+1) }{n} } + \sqrt{ \frac{2 \log(1/\delta)}{n} } + \sqrt{ \frac{ \log(2/\delta) }{2n} } \right] \\
  & \leq \frac{1}{\underline{e}\pi} \left[ \frac{2(\underline{e} + 1)}{n} 
    + 8 \sqrt{ \frac{ \log(2n/\delta) }{n} }  \right].
\end{split}
\end{align}

\paragraph{Step 2 (Bounding the bias due to approximation).}
We now bound the difference between $\tilde{R}(u)$ and the ideal target $R(u)$, which uses the true propensity $e(X)$. Using the same ratio inequality as before we obtain:
\begin{align} \label{eq:decomp-2}
\begin{split}
|\tilde{R}(u) - R(u)| 
&= \left| \tilde{N}(u) - R(u) \right| + \left| 1 - \tilde{A}\right|.
\end{split}
\end{align}

To bound the first term, note that, using the Cauchy-Schwarz inequality,
\begin{align*}
\left| \tilde{N}(u) - R(u) \right| 
&= \left| \mathbb{E} \left[ \left( \frac{1}{E} - \frac{1}{e(X)} \right) D \cdot \mathbf{1}(\psi(Z) \geq u) \right] \right| \\
&\leq \mathbb{E} \left[ \left| \frac{1}{E} - \frac{1}{e(X)} \right| \right] 
\leq \sqrt{ \mathbb{E} \left[ \left( \frac{1}{E} - \frac{1}{e(X)} \right)^2 \right] } =: \Delta_N.
\end{align*}
For the second term, using the same argument;
\begin{align*}
|1 - \tilde{A}| 
= \left| 1 - \mathbb{E}\left[ \frac{D}{E} \right] \right|
= \left| \mathbb{E} \left[ \left( \frac{1}{e(X)} - \frac{1}{E} \right) D \right] \right| 
\leq \Delta_N.
\end{align*}
Therefore, we conclude:
\begin{align} \label{eq:bias-bound}
\sup_{u \in \mathbb{R}} |\tilde{R}(u) - R(u)| \leq 2 \Delta_N.
\end{align}

\paragraph{Step 3 (Conclusion).} Combining the stochastic deviation bound~\eqref{eq:stochastic-bound} with the approximation bias bound~\eqref{eq:bias-bound}, we conclude that with probability at least $1 - \delta$,
\begin{align*}
\sup_{u \in \mathbb{R}} |R_n(u) - R(u)|
&\leq \sup_u |R_n(u) - \tilde{R}(u)| + \sup_u |\tilde{R}(u) - R(u)| \\
&\leq \frac{1}{\underline{e}\pi} \left[ \frac{2(\underline{e} + 1)}{n} + 8 \sqrt{ \frac{ \log(2n/\delta) }{n} } \right] + 2 \Delta_N.
\end{align*}
\end{proof}

\subsubsection{Stability and validity of the BH procedure}

\begin{lemma}[Stability of BH under small perturbations]
\label{lemma:BH-stability}
Let $p = (p_1, \dots, p_m)$ and $q = (q_1, \dots, q_m)$ be two vectors of p-values in $[0,1]^m$. Suppose:
\begin{itemize}
  \item (Closeness) For some $\epsilon > 0$ and $\delta \in (0,1)$, we have
  \[
  \mathbb{P}\left( \max_{1 \le j \le m} |p_j - q_j| \le \epsilon \right) \ge 1 - \delta.
  \]
  \item (Smoothness) Each $q_j$ has a density $f_j$ supported on $[0,1]$ satisfying $f_j(t) \le f_{\max}$ for all $t \in [0,1]$.
\end{itemize}
Let $\mathcal{R}_p$ and $\mathcal{R}_q$ be the rejection sets from applying the BH procedure at level $\alpha$ to $p$ and $q$, respectively. Then:
\[
\mathbb{P}(\mathcal{R}_p \ne \mathcal{R}_q) \le \delta + 2\epsilon m^2 f_{\max}.
\]
\end{lemma}

\begin{proof}[of Lemma~\ref{lemma:BH-stability}]
Define the event
\[
\mathcal{E} := \left\{ \max_{1 \le j \le m} |p_j - q_j| \le \epsilon \right\},
\]
which holds with probability at least $1 - \delta$ by assumption.

For each $k \in \{1, \dots, m\}$, the BH procedure compares p-values to candidate thresholds $\tau_k := \frac{\alpha k}{m}$. Define the ``instability region'':
\[
\mathcal{B}_\epsilon := \bigcup_{k=1}^m \left( \tau_k - \epsilon,\ \tau_k + \epsilon \right),
\]
which is the union of $m$ intervals of width $2\epsilon$, and hence has total Lebesgue measure at most $2\epsilon m$.

For any fixed index $j$, the probability that $q_j \in \mathcal{B}_\epsilon$ is easily bounded using the smoothness assumption:
$\mathbb{P}(q_j \in \mathcal{B}_\epsilon) \le 2\epsilon m f_{\max}$.
Applying a union bound over all $j = 1, \dots, m$, we get:
\[
\mathbb{P}(\exists j: q_j \in \mathcal{B}_\epsilon) \le 2\epsilon m^2 f_{\max}.
\]

Let $A$ be the ``bad event'' where either closeness fails or some $q_j$ lies near a threshold:
\[
A := \mathcal{E}^c \cup \{ \exists j: q_j \in \mathcal{B}_\epsilon \}.
\]
Then:
\[
\mathbb{P}(A) \le \delta + 2\epsilon m^2 f_{\max}.
\]

On the complement $A^c$ (i.e., when $p$ and $q$ are close and all $q_j$ are away from thresholds), the BH procedure makes the same rejection decisions on $p$ and $q$ because: each possible threshold $\tau_k$ is fixed; each $p_j$ and $q_j$ differ by at most $\epsilon$; no $q_j$ lies within $\epsilon$ of any threshold.
Hence, for every $j$, the relation \( p_j \le \tau_k \) matches \( q_j \le \tau_k \) for all $k$.
Therefore, the sorted comparisons and BH cutoffs lead to identical rejection sets:
$\mathcal{R}_p = \mathcal{R}_q$ on the event $A^c$.
Thus,
\[
\mathbb{P}(\mathcal{R}_p \ne \mathcal{R}_q) \le \mathbb{P}(A) \le \delta + 2\epsilon m^2 f_{\max}.
\]
\end{proof}

\begin{theorem}{\cite[Thm.~3]{jin2023selection}}
\label{thm:bh-fdr-joint-superuniform}
Let $(H_1, \dots, H_m) \in \{0,1\}^m$ be random variables indicating the status of $m$ hypotheses, where $H_j = 1$ means that the $j$th hypothesis is a true null.
Let $(p_1, \dots, p_m) \in [0,1]^m$ be random variables, interpreted as p-values for the corresponding hypotheses.
Assume:
\begin{itemize}
  \item The p-values $p_1, \dots, p_m$ are mutually independent.
  \item The hypothesis indicators $H_1, \dots, H_m$ are mutually independent.
  \item Each p-value $p_j$ is independent of $H_i$ for all $i \ne j$.
  \item For each $j \in \{1, \dots, m\}$, the p-value $p_j$ satisfies the joint super-uniformity condition:
  \[
  \mathbb{P}(p_j \le \alpha,\ H_j = 1) \le \alpha \qquad \text{for all } \alpha \in [0,1].
  \]
\end{itemize}

Then the Benjamini-Hochberg (BH) procedure at level $\alpha \in (0,1)$, applied to $(p_1, \dots, p_m)$, controls the false discovery rate:
\[
\mathrm{FDR} := \mathbb{E} \left[ \frac{|\mathcal{R} \cap \mathcal{H}_0|}{|\mathcal{R}| \vee 1} \right] \le \alpha,
\]
where:
\begin{itemize}
  \item $\mathcal{H}_0 := \{j : H_j = 1\}$ is the (random) set of true null hypotheses,
  \item $\mathcal{R} := \{j : p_j \le \widehat{\tau} \}$ is the BH rejection set,
  \item $\widehat{\tau} := \max \left\{ \frac{k \alpha}{m} : p_{(k)} \le \frac{k \alpha}{m} \right\}$ is the BH threshold, with $p_{(1)} \le \cdots \le p_{(m)}$ the order statistics of $(p_1, \dots, p_m)$.
\end{itemize}
Note: This result is a simplified version of the argument in~\cite{jin2023selection}, specialized to the case of independent p-values (rather than PRDS).
\end{theorem}

\begin{proof}[of Theorem~\ref{thm:bh-fdr-joint-superuniform}]
Let $\mathcal{R} := \{j : p_j \le \widehat{\tau} \}$ be the BH rejection set, and let $R_j := \mathbf{1}\{j \in \mathcal{R}\}$. Then the false discovery rate can be written as:
\[
\mathrm{FDR} = \mathbb{E} \left[ \frac{|\mathcal{R} \cap \mathcal{H}_0|}{|\mathcal{R}| \vee 1} \right] = \mathbb{E} \left[ \sum_{j=1}^m \frac{H_j R_j}{|\mathcal{R}| \vee 1} \right].
\]
Using the identity $R_j = \mathbf{1}\{p_j \le |\mathcal{R}|\alpha / m\}$ and decomposing over possible values of $|\mathcal{R}| = k$, we write:
\[
\mathrm{FDR} = \sum_{j=1}^m \sum_{k=1}^m \frac{1}{k} \, \mathbb{E} \left[ \mathbf{1}\{|\mathcal{R}| = k\} \cdot H_j \cdot \mathbf{1}\left\{p_j \le \tfrac{k \alpha}{m} \right\} \right].
\]
Now fix any index $j \in \{1, \dots, m\}$, and define the modified rejection set $\mathcal{R}(p_j \rightarrow 0)$ to be the BH rejection set obtained by replacing $p_j$ with 0 while keeping all other p-values fixed.
Because the BH procedure is monotone in each coordinate, and because reducing $p_j$ does not reduce the size of the rejection set, we have:
\[
\mathbf{1}\{|\mathcal{R}| = k\} \cdot \mathbf{1}\left\{p_j \le \tfrac{k \alpha}{m} \right\}
=
\mathbf{1}\{|\mathcal{R}(p_j \rightarrow 0)| = k\} \cdot \mathbf{1}\left\{p_j \le \tfrac{k \alpha}{m} \right\}.
\]
Let $\mathcal{F}_j := \sigma(p_1, \dots, p_{j-1}, 0, p_{j+1}, \dots, p_m)$ denote the sigma-algebra generated by the other p-values excluding $p_j$.
Then, using independence of $H_j$ and $p_j$ from $\mathcal{F}_j$, and applying the tower property:
\begin{align*}
\mathrm{FDR}
&= \sum_{j=1}^m \sum_{k=1}^m \frac{1}{k} \, \mathbb{E} \left[ \mathbf{1}\{|\mathcal{R}(p_j \rightarrow 0)| = k\} \cdot H_j \cdot \mathbf{1}\left\{p_j \le \tfrac{k \alpha}{m} \right\} \right] \\
&= \sum_{j=1}^m \sum_{k=1}^m \frac{1}{k} \, \mathbb{E} \left[ \mathbf{1}\{|\mathcal{R}(p_j \rightarrow 0)| = k\} \cdot \mathbb{E} \left[ H_j \cdot \mathbf{1}\left\{p_j \le \tfrac{k \alpha}{m} \right\} \mid \mathcal{F}_j \right] \right] \\
&= \sum_{j=1}^m \sum_{k=1}^m \frac{1}{k} \, \mathbb{E} \left[ \mathbf{1}\{|\mathcal{R}(p_j \rightarrow 0)| = k\} \cdot \mathbb{P}\left(p_j \le \tfrac{k \alpha}{m},\, H_j = 1 \right) \right] \\
&\le \sum_{j=1}^m \sum_{k=1}^m \frac{1}{k} \cdot \mathbb{E} \left[ \mathbf{1}\{|\mathcal{R}(p_j \rightarrow 0)| = k\} \cdot \tfrac{k \alpha}{m} \right] \\
&= \frac{\alpha}{m} \sum_{j=1}^m \sum_{k=1}^m \mathbb{P}\left( |\mathcal{R}(p_j \rightarrow 0)| = k \right) \le \alpha.
\end{align*}
\end{proof}

\FloatBarrier
\section{Additional Results from Experiments with Synthetic Data} \label{app:experiments-synthetic}

\subsection{Additional Details on Data Generating Distributions} \label{app:experiments-synthetic-main}

\begin{table}[H]
\centering
\caption{Summary of the four synthetic data generation settings considered for numerical experiments. Settings 1, 2, and 3 are adapted from \cite{sesia2024doubly}, with Setting 3 originally introduced in \cite{candes2023conformalized}. Setting 4 is new. In this table, $X$ denotes a $p$-dimensional vector of covariates for a random individual, and $X_j$ indicates its $j$-th element for any $j \in [p]$. Not all covariates are used to generate $Y$; most serve as noisy variables, making it more challenging to fit accurate survival and censoring models. In some experiments, we simplify the training process by excluding a subset of these noisy variables.}
\label{tab:distributions-synthetic}
\begin{tabular}{c c p{0.75\textwidth}}
\toprule
Setting & $p$ & Covariate, Survival, and Censoring Distributions \\ 
\midrule
\multirow{5}{*}{1} & \multirow{5}{*}{100} & $X$\quad\;\;: $\text{Unif}([0,1]^p)$ \\
                   &                      & $T \mid X$: LogNormal, $\mu_{\text{s}}(X) = (X_2 > \frac{1}{2}) + (X_3 < \frac{1}{2}) + (1-X_1)^{0.25}$,\newline{\color{white}$T \mid X$: LogNormal, }$\sigma_{\text{s}}(X) = \frac{1-X_1}{10}$ \\
                   &                      & $C \mid X$: LogNormal, $\mu_{\text{c}}(X) = (X_2 > \frac{1}{2}) + (X_3 < \frac{1}{2}) + (1-X_1)^4 + \frac{4}{10}$, \newline{\color{white}$C \mid X$: LogNormal, }$\sigma_{\text{c}}(X) = \frac{X_2}{10}$ \\
\midrule
\multirow{3}{*}{2} & \multirow{3}{*}{100} & $X$\quad\;\;: $\text{Unif}([0,1]^p)$ \\
                   &                      & $T \mid X$: LogNormal, $\mu_{\text{s}}(X) = X_1^{0.25}$, $\sigma_{\text{s}}(X) = 0.1$ \\
                   &                      & $C \mid X$: LogNormal, $\mu_{\text{c}}(X) = X_1^4 + \frac{4}{10}$, $\sigma_{\text{c}}(X) = 0.1$ \\
\midrule
\multirow{4}{*}{3} & \multirow{4}{*}{100} & $X$\quad\;\;: $\text{Unif}([-1,1]^p)$ \\
                   &                      & $T \mid X$: LogNormal, $\mu_{\text{s}}(X) = \log(2) + 1 + 0.55(X_1^2 - X_3X_5)$,\newline{\color{white}$T \mid X$: LogNormal, }$\sigma_{\text{s}}(X) = |X_{10}| + 1$ \\
                   &                      & $C \mid X$: Exponential, $\lambda_{\text{c}}(X) = 0.4$ \\
\midrule
\multirow{7}{*}{4} & \multirow{7}{*}{100} & $X$\quad\;\;: $\text{Unif}([-1,1]^p)$ \\
                   &                      & $T \mid X$: LogNormal, \newline{\color{white}$T \mid X$:}$\mu_{\text{s}}(X) = 0.2(1+X_1)X_2 + \log(2) \mathbf{1}(X_3 > 0) + \log(10) \mathbf{1}(X_3 \leq 0)$,\newline{\color{white}$T \mid X$:}$\sigma_{\text{s}}(X) = 0.25$ \\
                   &                      & $C \mid X$: LogNormal, $\mu_{\text{c}}(X) = 2 + 0.5 X_1$, $\sigma_{\text{c}}(X) = 0.1$ \\
                   &                      & Shifted $T \mid X$: LogNormal, \newline{\color{white}$T \mid X$:}$\mu_{\text{s}}(X) = 0.2(1+X_1)X_2 + \log(10)$, $\sigma_{\text{s}}(X) = 0.25$\\
\bottomrule
\end{tabular}
\end{table}

\FloatBarrier
\subsection{Effect of the Training Sample Size}

\begin{figure}[H]
    \centering
    \includegraphics[width=\textwidth]{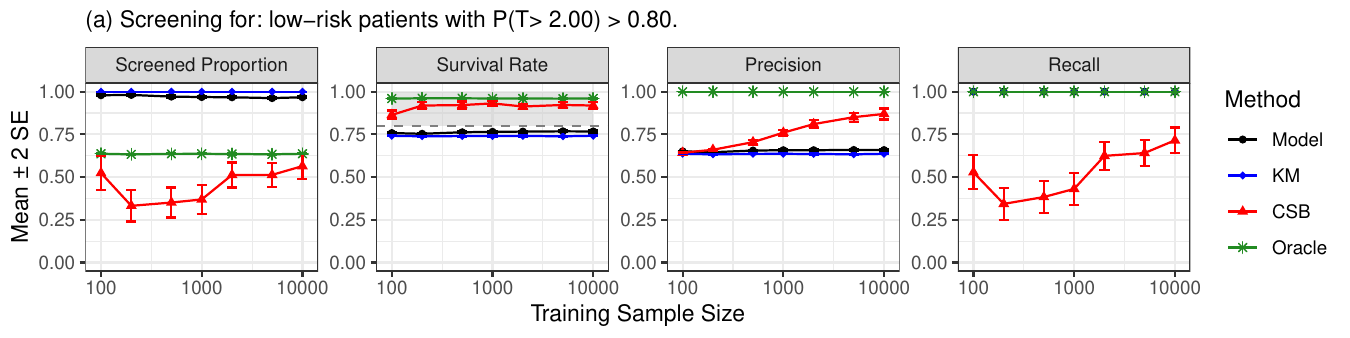}
    \includegraphics[width=\textwidth]{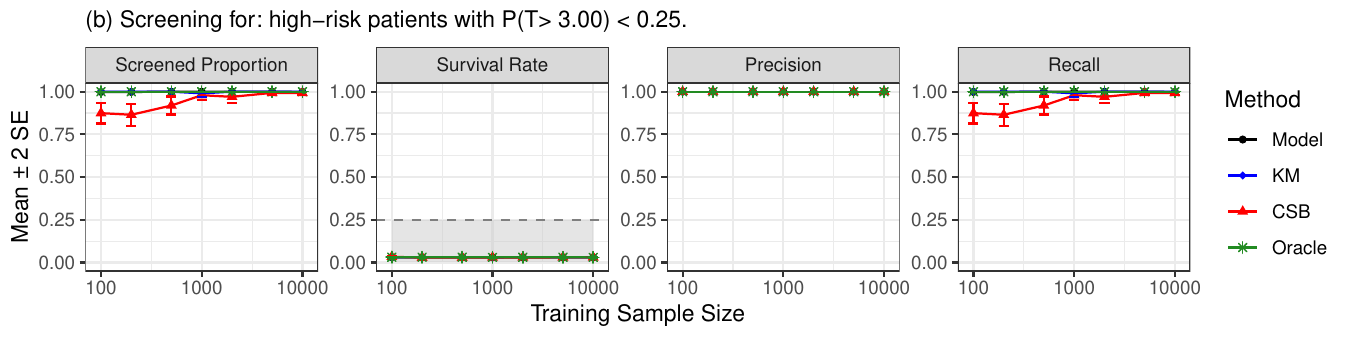}
    \caption{
Patient screening performance as a function of training sample size in a moderately challenging synthetic data scenario (Setting 2 from Table~\ref{tab:distributions-synthetic}).
Other details are as in Figure~\ref{fig:exp-setup1-setting1}.
}
    \label{fig:exp-setup1-setting2}
\end{figure}

\begin{figure}[H]
    \centering
    \includegraphics[width=\textwidth]{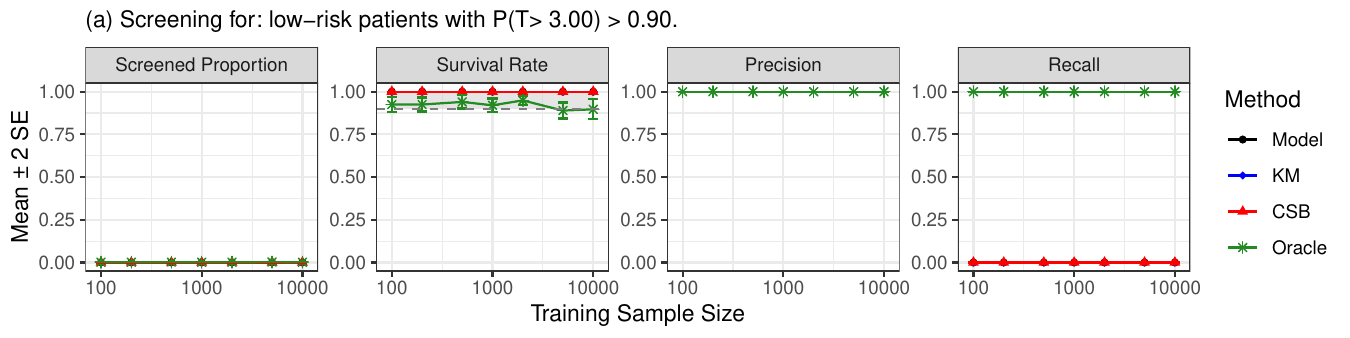}
    \includegraphics[width=\textwidth]{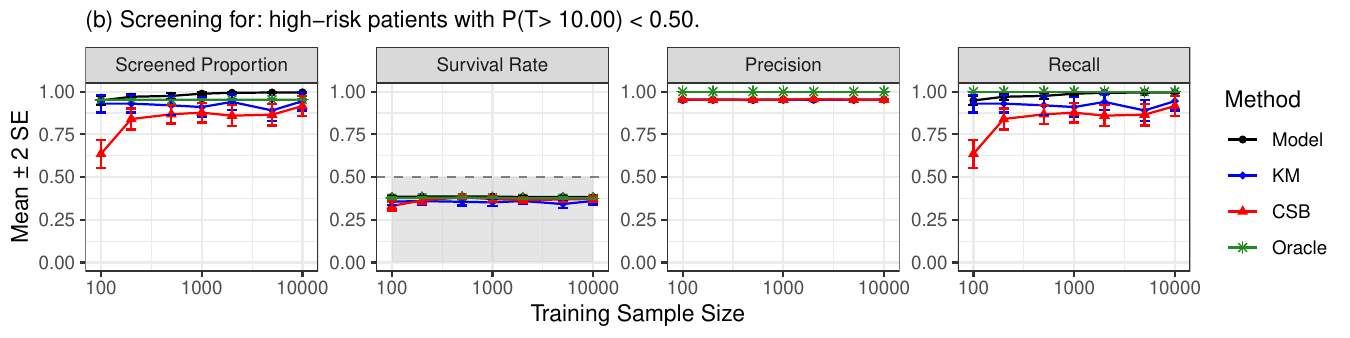}
    \caption{
Patient screening performance as a function of training sample size in a relatively easy synthetic data scenario (Setting 3 from Table~\ref{tab:distributions-synthetic}).
Other details are as in Figure~\ref{fig:exp-setup1-setting1}.
}
    \label{fig:exp-setup1-setting3}
\end{figure}

\FloatBarrier
\subsection{Effect of the Calibration Sample Size}

\begin{figure}[H]
    \centering
    \includegraphics[width=\textwidth]{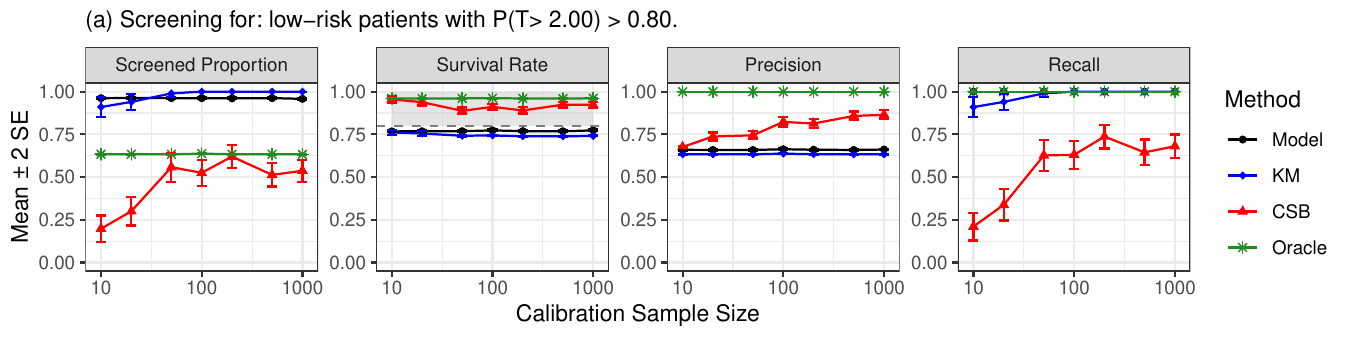}
    \includegraphics[width=\textwidth]{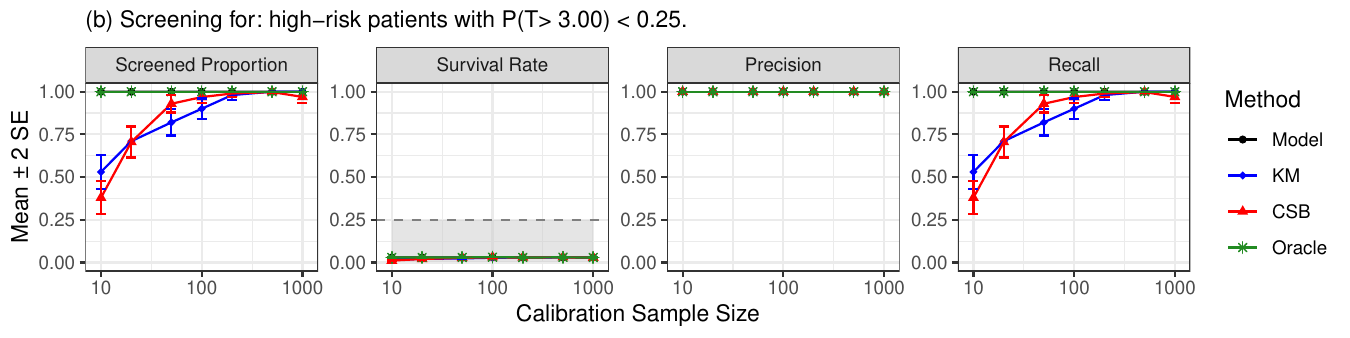}
    \caption{
Effect of calibration sample size on patient screening performance with conformal survival bands (CSB) in a moderately challenging synthetic data scenario (Setting 2 from Table~\ref{tab:distributions-synthetic}).
Other details are as in Figure~\ref{fig:exp-setup3-setting1}.
}
    \label{fig:exp-setup3-setting2}
\end{figure}

\begin{figure}[H]
    \centering
    \includegraphics[width=\textwidth]{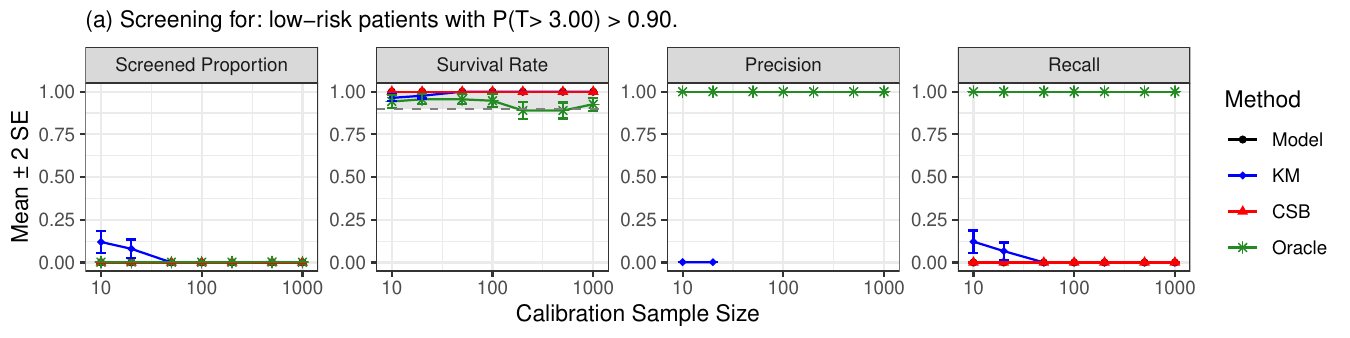}
    \includegraphics[width=\textwidth]{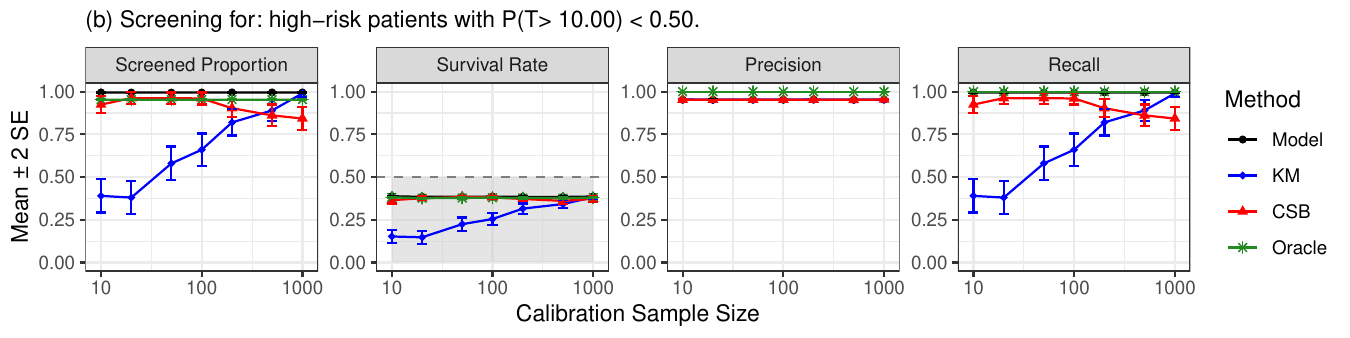}
    \caption{
Effect of calibration sample size on patient screening performance with conformal survival bands (CSB) in a relatively easy synthetic data scenario (Setting 3 from Table~\ref{tab:distributions-synthetic}).
Other details are as in Figure~\ref{fig:exp-setup3-setting1}.
}
    \label{fig:exp-setup3-setting3}
\end{figure}

\FloatBarrier
\subsection{Effect of the Training Sample Size for the Censoring Model}

\begin{figure}[H]
    \centering
    \includegraphics[width=\textwidth]{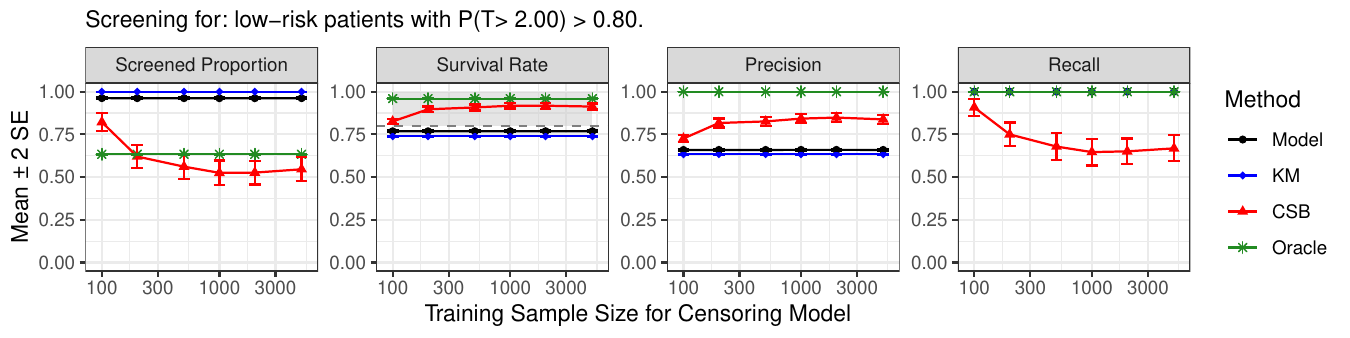}
    \includegraphics[width=\textwidth]{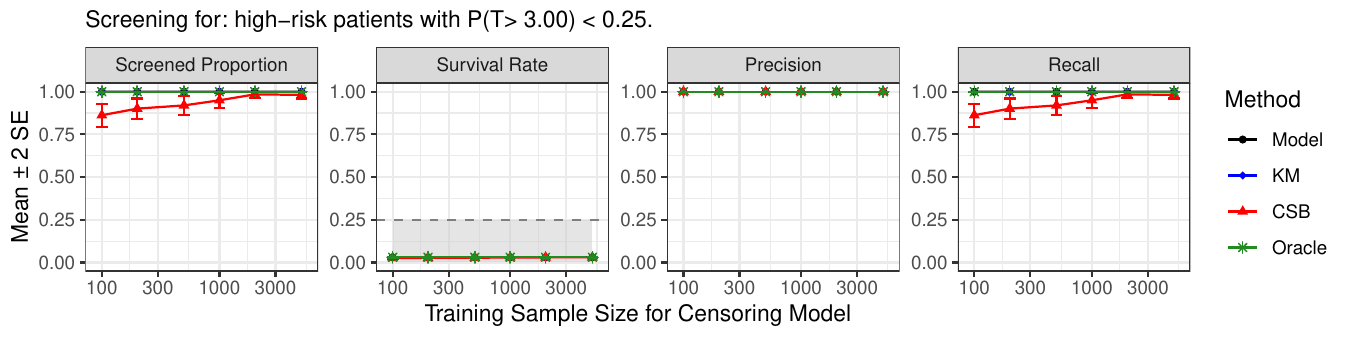}
\caption{
Effect of the training sample size used for fitting the censoring model on patient screening performance with conformal survival bands (CSB) in a moderately challenging synthetic data scenario (Setting 2 from Table~\ref{tab:distributions-synthetic}).
Other details are as in Figure~\ref{fig:exp-setup5-setting1}.
}
    \label{fig:exp-setup5-setting2}
\end{figure}

\begin{figure}[H]
    \centering
    \includegraphics[width=\textwidth]{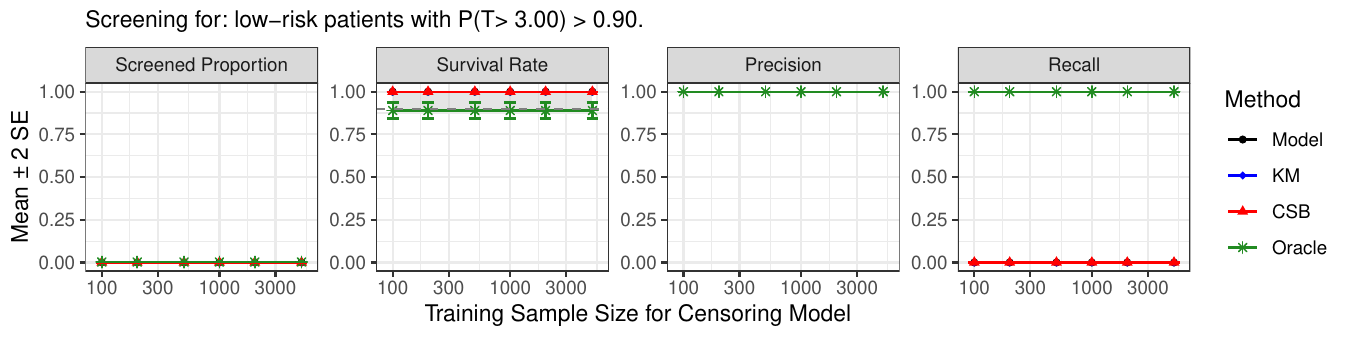}
    \includegraphics[width=\textwidth]{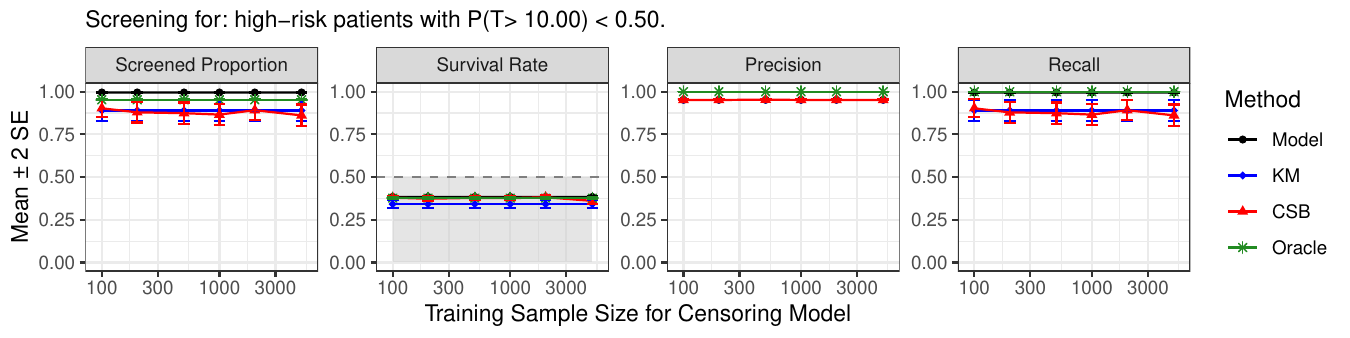}
\caption{
Effect of the training sample size used for fitting the censoring model on patient screening performance with conformal survival bands (CSB) in a relatively easy synthetic data scenario (Setting 3 from Table~\ref{tab:distributions-synthetic}).
Other details are as in Figure~\ref{fig:exp-setup5-setting1}.
}
    \label{fig:exp-setup5-setting3}
\end{figure}

\FloatBarrier
\subsection{Effect of the Number of Features for the Censoring Model}

\begin{figure}[H]
    \centering
    \includegraphics[width=\textwidth]{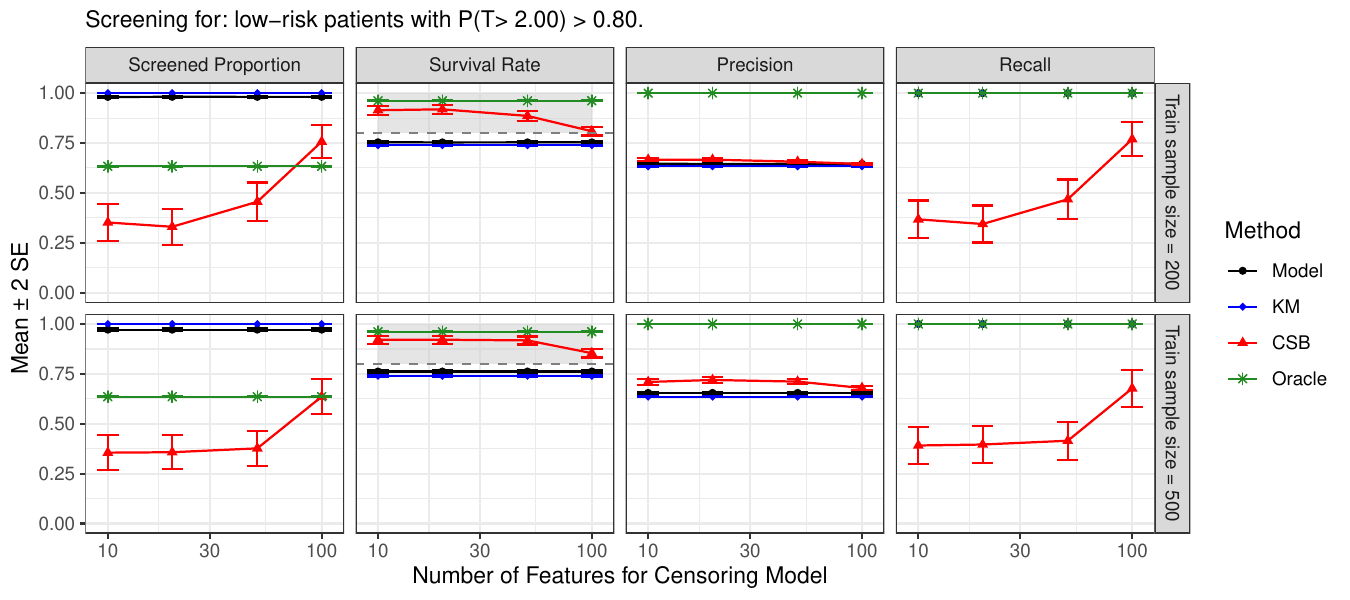}
    \includegraphics[width=\textwidth]{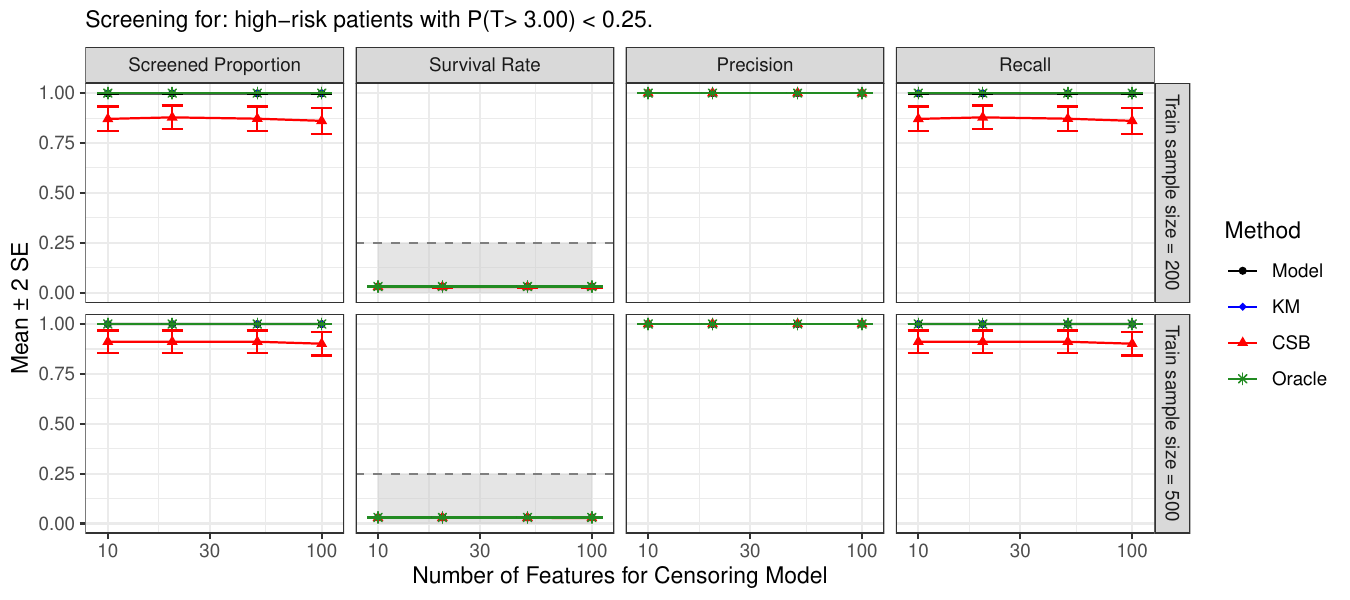}
\caption{
Impact of the number of features used in fitting the censoring model on patient screening performance with conformal survival bands (CSB) in a challenging synthetic data scenario (Setting 1 from Table~\ref{tab:distributions-synthetic}).
The total number of features is 100, but only a subset—including the relevant ones—is used to fit the censoring model.
Top: Low-risk screening. Screening performance is more sensitive to the number of features when the training sample size is small (200); using too many features degrades performance due to the difficulty of accurately estimating the censoring model. When the sample size is larger (500), performance is less sensitive to the number of features.
Bottom: High-risk screening (training sample sizes = 200 and 500). In this case, screening performance shows even less sensitivity to the number of features used to fit the censoring model.
}
    \label{fig:exp-setup6-setting1}
\end{figure}

\begin{figure}[H]
    \centering
    \includegraphics[width=\textwidth]{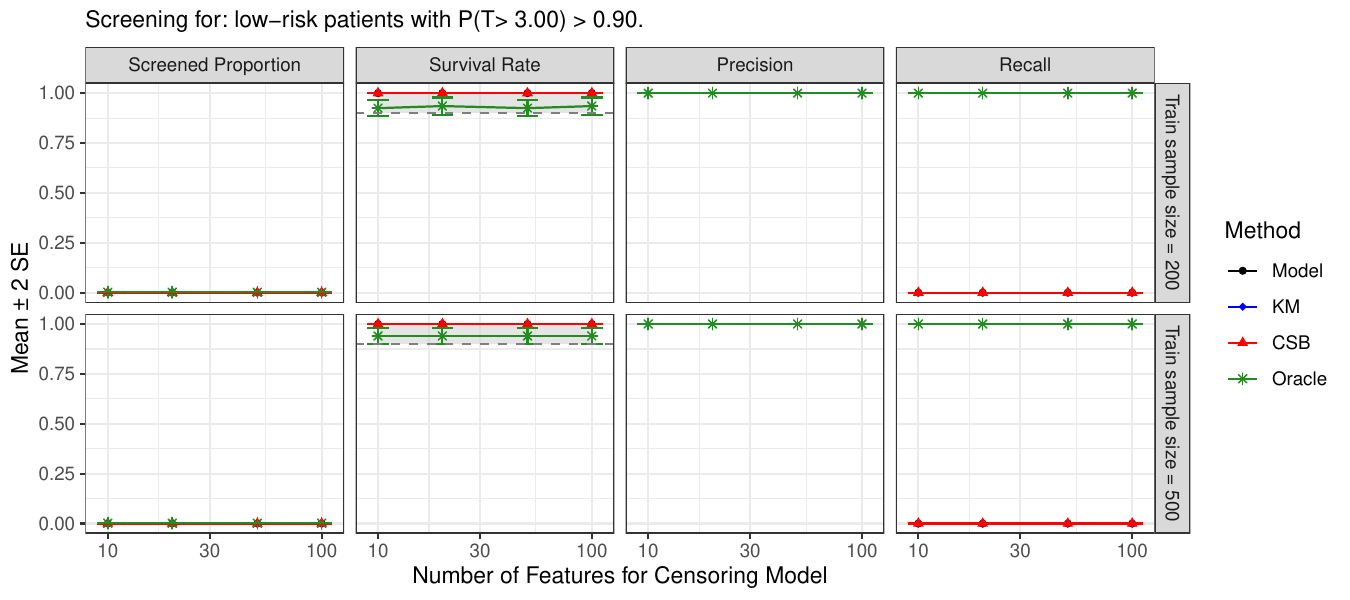}
    \includegraphics[width=\textwidth]{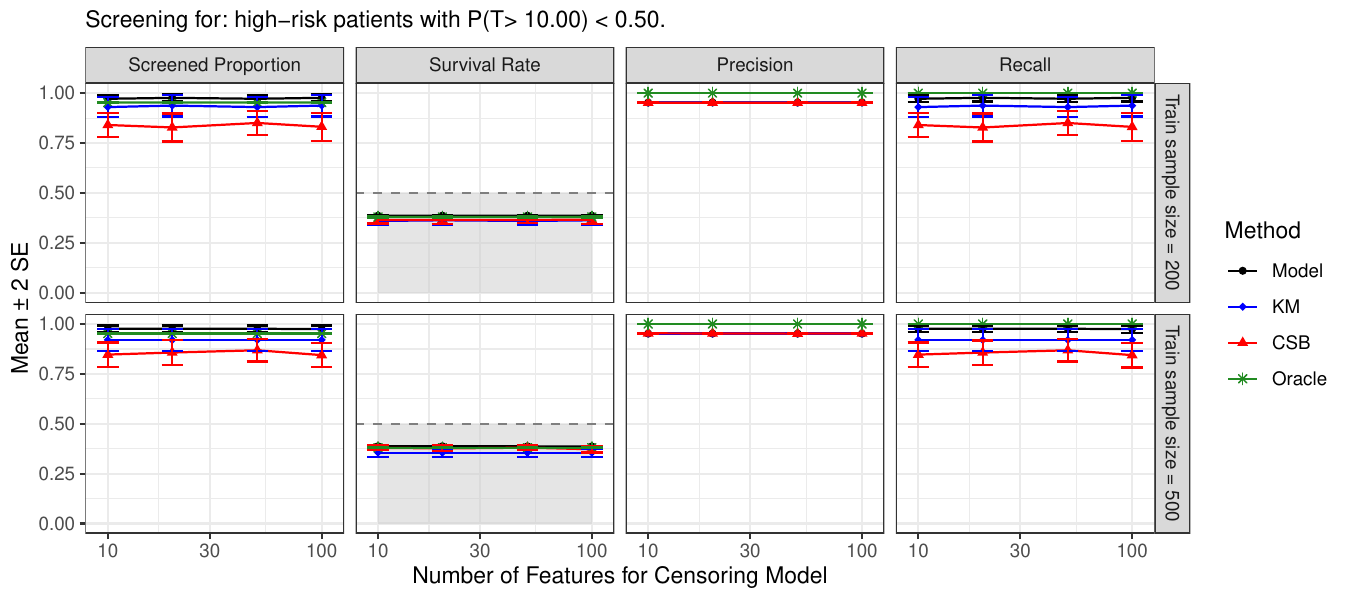}
\caption{
Impact of the number of features used in fitting the censoring model on patient screening performance with conformal survival bands (CSB) in an easier synthetic data scenario (Setting 3 from Table~\ref{tab:distributions-synthetic}).
Other details are as in Figure~\ref{fig:exp-setup6-setting1}.}
    \label{fig:exp-setup6-setting3}
\end{figure}

\FloatBarrier
\subsection{Effect of Different Censoring and Survival Models}

\begin{table}[H]
\centering
\caption{Performance of different methods for screening low-risk patients with $P(T > 6) > 0.80$ in a challenging synthetic scenario (Setting 1 from Table~\ref{tab:distributions-synthetic}), using various censoring and survival models.
The training sample size is 1000. Other details are as in Figure~\ref{fig:exp-setup1-setting1}.}
\label{tab:setup4-setting1-lr}

\begin{tabular}{lccccc}
\toprule
Survival Model & Method & Screened & Survival & Precision & Recall\\
\midrule
\addlinespace[0.3em]
\multicolumn{6}{l}{\textbf{Censoring Model: grf}}\\
\hspace{1em} & Model & 0.946 ± 0.008 & \textcolor{red}{0.570 ± 0.006} & 0.541 ± 0.006 & 0.999 ± 0.001\\

\hspace{1em} & KM & 1.000 ± 0.000 & \textcolor{red}{0.539 ± 0.003} & 0.511 ± 0.003 & 1.000 ± \vphantom{7} 0.000\\

\hspace{1em} & CSB & 0.221 ± 0.071 & \textcolor{black}{0.902 ± 0.033} & 0.685 ± 0.034 & 0.266 ± 0.079\\

\hspace{1em}\multirow[t]{-4}{*}{\raggedright\arraybackslash grf} & Oracle & 0.511 ± 0.003 & \textcolor{black}{0.974 ± 0.001} & 1.000 ± 0.000 & 1.000 ± \vphantom{1} 0.000\\

\hspace{1em} & Model & 0.995 ± 0.001 & \textcolor{red}{0.541 ± 0.003} & 0.514 ± 0.003 & 1.000 ± \vphantom{1} 0.000\\

\hspace{1em} & KM & 1.000 ± 0.000 & \textcolor{red}{0.539 ± 0.003} & 0.511 ± 0.003 & 1.000 ± \vphantom{6} 0.000\\

\hspace{1em} & CSB & 0.049 ± 0.043 & \textcolor{black}{0.977 ± 0.020} & 0.506 ± 0.003 & 0.050 ± 0.044\\

\hspace{1em}\multirow[t]{-4}{*}{\raggedright\arraybackslash survreg} & Oracle & 0.511 ± 0.003 & \textcolor{black}{0.974 ± 0.001} & 1.000 ± 0.000 & 1.000 ± \vphantom{1} 0.000\\

\hspace{1em} & Model & 0.814 ± 0.007 & \textcolor{red}{0.659 ± 0.006} & 0.626 ± 0.006 & 0.994 ± \vphantom{1} 0.001\\

\hspace{1em} & KM & 1.000 ± 0.000 & \textcolor{red}{0.539 ± 0.003} & 0.511 ± 0.003 & 1.000 ± \vphantom{5} 0.000\\

\hspace{1em} & CSB & 0.308 ± 0.061 & \textcolor{black}{0.857 ± 0.028} & 0.722 ± 0.019 & 0.413 ± 0.078\\

\hspace{1em}\multirow[t]{-4}{*}{\raggedright\arraybackslash rf} & Oracle & 0.511 ± 0.003 & \textcolor{black}{0.974 ± 0.001} & 1.000 ± 0.000 & 1.000 ± \vphantom{1} 0.000\\

\hspace{1em} & Model & 0.764 ± 0.005 & \textcolor{red}{0.652 ± 0.004} & 0.619 ± 0.004 & 0.925 ± \vphantom{1} 0.004\\

\hspace{1em} & KM & 1.000 ± 0.000 & \textcolor{red}{0.539 ± 0.003} & 0.511 ± 0.003 & 1.000 ± \vphantom{4} 0.000\\

\hspace{1em} & CSB & 0.303 ± 0.047 & \textcolor{black}{0.834 ± 0.022} & 0.794 ± 0.022 & 0.433 ± 0.057\\

\hspace{1em}\multirow[t]{-4}{*}{\raggedright\arraybackslash Cox} & Oracle & 0.511 ± 0.003 & \textcolor{black}{0.974 ± 0.001} & 1.000 ± 0.000 & 1.000 ± \vphantom{1} 0.000\\

\addlinespace[0.3em]
\multicolumn{6}{l}{\textbf{Censoring Model: Cox}}\\
\hspace{1em} & Model & 0.945 ± 0.008 & \textcolor{red}{0.571 ± 0.007} & 0.542 ± 0.007 & 0.999 ± 0.001\\

\hspace{1em} & KM & 1.000 ± 0.000 & \textcolor{red}{0.539 ± 0.003} & 0.511 ± 0.003 & 1.000 ± \vphantom{3} 0.000\\

\hspace{1em} & CSB & 0.688 ± 0.075 & \textcolor{red}{0.692 ± 0.036} & 0.619 ± 0.030 & 0.763 ± 0.071\\

\hspace{1em}\multirow[t]{-4}{*}{\raggedright\arraybackslash grf} & Oracle & 0.511 ± 0.003 & \textcolor{black}{0.974 ± 0.001} & 1.000 ± 0.000 & 1.000 ± 0.000\\

\hspace{1em} & Model & 0.995 ± 0.001 & \textcolor{red}{0.541 ± 0.003} & 0.514 ± 0.003 & 1.000 ± 0.000\\

\hspace{1em} & KM & 1.000 ± 0.000 & \textcolor{red}{0.539 ± 0.003} & 0.511 ± 0.003 & 1.000 ± \vphantom{2} 0.000\\

\hspace{1em} & CSB & 0.583 ± 0.098 & \textcolor{red}{0.733 ± 0.045} & 0.519 ± 0.003 & 0.590 ± 0.099\\

\hspace{1em}\multirow[t]{-4}{*}{\raggedright\arraybackslash survreg} & Oracle & 0.511 ± 0.003 & \textcolor{black}{0.974 ± 0.001} & 1.000 ± 0.000 & 1.000 ± 0.000\\

\hspace{1em} & Model & 0.814 ± 0.007 & \textcolor{red}{0.659 ± 0.006} & 0.626 ± 0.006 & 0.994 ± 0.001\\

\hspace{1em} & KM & 1.000 ± 0.000 & \textcolor{red}{0.539 ± 0.003} & 0.511 ± 0.003 & 1.000 ± \vphantom{1} 0.000\\

\hspace{1em} & CSB & 0.700 ± 0.045 & \textcolor{red}{0.692 ± 0.018} & 0.648 ± 0.014 & 0.864 ± 0.047\\

\hspace{1em}\multirow[t]{-4}{*}{\raggedright\arraybackslash rf} & Oracle & 0.511 ± 0.003 & \textcolor{black}{0.974 ± 0.001} & 1.000 ± 0.000 & 1.000 ± 0.000\\

\hspace{1em} & Model & 0.764 ± 0.005 & \textcolor{red}{0.652 ± 0.004} & 0.619 ± 0.004 & 0.925 ± 0.004\\

\hspace{1em} & KM & 1.000 ± 0.000 & \textcolor{red}{0.539 ± 0.003} & 0.511 ± 0.003 & 1.000 ± 0.000\\

\hspace{1em} & CSB & 0.661 ± 0.036 & \textcolor{red}{0.683 ± 0.013} & 0.654 ± 0.014 & 0.822 ± 0.035\\

\hspace{1em}\multirow[t]{-4}{*}{\raggedright\arraybackslash Cox} & Oracle & 0.511 ± 0.003 & \textcolor{black}{0.974 ± 0.001} & 1.000 ± 0.000 & 1.000 ± 0.000\\
\bottomrule
\end{tabular}

\end{table}

\begin{table}[H]
\centering
\caption{Performance of different methods for screening high-risk patients with $P(T > 12) < 0.80$ in a challenging synthetic data scenario (Setting 1 from Table~\ref{tab:distributions-synthetic}), using various censoring and survival models.
The training sample size is fixed at 1000. Other details are as in Figure~\ref{fig:exp-setup1-setting1}.}
\label{tab:setup4-setting1-hr}

\begin{tabular}{lccccc}
\toprule
Survival Model & Method & Screened & Survival & Precision & Recall\\
\midrule
\addlinespace[0.3em]
\multicolumn{6}{l}{\textbf{Censoring Model: grf}}\\
\hspace{1em} & Model & 0.755 ± 0.005 & \textcolor{black}{0.035 ± 0.002} & 0.965 ± 0.002 & 0.955 ± \vphantom{1} 0.006\\

\hspace{1em} & KM & 1.000 ± 0.000 & \textcolor{black}{0.238 ± 0.003} & 0.763 ± 0.003 & 1.000 ± \vphantom{7} 0.000\\

\hspace{1em} & CSB & 0.743 ± 0.022 & \textcolor{black}{0.040 ± 0.002} & 0.959 ± 0.002 & 0.934 ± 0.028\\

\hspace{1em}\multirow[t]{-4}{*}{\raggedright\arraybackslash grf} & Oracle & 0.763 ± 0.003 & \textcolor{black}{0.001 ± 0.000} & 1.000 ± 0.000 & 1.000 ± \vphantom{1} 0.000\\

\hspace{1em} & Model & 0.109 ± 0.004 & \textcolor{black}{0.000 ± 0.000} & 1.000 ± 0.000 & 0.142 ± \vphantom{1} 0.006\\

\hspace{1em} & KM & 1.000 ± 0.000 & \textcolor{black}{0.238 ± 0.003} & 0.763 ± 0.003 & 1.000 ± \vphantom{6} 0.000\\

\hspace{1em} & CSB & 0.137 ± 0.006 & \textcolor{black}{0.000 ± 0.000} & 1.000 ± 0.000 & 0.179 ± 0.008\\

\hspace{1em}\multirow[t]{-4}{*}{\raggedright\arraybackslash survreg} & Oracle & 0.763 ± 0.003 & \textcolor{black}{0.001 ± 0.000} & 1.000 ± 0.000 & 1.000 ± \vphantom{1} 0.000\\

\hspace{1em} & Model & 0.819 ± 0.004 & \textcolor{black}{0.080 ± 0.003} & 0.921 ± 0.003 & 0.988 ± \vphantom{1} 0.001\\

\hspace{1em} & KM & 1.000 ± 0.000 & \textcolor{black}{0.238 ± 0.003} & 0.763 ± 0.003 & 1.000 ± \vphantom{5} 0.000\\

\hspace{1em} & CSB & 0.817 ± 0.024 & \textcolor{black}{0.094 ± 0.005} & 0.904 ± 0.005 & 0.969 ± 0.028\\

\hspace{1em}\multirow[t]{-4}{*}{\raggedright\arraybackslash rf} & Oracle & 0.763 ± 0.003 & \textcolor{black}{0.001 ± 0.000} & 1.000 ± 0.000 & 1.000 ± \vphantom{1} 0.000\\

\hspace{1em} & Model & 0.713 ± 0.005 & \textcolor{black}{0.071 ± 0.003} & 0.929 ± 0.003 & 0.868 ± \vphantom{1} 0.005\\

\hspace{1em} & KM & 1.000 ± 0.000 & \textcolor{black}{0.238 ± 0.003} & 0.763 ± 0.003 & 1.000 ± \vphantom{4} 0.000\\

\hspace{1em} & CSB & 0.680 ± 0.015 & \textcolor{black}{0.073 ± 0.003} & 0.928 ± 0.003 & 0.826 ± 0.018\\

\hspace{1em}\multirow[t]{-4}{*}{\raggedright\arraybackslash Cox} & Oracle & 0.763 ± 0.003 & \textcolor{black}{0.001 ± 0.000} & 1.000 ± 0.000 & 1.000 ± \vphantom{1} 0.000\\

\addlinespace[0.3em]
\multicolumn{6}{l}{\textbf{Censoring Model: Cox}}\\
\hspace{1em} & Model & 0.755 ± 0.005 & \textcolor{black}{0.035 ± 0.002} & 0.965 ± 0.002 & 0.955 ± 0.006\\

\hspace{1em} & KM & 1.000 ± 0.000 & \textcolor{black}{0.238 ± 0.003} & 0.763 ± 0.003 & 1.000 ± \vphantom{3} 0.000\\

\hspace{1em} & CSB & 0.718 ± 0.034 & \textcolor{black}{0.039 ± 0.003} & 0.959 ± 0.002 & 0.903 ± 0.043\\

\hspace{1em}\multirow[t]{-4}{*}{\raggedright\arraybackslash grf} & Oracle & 0.763 ± 0.003 & \textcolor{black}{0.001 ± 0.000} & 1.000 ± 0.000 & 1.000 ± 0.000\\

\hspace{1em} & Model & 0.109 ± 0.004 & \textcolor{black}{0.000 ± 0.000} & 1.000 ± 0.000 & 0.142 ± 0.006\\

\hspace{1em} & KM & 1.000 ± 0.000 & \textcolor{black}{0.238 ± 0.003} & 0.763 ± 0.003 & 1.000 ± \vphantom{2} 0.000\\

\hspace{1em} & CSB & 0.134 ± 0.007 & \textcolor{black}{0.000 ± 0.000} & 1.000 ± 0.000 & 0.175 ± 0.009\\

\hspace{1em}\multirow[t]{-4}{*}{\raggedright\arraybackslash survreg} & Oracle & 0.763 ± 0.003 & \textcolor{black}{0.001 ± 0.000} & 1.000 ± 0.000 & 1.000 ± 0.000\\

\hspace{1em} & Model & 0.819 ± 0.004 & \textcolor{black}{0.080 ± 0.003} & 0.921 ± 0.003 & 0.988 ± 0.001\\

\hspace{1em} & KM & 1.000 ± 0.000 & \textcolor{black}{0.238 ± 0.003} & 0.763 ± 0.003 & 1.000 ± \vphantom{1} 0.000\\

\hspace{1em} & CSB & 0.790 ± 0.036 & \textcolor{black}{0.089 ± 0.007} & 0.909 ± 0.006 & 0.938 ± 0.042\\

\hspace{1em}\multirow[t]{-4}{*}{\raggedright\arraybackslash rf} & Oracle & 0.763 ± 0.003 & \textcolor{black}{0.001 ± 0.000} & 1.000 ± 0.000 & 1.000 ± 0.000\\

\hspace{1em} & Model & 0.713 ± 0.005 & \textcolor{black}{0.071 ± 0.003} & 0.929 ± 0.003 & 0.868 ± 0.005\\

\hspace{1em} & KM & 1.000 ± 0.000 & \textcolor{black}{0.238 ± 0.003} & 0.763 ± 0.003 & 1.000 ± 0.000\\

\hspace{1em} & CSB & 0.660 ± 0.026 & \textcolor{black}{0.070 ± 0.004} & 0.929 ± 0.004 & 0.802 ± 0.031\\

\hspace{1em}\multirow[t]{-4}{*}{\raggedright\arraybackslash Cox} & Oracle & 0.763 ± 0.003 & \textcolor{black}{0.001 ± 0.000} & 1.000 ± 0.000 & 1.000 ± 0.000\\
\bottomrule
\end{tabular}

\end{table}

\begin{table}[H]
\centering
\caption{Performance of different methods for screening low-risk patients with $P(T > 2) > 0.80$ in a relatively challenging data scenario (Setting 2 from Table~\ref{tab:distributions-synthetic}), using various censoring and survival models.
The training sample size is fixed at 1000. Other details are as in Figure~\ref{fig:exp-setup1-setting1}.}
\label{tab:setup4-setting2-lr}

\begin{tabular}{lccccc}
\toprule
Survival Model & Method & Screened & Survival & Precision & Recall\\
\midrule
\addlinespace[0.3em]
\multicolumn{6}{l}{\textbf{Censoring Model: grf}}\\
\hspace{1em} & Model & 0.968 ± 0.005 & \textcolor{red}{0.765 ± 0.005} & 0.658 ± 0.004 & 1.000 ± 0.000\\

\hspace{1em} & KM & 1.000 ± 0.000 & \textcolor{red}{0.741 ± 0.003} & 0.636 ± 0.002 & 1.000 ± \vphantom{7} 0.000\\

\hspace{1em} & CSB & 0.362 ± 0.082 & \textcolor{black}{0.932 ± 0.017} & 0.762 ± 0.018 & 0.422 ± 0.094\\

\hspace{1em}\multirow[t]{-4}{*}{\raggedright\arraybackslash grf} & Oracle & 0.636 ± 0.002 & \textcolor{black}{0.960 ± 0.001} & 1.000 ± 0.000 & 1.000 ± \vphantom{1} 0.000\\

\hspace{1em} & Model & 0.951 ± 0.003 & \textcolor{red}{0.776 ± 0.004} & 0.669 ± 0.003 & 1.000 ± \vphantom{1} 0.000\\

\hspace{1em} & KM & 1.000 ± 0.000 & \textcolor{red}{0.741 ± 0.003} & 0.636 ± 0.002 & 1.000 ± \vphantom{6} 0.000\\

\hspace{1em} & CSB & 0.055 ± 0.044 & \textcolor{black}{0.987 ± 0.010} & 0.682 ± 0.004 & 0.060 ± 0.048\\

\hspace{1em}\multirow[t]{-4}{*}{\raggedright\arraybackslash survreg} & Oracle & 0.636 ± 0.002 & \textcolor{black}{0.960 ± 0.001} & 1.000 ± 0.000 & 1.000 ± \vphantom{1} 0.000\\

\hspace{1em} & Model & 0.880 ± 0.005 & \textcolor{black}{0.834 ± 0.004} & 0.723 ± 0.005 & 1.000 ± \vphantom{1} 0.000\\

\hspace{1em} & KM & 1.000 ± 0.000 & \textcolor{red}{0.741 ± 0.003} & 0.636 ± 0.002 & 1.000 ± \vphantom{5} 0.000\\

\hspace{1em} & CSB & 0.556 ± 0.066 & \textcolor{black}{0.908 ± 0.011} & 0.822 ± 0.017 & 0.693 ± 0.077\\

\hspace{1em}\multirow[t]{-4}{*}{\raggedright\arraybackslash rf} & Oracle & 0.636 ± 0.002 & \textcolor{black}{0.960 ± 0.001} & 1.000 ± 0.000 & 1.000 ± \vphantom{1} 0.000\\

\hspace{1em} & Model & 0.596 ± 0.011 & \textcolor{black}{0.961 ± 0.003} & 0.967 ± 0.006 & 0.904 ± \vphantom{1} 0.011\\

\hspace{1em} & KM & 1.000 ± 0.000 & \textcolor{red}{0.741 ± 0.003} & 0.636 ± 0.002 & 1.000 ± \vphantom{4} 0.000\\

\hspace{1em} & CSB & 0.496 ± 0.030 & \textcolor{black}{0.969 ± 0.006} & 0.969 ± 0.010 & 0.747 ± 0.039\\

\hspace{1em}\multirow[t]{-4}{*}{\raggedright\arraybackslash Cox} & Oracle & 0.636 ± 0.002 & \textcolor{black}{0.960 ± 0.001} & 1.000 ± 0.000 & 1.000 ± \vphantom{1} 0.000\\

\addlinespace[0.3em]
\multicolumn{6}{l}{\textbf{Censoring Model: Cox}}\\
\hspace{1em} & Model & 0.969 ± 0.006 & \textcolor{red}{0.765 ± 0.005} & 0.657 ± 0.004 & 1.000 ± 0.000\\

\hspace{1em} & KM & 1.000 ± 0.000 & \textcolor{red}{0.741 ± 0.003} & 0.636 ± 0.002 & 1.000 ± \vphantom{3} 0.000\\

\hspace{1em} & CSB & 0.599 ± 0.072 & \textcolor{black}{0.889 ± 0.015} & 0.766 ± 0.015 & 0.704 ± 0.081\\

\hspace{1em}\multirow[t]{-4}{*}{\raggedright\arraybackslash grf} & Oracle & 0.636 ± 0.002 & \textcolor{black}{0.960 ± 0.001} & 1.000 ± 0.000 & 1.000 ± 0.000\\

\hspace{1em} & Model & 0.951 ± 0.003 & \textcolor{red}{0.776 ± 0.004} & 0.669 ± 0.003 & 1.000 ± 0.000\\

\hspace{1em} & KM & 1.000 ± 0.000 & \textcolor{red}{0.741 ± 0.003} & 0.636 ± 0.002 & 1.000 ± \vphantom{2} 0.000\\

\hspace{1em} & CSB & 0.000 ± 0.000 & NA & NA & 0.000 ± 0.000\\

\hspace{1em}\multirow[t]{-4}{*}{\raggedright\arraybackslash survreg} & Oracle & 0.636 ± 0.002 & \textcolor{black}{0.960 ± 0.001} & 1.000 ± 0.000 & 1.000 ± 0.000\\

\hspace{1em} & Model & 0.880 ± 0.005 & \textcolor{black}{0.834 ± 0.004} & 0.723 ± 0.005 & 1.000 ± 0.000\\

\hspace{1em} & KM & 1.000 ± 0.000 & \textcolor{red}{0.741 ± 0.003} & 0.636 ± 0.002 & 1.000 ± \vphantom{1} 0.000\\

\hspace{1em} & CSB & 0.776 ± 0.037 & \textcolor{black}{0.870 ± 0.009} & 0.773 ± 0.013 & 0.929 ± 0.040\\

\hspace{1em}\multirow[t]{-4}{*}{\raggedright\arraybackslash rf} & Oracle & 0.636 ± 0.002 & \textcolor{black}{0.960 ± 0.001} & 1.000 ± 0.000 & 1.000 ± 0.000\\

\hspace{1em} & Model & 0.596 ± 0.011 & \textcolor{black}{0.961 ± 0.003} & 0.967 ± 0.006 & 0.904 ± 0.011\\

\hspace{1em} & KM & 1.000 ± 0.000 & \textcolor{red}{0.741 ± 0.003} & 0.636 ± 0.002 & 1.000 ± 0.000\\

\hspace{1em} & CSB & 0.569 ± 0.021 & \textcolor{black}{0.959 ± 0.005} & 0.955 ± 0.011 & 0.846 ± 0.023\\

\hspace{1em}\multirow[t]{-4}{*}{\raggedright\arraybackslash Cox} & Oracle & 0.636 ± 0.002 & \textcolor{black}{0.960 ± 0.001} & 1.000 ± 0.000 & 1.000 ± 0.000\\
\bottomrule
\end{tabular}

\end{table}

\begin{table}[H]
\centering
\caption{Performance of different methods for screening high-risk patients with $P(T > 3) < 0.25$ in a relatively challenging data scenario (Setting 2 from Table~\ref{tab:distributions-synthetic}), using various censoring and survival models.
The training sample size is fixed at 1000. Other details are as in Figure~\ref{fig:exp-setup1-setting1}.}
\label{tab:setup4-setting2-hr}

\begin{tabular}{lccccc}
\toprule
Survival Model & Method & Screened & Survival & Precision & Recall\\
\midrule
\addlinespace[0.3em]
\multicolumn{6}{l}{\textbf{Censoring Model: grf}}\\
\hspace{1em} & Model & 0.999 ± 0.000 & \textcolor{black}{0.030 ± 0.001} & 1.000 ± 0.000 & 0.999 ± \vphantom{1} 0.000\\

\hspace{1em} & KM & 0.990 ± 0.020 & \textcolor{black}{0.030 ± 0.001} & 1.000 ± 0.000 & 0.990 ± \vphantom{7} 0.020\\

\hspace{1em} & CSB & 0.979 ± 0.028 & \textcolor{black}{0.029 ± 0.001} & 1.000 ± 0.000 & 0.979 ± 0.028\\

\hspace{1em}\multirow[t]{-4}{*}{\raggedright\arraybackslash grf} & Oracle & 1.000 ± 0.000 & \textcolor{black}{0.030 ± 0.001} & 1.000 ± 0.000 & 1.000 ± \vphantom{1} 0.000\\

\hspace{1em} & Model & 0.726 ± 0.006 & \textcolor{black}{0.008 ± 0.001} & 1.000 ± 0.000 & 0.726 ± \vphantom{1} 0.006\\

\hspace{1em} & KM & 0.990 ± 0.020 & \textcolor{black}{0.030 ± 0.001} & 1.000 ± 0.000 & 0.990 ± \vphantom{6} 0.020\\

\hspace{1em} & CSB & 0.738 ± 0.022 & \textcolor{black}{0.010 ± 0.001} & 1.000 ± 0.000 & 0.738 ± 0.022\\

\hspace{1em}\multirow[t]{-4}{*}{\raggedright\arraybackslash survreg} & Oracle & 1.000 ± 0.000 & \textcolor{black}{0.030 ± 0.001} & 1.000 ± 0.000 & 1.000 ± \vphantom{1} 0.000\\

\hspace{1em} & Model & 0.907 ± 0.014 & \textcolor{black}{0.033 ± 0.001} & 1.000 ± 0.000 & 0.907 ± \vphantom{1} 0.014\\

\hspace{1em} & KM & 0.990 ± 0.020 & \textcolor{black}{0.030 ± 0.001} & 1.000 ± 0.000 & 0.990 ± \vphantom{5} 0.020\\

\hspace{1em} & CSB & 0.898 ± 0.030 & \textcolor{black}{0.031 ± 0.002} & 1.000 ± 0.000 & 0.898 ± 0.030\\

\hspace{1em}\multirow[t]{-4}{*}{\raggedright\arraybackslash rf} & Oracle & 1.000 ± 0.000 & \textcolor{black}{0.030 ± 0.001} & 1.000 ± 0.000 & 1.000 ± \vphantom{1} 0.000\\

\hspace{1em} & Model & 0.964 ± 0.002 & \textcolor{black}{0.026 ± 0.001} & 1.000 ± 0.000 & 0.964 ± \vphantom{1} 0.002\\

\hspace{1em} & KM & 0.990 ± 0.020 & \textcolor{black}{0.030 ± 0.001} & 1.000 ± 0.000 & 0.990 ± \vphantom{4} 0.020\\

\hspace{1em} & CSB & 0.951 ± 0.019 & \textcolor{black}{0.025 ± 0.001} & 1.000 ± 0.000 & 0.951 ± 0.019\\

\hspace{1em}\multirow[t]{-4}{*}{\raggedright\arraybackslash Cox} & Oracle & 1.000 ± 0.000 & \textcolor{black}{0.030 ± 0.001} & 1.000 ± 0.000 & 1.000 ± \vphantom{1} 0.000\\

\addlinespace[0.3em]
\multicolumn{6}{l}{\textbf{Censoring Model: Cox}}\\
\hspace{1em} & Model & 0.999 ± 0.000 & \textcolor{black}{0.030 ± 0.001} & 1.000 ± 0.000 & 0.999 ± 0.000\\

\hspace{1em} & KM & 0.990 ± 0.020 & \textcolor{black}{0.030 ± 0.001} & 1.000 ± 0.000 & 0.990 ± \vphantom{3} 0.020\\

\hspace{1em} & CSB & 0.999 ± 0.000 & \textcolor{black}{0.030 ± 0.001} & 1.000 ± 0.000 & 0.999 ± 0.000\\

\hspace{1em}\multirow[t]{-4}{*}{\raggedright\arraybackslash grf} & Oracle & 1.000 ± 0.000 & \textcolor{black}{0.030 ± 0.001} & 1.000 ± 0.000 & 1.000 ± 0.000\\

\hspace{1em} & Model & 0.726 ± 0.006 & \textcolor{black}{0.008 ± 0.001} & 1.000 ± 0.000 & 0.726 ± 0.006\\

\hspace{1em} & KM & 0.990 ± 0.020 & \textcolor{black}{0.030 ± 0.001} & 1.000 ± 0.000 & 0.990 ± \vphantom{2} 0.020\\

\hspace{1em} & CSB & 0.753 ± 0.006 & \textcolor{black}{0.010 ± 0.001} & 1.000 ± 0.000 & 0.753 ± 0.006\\

\hspace{1em}\multirow[t]{-4}{*}{\raggedright\arraybackslash survreg} & Oracle & 1.000 ± 0.000 & \textcolor{black}{0.030 ± 0.001} & 1.000 ± 0.000 & 1.000 ± 0.000\\

\hspace{1em} & Model & 0.907 ± 0.014 & \textcolor{black}{0.033 ± 0.001} & 1.000 ± 0.000 & 0.907 ± 0.014\\

\hspace{1em} & KM & 0.990 ± 0.020 & \textcolor{black}{0.030 ± 0.001} & 1.000 ± 0.000 & 0.990 ± \vphantom{1} 0.020\\

\hspace{1em} & CSB & 0.918 ± 0.014 & \textcolor{black}{0.032 ± 0.001} & 1.000 ± 0.000 & 0.918 ± 0.014\\

\hspace{1em}\multirow[t]{-4}{*}{\raggedright\arraybackslash rf} & Oracle & 1.000 ± 0.000 & \textcolor{black}{0.030 ± 0.001} & 1.000 ± 0.000 & 1.000 ± 0.000\\

\hspace{1em} & Model & 0.964 ± 0.002 & \textcolor{black}{0.026 ± 0.001} & 1.000 ± 0.000 & 0.964 ± 0.002\\

\hspace{1em} & KM & 0.990 ± 0.020 & \textcolor{black}{0.030 ± 0.001} & 1.000 ± 0.000 & 0.990 ± 0.020\\

\hspace{1em} & CSB & 0.962 ± 0.002 & \textcolor{black}{0.026 ± 0.001} & 1.000 ± 0.000 & 0.962 ± 0.002\\

\hspace{1em}\multirow[t]{-4}{*}{\raggedright\arraybackslash Cox} & Oracle & 1.000 ± 0.000 & \textcolor{black}{0.030 ± 0.001} & 1.000 ± 0.000 & 1.000 ± 0.000\\
\bottomrule
\end{tabular}

\end{table}

\begin{table}[H]
\centering
\caption{Performance of different methods for screening low-risk patients with $P(T > 3) > 0.90$ in an easier synthetic data scenario (Setting 3 from Table~\ref{tab:distributions-synthetic}), using various censoring and survival models.
The training sample size is fixed at 1000. Other details are as in Figure~\ref{fig:exp-setup1-setting1}.}
\label{tab:setup4-setting3-lr}

\begin{tabular}{lccccc}
\toprule
Survival Model & Method & Screened & Survival & Precision & Recall\\
\midrule
\addlinespace[0.3em]
\multicolumn{6}{l}{\textbf{Censoring Model: grf}}\\
\hspace{1em} & Model & 0.000 ± 0.000 & NA & NA & 0.000 ± \vphantom{3} 0.000\\

\hspace{1em} & KM & 0.000 ± 0.000 & NA & NA & 0.000 ± \vphantom{7} 0.000\\

\hspace{1em} & CSB & 0.000 ± 0.000 & NA & NA & 0.000 ± \vphantom{7} 0.000\\

\hspace{1em}\multirow[t]{-4}{*}{\raggedright\arraybackslash grf} & Oracle & 0.002 ± 0.000 & \textcolor{black}{0.920 ± 0.039} & 1.000 ± 0.000 & 1.000 ± \vphantom{1} 0.000\\

\hspace{1em} & Model & 1.000 ± 0.000 & \textcolor{red}{0.701 ± 0.002} & 0.002 ± 0.000 & 1.000 ± \vphantom{1} 0.000\\

\hspace{1em} & KM & 0.000 ± 0.000 & NA & NA & 0.000 ± \vphantom{6} 0.000\\

\hspace{1em} & CSB & 0.000 ± 0.000 & NA & NA & 0.000 ± \vphantom{6} 0.000\\

\hspace{1em}\multirow[t]{-4}{*}{\raggedright\arraybackslash survreg} & Oracle & 0.002 ± 0.000 & \textcolor{black}{0.920 ± 0.039} & 1.000 ± 0.000 & 1.000 ± \vphantom{1} 0.000\\

\hspace{1em} & Model & 0.000 ± 0.000 & NA & NA & 0.000 ± \vphantom{2} 0.000\\

\hspace{1em} & KM & 0.000 ± 0.000 & NA & NA & 0.000 ± \vphantom{5} 0.000\\

\hspace{1em} & CSB & 0.000 ± 0.000 & NA & NA & 0.000 ± \vphantom{5} 0.000\\

\hspace{1em}\multirow[t]{-4}{*}{\raggedright\arraybackslash rf} & Oracle & 0.002 ± 0.000 & \textcolor{black}{0.920 ± 0.039} & 1.000 ± 0.000 & 1.000 ± \vphantom{1} 0.000\\

\hspace{1em} & Model & 0.046 ± 0.004 & \textcolor{red}{0.702 ± 0.016} & 0.003 ± 0.003 & 0.060 ± \vphantom{1} 0.038\\

\hspace{1em} & KM & 0.000 ± 0.000 & NA & NA & 0.000 ± \vphantom{4} 0.000\\

\hspace{1em} & CSB & 0.000 ± 0.000 & NA & NA & 0.000 ± \vphantom{4} 0.000\\

\hspace{1em}\multirow[t]{-4}{*}{\raggedright\arraybackslash Cox} & Oracle & 0.002 ± 0.000 & \textcolor{black}{0.920 ± 0.039} & 1.000 ± 0.000 & 1.000 ± \vphantom{1} 0.000\\

\addlinespace[0.3em]
\multicolumn{6}{l}{\textbf{Censoring Model: Cox}}\\
\hspace{1em} & Model & 0.000 ± 0.000 & NA & NA & 0.000 ± \vphantom{1} 0.000\\

\hspace{1em} & KM & 0.000 ± 0.000 & NA & NA & 0.000 ± \vphantom{3} 0.000\\

\hspace{1em} & CSB & 0.000 ± 0.000 & NA & NA & 0.000 ± \vphantom{3} 0.000\\

\hspace{1em}\multirow[t]{-4}{*}{\raggedright\arraybackslash grf} & Oracle & 0.002 ± 0.000 & \textcolor{black}{0.920 ± 0.039} & 1.000 ± 0.000 & 1.000 ± 0.000\\

\hspace{1em} & Model & 1.000 ± 0.000 & \textcolor{red}{0.701 ± 0.002} & 0.002 ± 0.000 & 1.000 ± 0.000\\

\hspace{1em} & KM & 0.000 ± 0.000 & NA & NA & 0.000 ± \vphantom{2} 0.000\\

\hspace{1em} & CSB & 0.000 ± 0.000 & NA & NA & 0.000 ± \vphantom{2} 0.000\\

\hspace{1em}\multirow[t]{-4}{*}{\raggedright\arraybackslash survreg} & Oracle & 0.002 ± 0.000 & \textcolor{black}{0.920 ± 0.039} & 1.000 ± 0.000 & 1.000 ± 0.000\\

\hspace{1em} & Model & 0.000 ± 0.000 & NA & NA & 0.000 ± 0.000\\

\hspace{1em} & KM & 0.000 ± 0.000 & NA & NA & 0.000 ± \vphantom{1} 0.000\\

\hspace{1em} & CSB & 0.000 ± 0.000 & NA & NA & 0.000 ± \vphantom{1} 0.000\\

\hspace{1em}\multirow[t]{-4}{*}{\raggedright\arraybackslash rf} & Oracle & 0.002 ± 0.000 & \textcolor{black}{0.920 ± 0.039} & 1.000 ± 0.000 & 1.000 ± 0.000\\

\hspace{1em} & Model & 0.046 ± 0.004 & \textcolor{red}{0.702 ± 0.016} & 0.003 ± 0.003 & 0.060 ± 0.038\\

\hspace{1em} & KM & 0.000 ± 0.000 & NA & NA & 0.000 ± 0.000\\

\hspace{1em} & CSB & 0.000 ± 0.000 & NA & NA & 0.000 ± 0.000\\

\hspace{1em}\multirow[t]{-4}{*}{\raggedright\arraybackslash Cox} & Oracle & 0.002 ± 0.000 & \textcolor{black}{0.920 ± 0.039} & 1.000 ± 0.000 & 1.000 ± 0.000\\
\bottomrule
\end{tabular}

\end{table}

\begin{table}[H]
\centering
\caption{Performance of different methods for screening high-risk patients with $P(T > 10) < 0.50$ in an easier synthetic data scenario (Setting 3 from Table~\ref{tab:distributions-synthetic}), using various censoring and survival models.
The training sample size is fixed at 1000. Other details are as in Figure~\ref{fig:exp-setup1-setting1}.}
\label{tab:setup4-setting3-hr}

\begin{tabular}{lccccc}
\toprule
Survival Model & Method & Screened & Survival & Precision & Recall\\
\midrule
\addlinespace[0.3em]
\multicolumn{6}{l}{\textbf{Censoring Model: grf}}\\
\hspace{1em} & Model & 0.988 ± 0.010 & \textcolor{black}{0.386 ± 0.003} & 0.953 ± 0.001 & 0.987 ± \vphantom{1} 0.010\\

\hspace{1em} & KM & 0.910 ± 0.058 & \textcolor{black}{0.352 ± 0.022} & 0.953 ± 0.001 & 0.910 ± \vphantom{7} 0.058\\

\hspace{1em} & CSB & 0.910 ± 0.051 & \textcolor{black}{0.370 ± 0.016} & 0.953 ± 0.001 & 0.910 ± 0.051\\

\hspace{1em}\multirow[t]{-4}{*}{\raggedright\arraybackslash grf} & Oracle & 0.953 ± 0.001 & \textcolor{black}{0.378 ± 0.003} & 1.000 ± 0.000 & 1.000 ± \vphantom{1} 0.000\\

\hspace{1em} & Model & 0.000 ± 0.000 & \textcolor{black}{0.000 ± 0.000} & 0.500 ± 0.141 & 0.000 ± \vphantom{1} 0.000\\

\hspace{1em} & KM & 0.910 ± 0.058 & \textcolor{black}{0.352 ± 0.022} & 0.953 ± 0.001 & 0.910 ± \vphantom{6} 0.058\\

\hspace{1em} & CSB & 0.000 ± 0.000 & \textcolor{black}{0.029 ± 0.031} & 0.972 ± 0.019 & 0.000 ± \vphantom{1} 0.000\\

\hspace{1em}\multirow[t]{-4}{*}{\raggedright\arraybackslash survreg} & Oracle & 0.953 ± 0.001 & \textcolor{black}{0.378 ± 0.003} & 1.000 ± 0.000 & 1.000 ± \vphantom{1} 0.000\\

\hspace{1em} & Model & 0.791 ± 0.025 & \textcolor{black}{0.386 ± 0.003} & 0.953 ± 0.001 & 0.791 ± \vphantom{1} 0.025\\

\hspace{1em} & KM & 0.910 ± 0.058 & \textcolor{black}{0.352 ± 0.022} & 0.953 ± 0.001 & 0.910 ± \vphantom{5} 0.058\\

\hspace{1em} & CSB & 0.771 ± 0.039 & \textcolor{black}{0.390 ± 0.005} & 0.951 ± 0.003 & 0.771 ± 0.039\\

\hspace{1em}\multirow[t]{-4}{*}{\raggedright\arraybackslash rf} & Oracle & 0.953 ± 0.001 & \textcolor{black}{0.378 ± 0.003} & 1.000 ± 0.000 & 1.000 ± \vphantom{1} 0.000\\

\hspace{1em} & Model & 0.688 ± 0.017 & \textcolor{black}{0.386 ± 0.003} & 0.953 ± 0.002 & 0.688 ± \vphantom{1} 0.017\\

\hspace{1em} & KM & 0.910 ± 0.058 & \textcolor{black}{0.352 ± 0.022} & 0.953 ± 0.001 & 0.910 ± \vphantom{4} 0.058\\

\hspace{1em} & CSB & 0.605 ± 0.034 & \textcolor{black}{0.388 ± 0.008} & 0.952 ± 0.004 & 0.605 ± 0.034\\

\hspace{1em}\multirow[t]{-4}{*}{\raggedright\arraybackslash Cox} & Oracle & 0.953 ± 0.001 & \textcolor{black}{0.378 ± 0.003} & 1.000 ± 0.000 & 1.000 ± \vphantom{1} 0.000\\

\addlinespace[0.3em]
\multicolumn{6}{l}{\textbf{Censoring Model: Cox}}\\
\hspace{1em} & Model & 0.988 ± 0.010 & \textcolor{black}{0.386 ± 0.003} & 0.953 ± 0.001 & 0.987 ± 0.010\\

\hspace{1em} & KM & 0.910 ± 0.058 & \textcolor{black}{0.352 ± 0.022} & 0.953 ± 0.001 & 0.910 ± \vphantom{3} 0.058\\

\hspace{1em} & CSB & 0.916 ± 0.049 & \textcolor{black}{0.370 ± 0.016} & 0.953 ± 0.001 & 0.916 ± 0.049\\

\hspace{1em}\multirow[t]{-4}{*}{\raggedright\arraybackslash grf} & Oracle & 0.953 ± 0.001 & \textcolor{black}{0.378 ± 0.003} & 1.000 ± 0.000 & 1.000 ± 0.000\\

\hspace{1em} & Model & 0.000 ± 0.000 & \textcolor{black}{0.000 ± 0.000} & 0.500 ± 0.141 & 0.000 ± 0.000\\

\hspace{1em} & KM & 0.910 ± 0.058 & \textcolor{black}{0.352 ± 0.022} & 0.953 ± 0.001 & 0.910 ± \vphantom{2} 0.058\\

\hspace{1em} & CSB & 0.000 ± 0.000 & \textcolor{black}{0.029 ± 0.031} & 0.972 ± 0.019 & 0.000 ± 0.000\\

\hspace{1em}\multirow[t]{-4}{*}{\raggedright\arraybackslash survreg} & Oracle & 0.953 ± 0.001 & \textcolor{black}{0.378 ± 0.003} & 1.000 ± 0.000 & 1.000 ± 0.000\\

\hspace{1em} & Model & 0.791 ± 0.025 & \textcolor{black}{0.386 ± 0.003} & 0.953 ± 0.001 & 0.791 ± 0.025\\

\hspace{1em} & KM & 0.910 ± 0.058 & \textcolor{black}{0.352 ± 0.022} & 0.953 ± 0.001 & 0.910 ± \vphantom{1} 0.058\\

\hspace{1em} & CSB & 0.762 ± 0.043 & \textcolor{black}{0.385 ± 0.009} & 0.953 ± 0.002 & 0.761 ± 0.042\\

\hspace{1em}\multirow[t]{-4}{*}{\raggedright\arraybackslash rf} & Oracle & 0.953 ± 0.001 & \textcolor{black}{0.378 ± 0.003} & 1.000 ± 0.000 & 1.000 ± 0.000\\

\hspace{1em} & Model & 0.688 ± 0.017 & \textcolor{black}{0.386 ± 0.003} & 0.953 ± 0.002 & 0.688 ± 0.017\\

\hspace{1em} & KM & 0.910 ± 0.058 & \textcolor{black}{0.352 ± 0.022} & 0.953 ± 0.001 & 0.910 ± 0.058\\

\hspace{1em} & CSB & 0.625 ± 0.029 & \textcolor{black}{0.388 ± 0.008} & 0.952 ± 0.005 & 0.625 ± 0.029\\

\hspace{1em}\multirow[t]{-4}{*}{\raggedright\arraybackslash Cox} & Oracle & 0.953 ± 0.001 & \textcolor{black}{0.378 ± 0.003} & 1.000 ± 0.000 & 1.000 ± 0.000\\
\bottomrule
\end{tabular}

\end{table}

\FloatBarrier
\subsection{Effect of Distribution Shift}

\begin{table}[H]
\centering
\caption{
Performance of different methods for screening high-risk patients with $P(T > 3) < 0.50$ under distribution shift (Setting 4 from Table~\ref{tab:distributions-synthetic}).
These results correspond to the same experiments reported in Table~\ref{tab:setup2-lr}, but evaluate high-risk screening instead of low-risk.
In this case, a misspecified \textit{grf} model primarily reduces the power of the screening rule by selecting fewer patients than desired.
}
\label{tab:setup2-hr}

\begin{tabular}{lcccc}
\toprule
Method & Screened & Survival & Precision & Recall\\
\midrule
\addlinespace[0.3em]
\multicolumn{5}{l}{\textbf{Low-Quality Model}}\\
\hspace{1em}Model & 0.000 ± 0.000 & NA & NA & 0.000 ± 0.000\\

\hspace{1em}KM & 0.030 ± 0.034 & \textcolor{black}{0.016 ± 0.019} & 0.495 ± 0.002 & 0.030 ± 0.034\\

\hspace{1em}CSB & 0.000 ± 0.000 & NA & NA & 0.000 ± 0.000\\

\hspace{1em}Oracle & 0.501 ± 0.003 & \textcolor{black}{0.077 ± 0.002} & 1.000 ± 0.000 & 1.000 ± 0.000\\

\addlinespace[0.3em]
\multicolumn{5}{l}{\textbf{High-Quality Model}}\\
\hspace{1em}Model & 0.503 ± 0.004 & \textcolor{black}{0.080 ± 0.003} & 0.998 ± 0.000 & 1.000 ± 0.000\\

\hspace{1em}KM & 0.026 ± 0.036 & \textcolor{black}{0.014 ± 0.020} & 0.502 ± 0.001 & 0.026 ± 0.036\\

\hspace{1em}CSB & 0.503 ± 0.004 & \textcolor{black}{0.080 ± 0.003} & 0.998 ± 0.000 & 0.999 ± 0.002\\

\hspace{1em}Oracle & 0.503 ± 0.004 & \textcolor{black}{0.078 ± 0.003} & 1.000 ± 0.000 & 1.000 ± 0.000\\
\bottomrule
\end{tabular}

\end{table}

\begin{figure}[H]
    \centering
    \includegraphics[width=\textwidth]{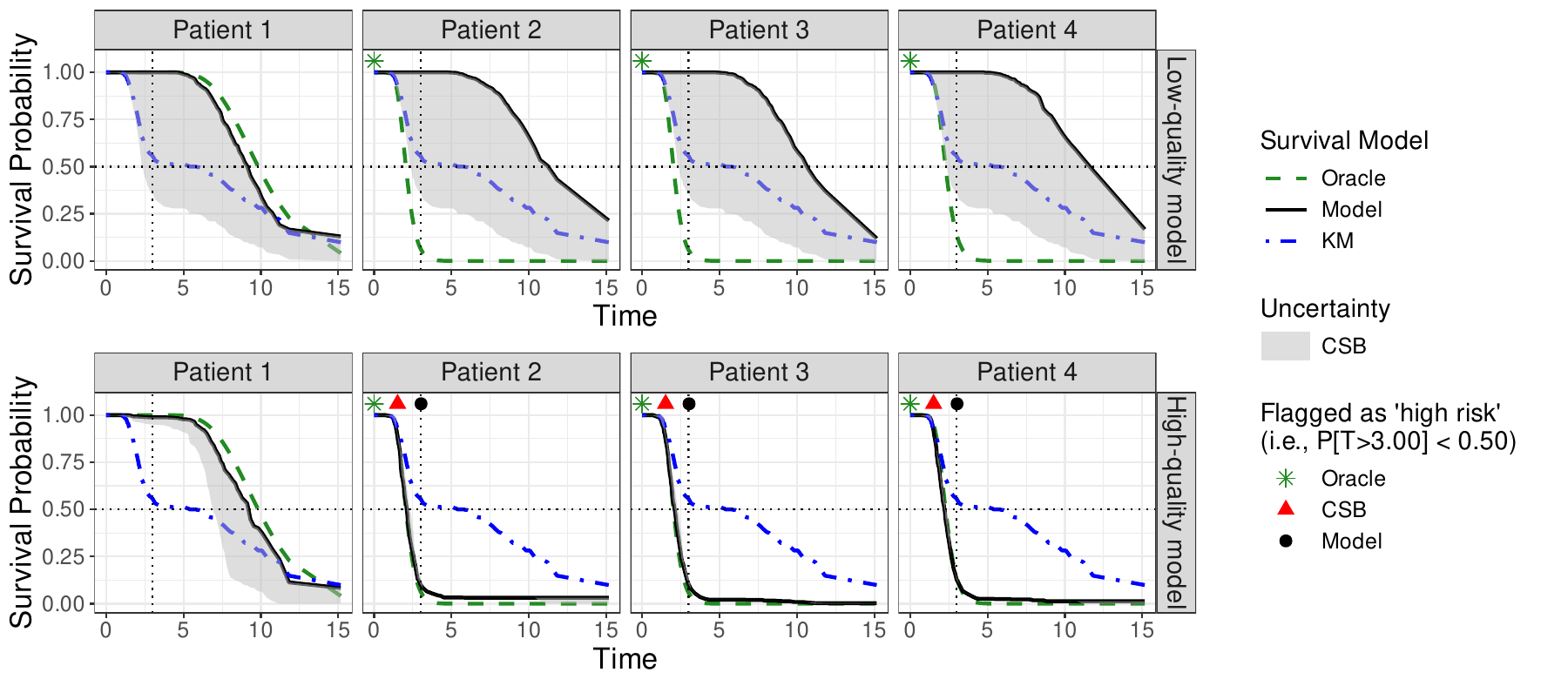}
    \caption{
Illustration of the use of conformal survival bands (shaded regions) for screening test patients in a simulated censored dataset under distribution shift, corresponding to the high-risk screening experiments discussed in Table~\ref{tab:setup2-hr}.
Solid black curves show survival estimates from either an inaccurate (top) or accurate (bottom) survival forest model, while dashed green curves represent the true survival probabilities.
The goal is to identify high-risk patients—those with less than 50\% probability (horizontal dotted line) of surviving beyond 3 years (vertical dotted line).
Other details are as in Figure~\ref{fig:calibration_band_example}.
}
    \label{fig:calibration_band_example-hr}
\end{figure}

\FloatBarrier
\section{Additional Details on Real Data Applications} \label{app:data}

We apply our method to seven datasets previously utilized by \citet{sesia2024doubly}: the Colon Cancer Chemotherapy (COLON) dataset; the German Breast Cancer Study Group (GBSG) dataset; the Stanford Heart Transplant Study (HEART); the Molecular Taxonomy of Breast Cancer International Consortium (METABRIC) dataset; the Primary Biliary Cirrhosis (PBC) dataset; the Diabetic Retinopathy Study (RETINOPATHY); and the Veterans' Administration Lung Cancer Trial (VALCT).
Table~\ref{tab:datasets_summary} provides details on the number of observations, covariates, and data sources.

The datasets were obtained from various publicly available sources. COLON, HEART, PBC, RETINOPATHY, and VALCT are included in the \texttt{survival} R package.
GBSG was sourced from GitHub: \url{https://github.com/jaredleekatzman/DeepSurv/}.
METABRIC was accessed via \url{https://www.cbioportal.org/study/summary?id=brca_metabric}.

Each dataset underwent a pre-processing pipeline to ensure consistency and prepare the data for analysis, as in \citet{sesia2024doubly}. Survival times equal to zero were replaced with half the smallest observed non-zero time. Missing values were imputed using the median for numeric variables and the mode for categorical variables. Factor variables were processed to merge rare levels (frequency below 2\%) into an ``Other'' category, while binary factors with one rare level were removed entirely. Dummy variables were created for all factors, and redundant features were identified and removed using an alias check. Additionally, highly correlated features (correlation above 0.75) were iteratively filtered.

\begin{table}[ht]
\caption{Summary of the publicly available survival analysis datasets used in Section~\ref{sec:application}. }
\label{tab:datasets_summary}
\centering
\begin{tabular}{@{}lcccc@{}}
\toprule
\textbf{Dataset} & \textbf{Obs.} & \textbf{Vars.} & \textbf{Source} & \textbf{Citation} \\ \midrule
COLON                  & 1858                        & 11                       & \texttt{survival}    & \cite{moertel1990levamisole} \\
GBSG                   & 2232                        & 6                        & github.com   & \cite{katzman2018deepsurv} \\
HEART                  & 172                         & 4                        & \texttt{survival}    & \cite{crowley1977covariance} \\
METABRIC               & 1981                        & 41                       & cbioportal.org        & \cite{curtis2012genomic} \\
PBC                    & 418                         & 17                       & \texttt{survival}    & \cite{therneau2000cox} \\
RETINOPATHY            & 394                         & 5                        & \texttt{survival}    & \cite{blair19805} \\ 
VALCT                  & 137                         & 6                        & \texttt{survival}    & \cite{kalbfleisch2002statistical} \\\bottomrule
\end{tabular}
\end{table}

\begin{table}[H]
\centering
\caption{Detailed screening results for low-risk selection using the \textit{grf} model, with threshold rule $P(T > t_1) > 0.80$ and $t_1$ set to the 0.1 quantile of observed times in each dataset. Shown are the screened proportion and bounds on the survival rate among selected patients, aggregated over 100 repetitions.}
\label{tab:data-detailed-lr}

\begin{tabular}{lccc}
\toprule
Method & Screened & Survival (lower bound) & Survival (upper bound)\\
\midrule
\addlinespace[0.3em]
\multicolumn{4}{l}{\textbf{COLON}}\\
\hspace{1em}Model & 0.921 ± 0.006 & \textcolor{darkgreen}{0.917 ± 0.004} & \textcolor{darkgreen}{0.920 ± 0.004}\\

\hspace{1em}KM & 1.000 ± 0.000 & \textcolor{darkgreen}{0.902 ± 0.004} & \textcolor{darkgreen}{0.905 ± 0.004}\\

\hspace{1em}CSB & 0.647 ± 0.121 & \textcolor{darkgreen}{0.939 ± 0.012} & \textcolor{darkgreen}{0.941 ± 0.012}\\

\addlinespace[0.3em]
\multicolumn{4}{l}{\textbf{GBSG}}\\
\hspace{1em}Model & 0.897 ± 0.006 & \textcolor{darkgreen}{0.921 ± 0.003} & \textcolor{darkgreen}{0.931 ± 0.003}\\

\hspace{1em}KM & 1.000 ± 0.000 & \textcolor{darkgreen}{0.901 ± 0.003} & \textcolor{darkgreen}{0.912 ± 0.003}\\

\hspace{1em}CSB & 0.905 ± 0.007 & \textcolor{darkgreen}{0.921 ± 0.003} & \textcolor{darkgreen}{0.930 ± 0.003}\\

\addlinespace[0.3em]
\multicolumn{4}{l}{\textbf{HEART}}\\
\hspace{1em}Model & 1.000 ± 0.000 & \textcolor{darkgreen}{0.930 ± 0.011} & \textcolor{darkgreen}{0.965 ± 0.009}\\

\hspace{1em}KM & 1.000 ± 0.000 & \textcolor{darkgreen}{0.930 ± 0.011} & \textcolor{darkgreen}{0.965 ± 0.009}\\

\hspace{1em}CSB & 0.919 ± 0.077 & \textcolor{darkgreen}{0.934 ± 0.012} & \textcolor{darkgreen}{0.965 ± 0.009}\\

\addlinespace[0.3em]
\multicolumn{4}{l}{\textbf{METABRIC}}\\
\hspace{1em}Model & 0.964 ± 0.004 & \textcolor{darkgreen}{0.911 ± 0.003} & \textcolor{darkgreen}{0.927 ± 0.004}\\

\hspace{1em}KM & 1.000 ± 0.000 & \textcolor{darkgreen}{0.905 ± 0.003} & \textcolor{darkgreen}{0.921 ± 0.003}\\

\hspace{1em}CSB & 0.970 ± 0.005 & \textcolor{darkgreen}{0.910 ± 0.004} & \textcolor{darkgreen}{0.926 ± 0.004}\\

\addlinespace[0.3em]
\multicolumn{4}{l}{\textbf{PBC}}\\
\hspace{1em}Model & 0.857 ± 0.012 & \textcolor{darkgreen}{0.961 ± 0.007} & \textcolor{darkgreen}{0.961 ± 0.007}\\

\hspace{1em}KM & 1.000 ± 0.000 & \textcolor{darkgreen}{0.904 ± 0.010} & \textcolor{darkgreen}{0.904 ± 0.010}\\

\hspace{1em}CSB & 0.809 ± 0.061 & \textcolor{darkgreen}{0.956 ± 0.009} & \textcolor{darkgreen}{0.956 ± 0.009}\\

\addlinespace[0.3em]
\multicolumn{4}{l}{\textbf{RETINOPATHY}}\\
\hspace{1em}Model & 0.982 ± 0.010 & \textcolor{darkgreen}{0.884 ± 0.008} & \textcolor{darkgreen}{0.903 ± 0.008}\\

\hspace{1em}KM & 1.000 ± 0.000 & \textcolor{darkgreen}{0.884 ± 0.008} & \textcolor{darkgreen}{0.904 ± 0.008}\\

\hspace{1em}CSB & 0.515 ± 0.137 & \textcolor{darkgreen}{0.934 ± 0.019} & \textcolor{darkgreen}{0.945 ± 0.016}\\

\addlinespace[0.3em]
\multicolumn{4}{l}{\textbf{VALCT}}\\
\hspace{1em}Model & 0.851 ± 0.028 & \textcolor{darkgreen}{0.914 ± 0.018} & \textcolor{darkgreen}{0.914 ± 0.018}\\

\hspace{1em}KM & 0.980 ± 0.040 & \textcolor{darkgreen}{0.905 ± 0.015} & \textcolor{darkgreen}{0.905 ± 0.015}\\

\hspace{1em}CSB & 0.805 ± 0.068 & \textcolor{darkgreen}{0.928 ± 0.019} & \textcolor{darkgreen}{0.928 ± 0.019}\\
\bottomrule
\end{tabular}

\end{table}

\begin{table}[H]
\centering
\caption{Detailed screening results for high-risk selection using the \textit{Cox} model, with threshold rule $P(T > t_1) < 0.80$ and $t_1$ set to the 0.1 quantile of observed times in each dataset. Shown are the screened proportion and bounds on the survival rate among selected patients, aggregated over 100 repetitions. In some cases, no patients are selected. }
\label{tab:data-detailed-hr}

\begin{tabular}{lccc}
\toprule
Method & Screened & Survival (lower bound) & Survival (upper bound)\\
\midrule
\addlinespace[0.3em]
\multicolumn{4}{l}{\textbf{COLON}}\\
\hspace{1em}Model & 0.069 ± 0.005 & \textcolor{darkgreen}{0.747 ± 0.030} & \textcolor{darkgreen}{0.747 ± 0.030}\\

\hspace{1em}KM & 0.000 ± 0.000 & NA & \vphantom{5} NA\\

\hspace{1em}CSB & 0.014 ± 0.008 & \textcolor{darkgreen}{0.166 ± 0.090} & \textcolor{darkgreen}{0.166 ± 0.090}\\

\addlinespace[0.3em]
\multicolumn{4}{l}{\textbf{GBSG}}\\
\hspace{1em}Model & 0.028 ± 0.002 & \textcolor{red}{0.849 ± 0.027} & \textcolor{red}{0.861 ± 0.024}\\

\hspace{1em}KM & 0.000 ± 0.000 & NA & \vphantom{4} NA\\
\hspace{1em}CSB & 0.000 ± 0.000 & NA & \vphantom{2} NA\\

\addlinespace[0.3em]
\multicolumn{4}{l}{\textbf{HEART}}\\
\hspace{1em}Model & 0.000 ± 0.000 & NA & NA\\

\hspace{1em}KM & 0.000 ± 0.000 & NA & \vphantom{3} NA\\
\hspace{1em}CSB & 0.000 ± 0.000 & NA & \vphantom{1} NA\\
\addlinespace[0.3em]
\multicolumn{4}{l}{\textbf{METABRIC}}\\
\hspace{1em}Model & 0.036 ± 0.003 & \textcolor{darkgreen}{0.749 ± 0.031} & \textcolor{darkgreen}{0.752 ± 0.031}\\

\hspace{1em}KM & 0.000 ± 0.000 & NA & \vphantom{2} NA\\
\hspace{1em}CSB & 0.000 ± 0.000 & NA & NA\\
\addlinespace[0.3em]
\multicolumn{4}{l}{\textbf{PBC}}\\
\hspace{1em}Model & 0.125 ± 0.013 & \textcolor{darkgreen}{0.597 ± 0.054} & \textcolor{darkgreen}{0.597 ± 0.054}\\

\hspace{1em}KM & 0.000 ± 0.000 & NA & \vphantom{1} NA\\
\hspace{1em}CSB & 0.023 ± 0.015 & \textcolor{darkgreen}{0.140 ± 0.085} & \textcolor{darkgreen}{0.140 ± 0.085}\\

\addlinespace[0.3em]
\multicolumn{4}{l}{\textbf{RETINOPATHY}}\\
\hspace{1em}Model & 0.012 ± 0.006 & \textcolor{darkgreen}{0.256 ± 0.115} & \textcolor{darkgreen}{0.266 ± 0.119}\\

\hspace{1em}KM & 0.000 ± 0.000 & NA & NA\\
\hspace{1em}CSB & 0.002 ± 0.002 & \textcolor{darkgreen}{0.075 ± 0.073} & \textcolor{darkgreen}{0.075 ± 0.073}\\

\addlinespace[0.3em]
\multicolumn{4}{l}{\textbf{VALCT}}\\
\hspace{1em}Model & 0.139 ± 0.018 & \textcolor{darkgreen}{0.698 ± 0.077} & \textcolor{darkgreen}{0.698 ± 0.077}\\

\hspace{1em}KM & 0.020 ± 0.040 & \textcolor{darkgreen}{0.018 ± 0.036} & \textcolor{darkgreen}{0.018 ± 0.036}\\

\hspace{1em}CSB & 0.001 ± 0.003 & \textcolor{darkgreen}{0.010 ± 0.020} & \textcolor{darkgreen}{0.010 ± 0.020}\\
\bottomrule
\end{tabular}

\end{table}

\end{document}